%% file: repeatability_paper.tex
\documentclass{article}
\usepackage{arxiv}

\usepackage[utf8]{inputenc} 
\usepackage[T1]{fontenc}    
\usepackage{hyperref}       
\usepackage{url}            
\usepackage{booktabs}       
\usepackage{amsfonts}       
\usepackage{nicefrac}       
\usepackage{microtype}      
\usepackage{lipsum}		
\usepackage{graphicx}
\usepackage{natbib}
\usepackage{doi}
\usepackage{fancyvrb}
\usepackage{framed}

\usepackage{amssymb}
\usepackage{amsmath, amsfonts}
\usepackage{amsthm}
\usepackage{authblk}
\usepackage{latexsym}
\usepackage{graphicx}
\usepackage{nomencl}
\usepackage{bm}
\usepackage{textcomp}
\usepackage{algorithm, algpseudocode}
\usepackage{url}
\usepackage{xcolor}
\usepackage{hyperref}
\usepackage[acronym]{glossaries}
\usepackage{tikz}
\usetikzlibrary{bayesnet, positioning, calc}
\usepackage{adjustbox}
\usepackage{tabularx}
\usepackage{algorithmicx}
\usepackage{subcaption}
\usepackage{fancyhdr}

\newtheorem{theorem}{Theorem}[section]
\newtheorem{corollary}{Corollary}[theorem]

\newacronym{mcmc}{MCMC}{Markov-Chain Monte-Carlo}
\newacronym{psa}{PSA}{Prostate Specific Antigen}
\newacronym{icc}{$ICC$}{Intraclass Correlation Coefficient}
\newacronym{ims}{$IMS$}{Inter Measurement Sensitivity}
\newacronym{roi}{ROI}{Region of Interest}
\newacronym{rois}{ROIs}{Regions of Interest}
\newacronym{loa}{LoA}{Limits of Agreement}
\newacronym{rc}{$RC$}{Repeatability Coefficient}
\newacronym{cov}{$CoV$}{Coefficient of Variation}
\newacronym{adc}{ADC}{Apparent Diffusion Coefficient}
\newacronym{mip}{MIP}{Maximum Intensity Projection}
\newacronym{mcrpc}{mCRPC}{metastatic castrate-resistant prostate cancer}
\newacronym{epi}{EPI}{echo-planar imaging}
\newacronym{iqr}{IQR}{inter-quartile range}
\newacronym{wbdwi}{WBDWI}{whole-body diffusion-weighted MRI}
\newacronym{ai}{AI}{artificial intelligence}
\newacronym{hmc}{HMC}{Hamiltonean Monte Carlo}
\newacronym{ba}{BA}{Bland-Altman}
\newacronym{po}{$PO_{10}$}{Posterior Odds}
\newacronym{bf}{$BF_{10}$}{Bayes Factors}
\newacronym{dm}{DM}{Dirichlet-Multinomial}

\DeclareMathOperator*{\argmax}{arg\,max}

\title{Generalizing imaging biomarker repeatability studies using Bayesian inference: Applications in
    detecting heterogeneous treatment response in whole-body diffusion-weighted MRI of metastatic prostate cancer}

\author[1,*]{Matthew D~Blackledge}
\author[2,3]{Konstantinos~Zormpas-Petridis}
\author[1,4]{Ricardo~Donners}
\author[1]{Antonio~Candito}
\author[1,5]{David J~Collins}
\author[1,5]{Johann de~Bono}
\author[1,5]{Chris~Parker}
\author[1,5]{Dow-Mu~Koh}
\author[1,5]{Nina~Tunariu}

\affil[1]{The Institute of Cancer Research, London, United Kingdom}
\affil[2]{Fondazione Policlinico Universitario Agostino Gemelli IRCCS, Rome, Italy}
\affil[3]{Università Cattolica del Sacro Cuore, Rome, Italy}
\affil[4]{Universitätsspital Basel, Basel, Switzerland}
\affil[5]{The Royal Marsden NHS Foundation Trust, London, United Kingdom \vspace{6mm}}
\affil[*]{Correspondence: \texttt{matthew.blackledge@icr.ac.uk}; 15 Cotswold Road, Sutton, SM2 5NG, United Kingdom}

\date{}

\input{pygments_def.tex}

\begin{document}
\maketitle

\begin{abstract}
\input{sections/abstract}
\end{abstract}


\section{Introduction}\label{sec:introduction}
\input{sections/introduction}

\section{Theory}\label{sec:theory}
\input{sections/theory}

\section{Materials and Methods}\label{sec:methods}
\input{sections/methods}

\section{Results}\label{sec:results}
\input{sections/results}

\section{Discussion}\label{sec:discussion}
\input{sections/discussion}

\section{Acknowledgements}\label{sec:acknowledgements}
This study represents independent research supported by the National Institute for Health and Care Research (NIHR)
Biomedical Research Centre and the Clinical Research Facility in Imaging at The Royal Marsden NHS Foundation Trust and
the Institute of Cancer Research, London.
The views expressed are those of the author(s) and not necessarily those of the NIHR or the Department of Health and
Social Care.

\section{Declaration of generative AI and AI-assisted technologies in the writing process}
During the preparation of this work the authors used ChatGPT to improve grammar within the paper only.
After using this service, the authors reviewed and edited the content as needed and take full responsibility
for the content of the publication.

\bibliographystyle{unsrtnat}
\bibliography{repeatability_paper}

\newpage
\section*{Supplementary Material}
\input{sections/supplementary-material}

\end{document}

%% file: sections/abstract.tex
The assessment of imaging biomarkers is critical for advancing precision medicine and improving disease
characterization.
Despite the availability of methods to derive disease heterogeneity metrics in imaging studies, a robust framework for
evaluating measurement uncertainty remains underdeveloped.
To address this gap, we propose a novel Bayesian framework to assess the precision of disease heterogeneity measures in
biomarker studies.

Our approach extends traditional methods for evaluating biomarker precision by providing greater flexibility in
statistical assumptions and enabling the analysis of biomarkers beyond univariate or multivariate normally-distributed
variables.
Using Hamiltonian Monte Carlo (HMC) sampling, the framework supports both, for example, normally-distributed and
Dirichlet-Multinomial distributed variables, enabling the derivation of posterior distributions for biomarker parameters
under diverse model assumptions.
Designed to be broadly applicable across various imaging modalities and biomarker types, the framework builds a
foundation for generalizing reproducible and objective biomarker evaluation.

To demonstrate its utility, we apply the framework to whole-body diffusion-weighted MRI (WBDWI) to assess heterogeneous
therapeutic responses in metastatic bone disease.
Specifically, we analyze data from two patient studies investigating treatments for metastatic castrate-resistant
prostate cancer (mCRPC).
Mixed therapeutic responses, recognized as a sign of disease heterogeneity with potential adverse clinical implications,
are part of the Prostate Cancer Working Group 3.0 criteria.
Our results reveal an approximately 70\% response rate among individual tumors across both studies, objectively
characterizing differential responses to systemic therapies and validating the clinical relevance of the proposed
methodology.

This Bayesian framework provides a powerful tool for advancing biomarker research across diverse imaging-based studies
while offering valuable insights into specific clinical applications, such as mCRPC treatment response.

%% file: sections/introduction.tex
There is a clear trend that medical imaging is become more quantitative; images are not only being used for
qualitative radiological evaluation, but more frequently to derive \emph{imaging biomarkers} (IBs) that act as
non-invasive surrogate measurements of biological and/or pathogenic characteristics of the imaged
tissue~\cite{Gillies:2016wk, OConner_IB}.
An example of a widley used biomarker in oncology is tumour size measurement defined by RECIST, where the
sum of the longest diameters in up to five target lesions (maximum two per organ) is recorded and monitored throughout
the patient treatment pathway~\cite{Eisenhauer2009}.
As imaging technologies advance and acquisition times accelerate, this has facilitated increased
data collection to derive multiple IBs in a single patient study.
For example, magnetic resonance imaging (MRI) offers measurement of tumour cellularity with diffusion-weighted imaging
(DWI)~\cite{Charles-Edwards:2006tb, koh_cancer, Surov:2017vk}, tissue stiffness using elastography
(MRE)~\cite{Manduca:2021wp,sack_mre}, and tumour vasculature through dynamic contrast-enhanced imaging
(DCE)~\cite{Leach:2012tt,DCE-review}.

Whilst this presents practical challenges to radiologists needing to evaluate more images per patient study, it also
offers multi-faceted insight into how the tumour microenvironment responds to therapy~\cite{Hoffmann:2024ws}.
Cancer is an intrinsically complex, heterogeneous and evolving disease~\cite{McGranahan:2015vr}.
Within such an ecosystem, molecularly and morphologically-distinct tumour regions (“habitats”) interact and adapt
differently to selection pressures (including therapy) to gain evolutionary fitness
advantages~\cite{Maley:2017wa, Amend:2015vn}.
These changes may differ between spatially isolated sites of disease (inter-lesion heterogeneity) or indeed within a
single lesion (intra-lesion heterogeneity).
Recent techniques have been developed to combine features from multiple quantitative imaging biomarkers and identify
these spatially distinct regions, while validating these features using ex-vivo histopathological
assessment~\cite{Sala:2017vi, Crispin-Ortuzar:2020us, Beer:2021uo, Zormpas-Petridis:2020ti}.
A key advantage of such habitat mapping is that it can cover the entire tumour volume and non-invasively monitor changes
in the habitats across treatment.
This is in stark contrast to tissue biopsies, which although offer great biological detail, cannot sample every location
within the tumour and is therefore unable to reflect the full range of tissue characteristics before and after treatment
~\cite{Panico:2023wd}.
The habitats derived from such techniques can be viewed as representing new forms of imaging biomarkers that equally
require validation if they are to be used in the clinic.

An important question remains as to how to best treat quantification of habitats and their change over treatment within
the statistical framework of imaging biomarkers.
For example, a vital characteristic of any IB intended to monitor cancer response is that it be robust and
precise~\cite{obuchowski_intero_change,OBUCHOWSKI2022543}.
If the IB measurement is found to be highly sensitive to imaging/patient set-up, statistical image noise, or operator
variability, then it might be that the IB is not sensitive or robust enough to detect subtle changes occurring
within tumours.
The importance of IB precision has been highlighted in the Cancer Research UK (CRUK) and the European Organisation for
Research and Treatment of Cancer (EORTC) IB roadmap~\cite{OConner_IB}, where evaluation of measurement precision is a
core requirement at every step of technical validation, from IB discovery to clinical implementation.
The current gold-standard for measuring precision of IBs are dedicated double baseline (repeatability)
studies~\cite{Winfield:2017wq, Winfield:2019vy}.
Patients are asked to undergo the same scan twice with an identical patient set-up performed each time, preferably
prior to any medical intervention, to ensure that any detected changes characterize the precision of the imaging device.
Standard statistical methodologies such as Bland-Altman analysis are then used to characterize measurement precision
and subsequently applied in patient investigations to monitor tumour change following treatment once the precision of
the measurement is known.
A limitation of current approaches to assessing repeatability is that any measured IB, $x$, is typically assumed to be
a real, univariate or multivariate variable ($x \in \mathbb{R}$ or $x \in \mathbb{R}^{D}$).
Although this may be justifiable in many cases, it limits the scope of IBs from imaging studies and does not
directly apply to monitoring changes in intra-tumoral heterogeneity and habitat maps.

Here we propose a Bayesian paradigm to determine general precision in biomarker studies.
In particular, this methodology broadens the scope of IB precision studies to include biomarkers that could be another
category of variable.
For example, in the case of habitat imaging we hypothesize that IBs are better described by the unit simplex.

In the first section of our paper, we present the theoretical foundation of our approach, structured around four key
innovations:
\begin{enumerate}
    \item {\bf{Equivalence to Conventional Methods}}.
    We establish, through theoretical arguments, that this Bayesian approach is equivalent to conventional methods such
    as \acrfull{ba} analysis~\cite{MartinBland1986} under reasonable assumptions.
    Importantly, we demonstrate the mathematical equivalence between \acrshort{ba} analysis and the broader concept of
    novelty detection when applied to real-valued biomarkers.
    \item {\bf{Bayesian Novelty Detection Framework}}.
    We propose a general framework for Bayesian novelty detection applicable to arbitrary response biomarkers.
    This approach leverages statistical modeling of posterior predictive samples under the null hypothesis of no
    biomarker change.
    \item {\bf{Bayesian Metrics for Change Detection}}.
    As an alternative to p-value or novelty estimation, we explore the use of \acrfull{bf} and a closely associated
    equivalent, \acrfull{po}, to assess the relative likelihood of biomarker change.
    Determination of \acrshort{po} is achieved through Bayesian mixture modeling, which simultaneously derives posterior
    distributions of the rate of inter-tumoral response heterogeneity, both at the population and patient levels.
    \item {\bf{Modeling Intra-Tumoral Heterogeneity}}.
    Finally, we introduce our mathematical model for quantifying intra-tumoral response heterogeneity.
    The precision of this biomarker is modeled using a \acrfull{dm} distribution, providing robust statistical insights
    into intra-tumoral response variability.
\end{enumerate}

To validate the robustness of \acrfull{hmc} in our proposed framework and models, we first conduct comprehensive
simulation studies.
These simulations model pre- and post-treatment biomarker changes, utilizing both conventional real-valued
biomarkers and \acrshort{dm} variates representing intra-tumoral heterogeneity.
The simulations are performed across a range of plausible biomarker values to ensure their applicability and
reliability under realistic conditions.

Finally, we provide a proof-of-concept demonstration of these techniques for assessing heterogeneous responses in
\acrfull{mcrpc}.
Specifically, we use estimates of \acrfull{adc} from \acrfull{wbdwi} to compare \acrshort{dm}-based biomarkers, which
capture intra-tumoral heterogeneity, with conventional median ADC estimates.
The documentation of mixed imaging responses to treatment has been highlighted as a key recommendation by the Prostate
Cancer Working Group~\cite{Scher2016}.
This reflects the growing recognition of disease heterogeneity in \acrshort{mcrpc}, as well as the coexistence of mixed
responses to therapy and overall clinical benefit.
\acrshort{wbdwi} is emerging as a core imaging biomarker for evaluating therapeutic responses in bone metastases, a
critical area where accurate imaging techniques for assessing therapeutic benefit remain
limited~\cite{Blackledge2014, Padhani2017}.
However, the evaluation of heterogeneous changes in these imaging studies has been challenging to date.
Our proposed techniques aim to address this gap by enabling a more nuanced assessment of intra-tumoral response
heterogeneity.

%% file: sections/theory.tex
\subsection{Introduction to biomarker measurement repeatability}\label{subsec:blandaltman}
In this article we use the notation $y_{itk}$ to represent the $k$-th biomarker \emph{measurement}, as measured by some
medical device, for subject $i\in\{1, 2, \dots, N\}$ at time $t\in \{0, 1, \dots, T\}$.
We use $t=0$ to indicate that the measurement made before treatment administration (baseline).
To simplify notation, if index $k$ is dropped, this indicates that only a single measurement is available
($y_{it} \equiv y_{it1}$).
Furthermore, if index $i$ is dropped, this indicates a single patient observation ($y_{t} \equiv y_{it}$).
The measurement is assumed to be an estimate of the desired true biomarker of interest $x_{it}$, corrupted by bias
$\beta$ and stochastic measurement error $\sigma$.
The measurement error $\sigma$ may be due to a number of factors, including inherent noise in the measurement system,
inter- and intra-operator variability (e.g.\ defining the target \acrfull{roi} in imaging studies), or patient movement
during a scan.
Typically, all biomarkers are considered to be single, real-values numbers, $y \in \mathbb{R}$, such that:
\[
y_{itk} = \beta_{0} + \beta_{1}x_{it} + \varepsilon_{itk}
\]
where $\varepsilon_{itk}$ is sampled from a zero-mean normal distribution with standard deviation
$\sigma \in \mathbb{R}^{+}$, and $\beta_{0}$ and $\beta_{1}$ represent fixed and proportional bias terms
respectively~\cite{obuchowski_intero_change}.
Finally, we use $\alpha_{t}$ to represent the parameters that describe the population distribution of true biomarker
values $x_{t}$ at time $t$.

A common task in imaging biomarker research is to quantify measurement error, $\sigma$, so that it may be accounted for
when using biomarkers for patient diagnosis, treatment response evaluation, and population screening~\cite{OConner_IB}.
The most popular approach for characterising measurement error is that of Bland and Altman~\cite{MartinBland1986}, which
typically relies on access to two or more measurements of the same subject before treatment.
Here we use index $j$ to indicate those patients who have received a double baseline measurement:
$\mathbf{y}_{j0} = (y_{j01}, y_{j02})$ for subject $j \in \{1, 2, \dots, N_{b}\}$.
An estimate of measurement error is made using the method of least-squares:
\begin{equation} \label{eq:sigma_lsq}
\hat{\sigma} = \sqrt{\frac{1}{2N_{b}}\sum\limits_{j=1}^{N_{b}}\left(y_{j02} - y_{j01}\right)^{2}}
\end{equation}
which makes assumptions that the expected difference is zero, $\mathbb{E}\left\{Y_{02} - Y_{01}\right\} = 0$, and that
the measurement error is homoskedastic (independent of the biomarker value).
It is often useful to verify these assumptions through the Bland-Altman plot, which plots values of
$\left(y_{j02} - y_{j01}\right)$ against $\left(y_{j01} + y_{j02}\right)/2$.
In certain instances it can be advantageous to use the logarithm of the biomarker to improve the assumption of
normality and homoscedasticity~\cite{EUSER2008978}.

A common use of response biomarkers is to determine whether a treatment causes a significant change in the biomarker
value.
This is typically assessed by evaluating whether the change in the biomarker measurement at a later time point
$d_{i} = y_{i1} - y_{i0}$ is sufficiently large to be considered unlikely to arise from measurement error.
The \acrfull{rc}, defined as $RC = 1.96\sqrt{2}\hat{\sigma}$, provides a symmetric threshold.
If $|d_{i}| > RC$, this indicates a significant change in the biomarker, allowing us to reject the null hypothesis
($H_{0}: x_{i0} = x_{i1}$) at a significance level of 5\%.
Additional summary statistics that are reported in repeatability studies include the \acrfull{icc} and \acrfull{cov}
defined as
\begin{equation}\label{eq:icc}
ICC = \frac{\sigma^{2}_{0}}{\sigma^{2}_{0} + \sigma^{2}} \qquad CoV = \frac{\sigma}{\mu_{0}} \times 100\%
\end{equation}
where $\mu_{0}$ and $\sigma_{0}$ are parameters that represent the population mean and standard deviation of imaging
biomarkers at baseline respectively~\cite{MartinBland1986, obuchowski_intero_change}.

A key assumption relevant to the above discussion is that the biomarkers are real-valued and normally distributed,
which is somewhat restrictive as biomarkers may not fit this paradigm.
For example, with an increase in the efficiency of MR acquisition, many studies now acquire multiple imaging biomarkers
in a single sitting, such that they become multidimensional
($y_{it} \in \mathbb{R}^{D}$ for $D>1$)~\cite{Obuchowski:2023vj}.
Furthermore, biomarkers such as fat fraction and \acrshort{adc} are known to be limited to certain
value ranges ($x_{it} \in [0,1]$ for fat fraction and $x_{it} \in [0,3\times 10^{-3}\text{ mm}^{2}/\text{s}]$) that may
bring the assumption of normally distributed variables into question.
Here, we provide an alternative Bayesian framework that might be used to overcome these limitations for use in imaging
biomarker research.

\subsection{Bayesian interpretation of imaging biomarker changes}\label{subsec:bayesian}
As illustrated in Figure~\ref{fig:models}, we consider the problem of detecting a significant change in a given imaging
biomarker after treatment as a decision between two competing models for the acquired data: Model $M_{0}$ assumes that
the true biomarker value at times $t = 0$ and $t = 1$ is the same ($M_{0}: x_{i1} = x_{i0}$), whilst for Model $M_{1}$
this is not the case ($M_{1}: x_{i1} \neq x_{i0}$).
Both models share parameters $\alpha_{0}$, which describe the distribution of biomarker values at baseline for the
population, measurement bias, $\beta$, and measurement error, $\sigma$.
Model $M_{1}$ contains an additional set of parameters, $\alpha_{\Delta}$, that describe the distribution of biomarker
\emph{changes} across the population.
It also permits biomarkers that may be Markovian where $X_{1} \sim p(x_{1} | x_{0}, \alpha_{\Delta}, M_{1})$, though
this is not essential.
Explicit examples of these parameters are provided in the following sections.
\begin{figure}[!htpb]
  \begin{center}
    \input{sections/models.tikz}
  \end{center}
  \caption{Bayesian networks of two competing models for biomarker measurements made and before, $y_{i0}$, and after,
      $y_{i1}$, some intervention for subject $i$.}
  \label{fig:models}
\end{figure}
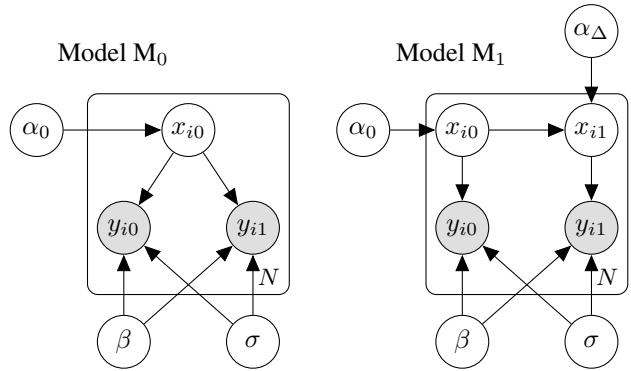

While bias is important in biomarker measurement, it is worth noting here that in the context of
measuring post-treatment change, working with biased measurements presents the same fundamental challenge under certain
assumptions: If we assume bias to be a generalised function such that $\beta: x_{t} \rightarrow x'_{t}$, and assume
that this function is injective (one-to-one) and invertible, then models $M_{0}$ and $M_{1}$ may equivalently be
represented as $M_{0}: x'_{1} = x'_{0}$ and $M_{1}: x'_{1} \neq x'_{0}$.
In the rest of this article, therefore, we do not distinguish between $x$ and $x'$ and use the notation interchangeably
and without loss of generality.

Note that although both models only consider pairs of pre- and post-treatment measurements, $t \in \{0, 1\}$, some
study designs might be interested in biomarker measurements at multiple discreet time intervals,
$t \in \{0, 1, \dots, T\}$, or indeed in cases with arbitrary time intervals, in which case $t \in \mathbb{R}^{+}$.
These use-cases will require further development, but our arguments may also be considered useful in these contexts.

In this paper, we consider two approaches for selecting the model that best explains the observed data.
First, we investigate anomaly detection in post-treatment measurements under model $M_{0}$ serving as a Bayesian
counterpart to null hypothesis testing traditionally used in double-baseline repeatability studies.
Second, we explore the use of \acrfull{bf} and the closely related \acrfull{po} as methods for directly comparing the
likelihoods of the observed data under competing models.
We highlight that an advantage of the latter approach is its ability to incorporate prior knowledge about expected
changes following treatment into the decision-making process.
However, this advantage is accompanied by the added complexity of needing to estimate these expected changes, typically
requiring data from clinical trials.

To compare these approaches with classical double baseline methods, we use the case of real, normally-distributed
biomarkers as a working example.
In this case, the set of baseline parameters are $\alpha_{0} = \{\mu_{0}, \sigma_{0}\}$, and the change parameters
$\alpha_{\Delta} = \{\mu_{\Delta}, \sigma_{\Delta}\}$, which are fully described in Table~\ref{tab:real_assumptions}
along with the assumed data-generating model.

\begin{table}[!htbp]
\centering
\caption{Paramters and assumed models for univariate, normally-disitributed biomarkers.
         Note: Half-Cauchy distribution, HC, assumed to have location 0. }
\begin{tabular}{|p{0.18\linewidth} | p{0.7\linewidth}|}
\hline
    {\bf{Parameter}} & {\bf{Description}} \\
\hline
    $\mu_{0} \in \mathbb{R}$ & Average baseline biomarker value across the population. \\
\hline
    $\sigma_{0} \in \mathbb{R}^{+}$ & Standard-deviation of baseline biomarker values across the population. \\
\hline
    $\mu_{\Delta} \in \mathbb{R}$ & The average change (effect) in the biomarker after treatment. \\
\hline
    $\sigma_{\Delta} \in \mathbb{R}^{+}$ & The standard deviation (heterogeneity) of population biomarker changes. \\
\hline
    $\sigma \in \mathbb{R}^{+}$ & The standard deviation of zero-mean normally-distributed noise (measurement uncertainty). \\
\hline
    $\sigma_{\mu} \in \mathbb{R}^{+}$ & Standard deviation of normal prior used when estimating $\mu_{0}$ and $\mu_{\Delta}$. \\
\hline
    $\gamma_{\sigma} \in \mathbb{R}^{+}$ & Scale of the half-Cauchy prior used when estimating $\sigma_{0}$, $\sigma_{\Delta}$, and $\sigma$. \\
\hline
\hline
    \multicolumn{2}{|c|}{\bf{Model Distributions}} \\
\hline
    \multicolumn{2}{|c|}{$\mu_{0} \sim \mathcal{N}(0, \sigma_{\mu}) \quad \mu_{\Delta} \sim \mathcal{N}(0,  \sigma_{\mu})$} \\
    \multicolumn{2}{|c|}{$\sigma_{0} \sim \text{HC}(\gamma_{\sigma}) \quad \sigma_{\Delta} \sim \text{HC}(\gamma_{\sigma}) \quad \sigma \sim \text{HC}(\gamma_{\sigma})$} \\
    \multicolumn{2}{|c|}{$x_{i0} \sim \mathcal{N}(\mu_{0}, \sigma_{0})$} \\
    \multicolumn{2}{|c|}{$x_{i1} \sim \mathcal{N}(x_{i0} + \mu_{\Delta}, \sigma_{\Delta})$ if $M_{1}$ is true \textbf{else} $x_{i1} = x_{i0}$} \\
    \multicolumn{2}{|c|}{$y_{it} \sim \mathcal{N}(x_{it}, \sigma) $}\\
    \multicolumn{2}{|l|}{where}\\
    \multicolumn{2}{|c|}{$\mathcal{N}(x; a, b) = (2\pi b^{2})^{-1/2}\exp\{- (x - a)^{2} / 2b^{2} \}$}\\
    \multicolumn{2}{|c|}{$\text{HC}(x; \gamma_{\sigma}) = 2\left[ 1+\left(x/\gamma_{\sigma}\right)^{2} \right]/(\pi\gamma_{\sigma}), \quad x \geq 0$}\\
\hline
\end{tabular}
\label{tab:real_assumptions}
\end{table}

\subsection{Assessing change through anomaly detection}\label{subsec:anomaly_detection}
Conventional assessment of response biomarkers for individual patients follows the paradigm of anomaly detection.
In the framework presented above this generally involves estimation of the parameters in model $M_{0}$, and then
determining whether a change observed after treatment would represent a significant anomaly if model $M_{0}$ still
holds.
This is summarized by a \emph{p-value}, which represents the probability of observing the post-treatment measurement—or
a measurement as ``extreme'' as the post-treatment measurement—under the null hypothesis that no real change has
occurred (i.e. $M_{0}$ is true).
Formerly, we wish to evaluate
\begin{equation} \label{eq: pval}
\text{p} = \Pr(T(y) \geq T(y_{1})\ |\ y_{0}, \mathbf{y}_{b}, M_{0})
\end{equation}
where $\mathbf{y}_{b}$ represent the repeat baseline measurements used to determine the parameters of $M_{0}$, and
$T(\cdot)$ represents some test statistic on observed data that describes how unlikely observed post-treatment data are
if $M_{0}$ is true.
Exact calculation of equation~\ref{eq: pval} requires knowledge of the marginal distribution of $y$ conditional on all
observed data
\begin{align}
p(&y | y_{0}, \mathbf{y}_{b},M_{0}) = \notag \\
&\int\limits_{\alpha_{0}}\int\limits_{\sigma}p(y | y_{0}, \alpha_{0}, \sigma, M_{0})p(\alpha_{0}| y_{0},\mathbf{y}_{b},
M_{0})p(\sigma|\mathbf{y}_{b}, M_{0})\text{d}\alpha_{0}\text{d}\sigma \label{eq:conditional}
\end{align}
In all but a handful of cases, it is generally not possible to directly evaluate this integral.
In the following two subsections we demonstrate how conventional repeatability analysis may be considered a special case
of this formulation, and then discuss how numerical techniques can be used to solve this problem in a general setting.

\subsubsection{Analytical solutions}\label{subsubsec:anomaly_detection_analytical}
In the working example of normally distributed biomarkers with $\alpha_{0} = \left(\mu_{0}, \sigma_{0}\right)$ the
convention is to assume that these parameters, along with error $\sigma$, are known \emph{a-priori} (estimated from
double baseline data for example).
In the Supplementary Material, theorem~\ref{theorem:conditional_posterior} we prove that under these assumptions, the
following conditional probability is true.
\begin{equation} \label{eq:conditional_gaussian}
p(y | y_{0}, \sigma, \mu_{0}, \sigma_{0}, M_{0}) = \mathcal{N}\left(y; \mu^{*}, \sigma^{*} \right)
\end{equation}
where
\[
\mu^{*} = ICC\cdot y_{0} + (1-ICC)\cdot \mu_{0} \quad \sigma^{*} = \sqrt{1+ICC}\cdot \sigma
\]
If the variance of baseline biomarker values amongst the population is much larger than the measurement uncertainty
($\sigma_{0} \gg \sigma$) such that $ICC \rightarrow 1$, then we obtain the classical result
$p(y | y_{0}, \sigma, M_{0}) = \mathcal{N}\left(y; y_{0}, \sqrt{2}\sigma\right)$.
This may be re-written as $p(d | \sigma, M_{0}) = \mathcal{N}\left(d; 0, \sqrt{2}\sigma\right)$ where $d = y - y_{0}$.
When determining whether a change following treatment is significant, the test static is the
$z$-statistic, $z = (y_{1} - y_{0}) / \sqrt{2}\sigma$, such that for $\text{p}< 0.05$ in a two-sided hypothesis test we
require $|d| > 1.96\sqrt{2}\sigma = RC$ as discussed earlier (Section~\ref{subsec:blandaltman}).

Although simpler derivations of this final result are possible, equation~\ref{eq:conditional_gaussian}
provides a transparent formulation; it demonstrates a key assumption made throughout conventional repeatability
analysis that $ICC \rightarrow 1$ indicating that significant changes after treatment are independent of baseline value.
This may be a poor assumption to make.
For example, in the case of changes in ADC in bone cancer, a change from 0.8 to 1.5 $\times 10^{-3}$ s/mm$^{2}$ might
have a different biological interpretation (reduction in tumour cellularity) than a change from 0.3 to 1.0
$\times 10^{-3}$ s/mm$^{2}$ (a potential region of new disease in an area of displaced yellow marrow).
The \acrshort{icc}, along with $\mu_{0}$, provide some context to these changes as they describe the distribution
of expected baseline measurements for a patient cohort.

As a corollary, if we continue with the assumption that $ICC\rightarrow 1$ it is also relatively straightforward
to marginalize over the uncertainty of $\sigma$ and obtain the precise form of equation~\ref{eq:conditional} (see
Supplementary Material, corollary~\ref{theorem:studentt}):
\[
p(y | y_{0}, \textbf{y}_{b}, M_{0}) = \frac{\Gamma\left(\frac{N_{b} + 1}{2}\right)}{\sqrt{\pi N_{b}}\Gamma\left(\frac{N_{b}}{2}\right)}
   \left( 1 + \frac{t^{2}}{N_{b}} \right)^{-\frac{N_{b}+1}{2}}
\]
where
\begin{gather*}
t = \frac{y - y_{0}}{\sqrt{\frac{1}{N_{b}}\sum_{i}\left(y_{j02} -
    y_{j01}\right)^{2}}} = \frac{d}{\sqrt{2}\hat{\sigma}}
\end{gather*}
Specifically, $p(y|y_0, \mathbf{y}_b, M_0)$ follows a Student-$t$ distribution with $N_{b}$ degrees of freedom.
Significant change following treatment is determined by assessing whether $t$ exceeds a threshold defined by the
Student-$t$ distribution at the chosen significance level.
This corresponds to the condition $|d| > \gamma \sqrt{2} \hat{\sigma}$ where $\gamma$ is a factor dependent on $N_{b}$
and the specified significance level.
Notably, $\gamma \to 1.96$ as $N_{b} \to \infty$ for $p = 0.05$.

\subsubsection{Numerical solutions}\label{subsubsec:anomaly_detection_numerical}
Although exact marginalization of parameters $\alpha_{0}$ and $\sigma$ in equation~\ref{eq:conditional} is generally not
possible, it can be achieved numerically through \acrfull{mcmc} sampling techniques.
Using this approach we draw samples, $s \in \{1, 2, \dots, M_{s}\}$, of the parameters of interest from their
conditional posterior distribution,
$\alpha^{s}_{0}, x_{0}^{s}, \sigma^{s} \sim p(\alpha_{0}, x_{0}, \sigma | \alpha^{s-1}_{0}, x^{s-1}_{0}, \sigma^{s-1}, y_{0},\textbf{y}_{b}) $,
and subsequently draw posterior predictive data samples from the assumed likelihood distribution for model $M_{0}$,
$y^{s} \sim p(y | x_{0}^{s}, \sigma^{s},M_{0})$; the complete set of $y^{s}$ are then samples from the marginalised
distribution:
$p(y | y_{0}, \mathbf{y}_{b}, M_{0})$.

\vspace{5pt}
\underline{\emph{Pseudocode}}
\begin{algorithmic}
\State $\alpha^{0}_{0}, x_{0}^{0}, \sigma^{0}$ = Initial parameter estimates
\State $N_{s}$ = number of samples
\For{$s \in \{1, 2, \dots, N_{s} \}$}
\State $\alpha_{0}^{s} \sim p(\alpha_{0} | \alpha_{0}^{s-1}, x_{0}^{s-1}, \sigma^{s-1}, y_{0}, \mathbf{y}_{b})$
\State $x_{0}^{s} \sim p(x_{0} | \alpha_{0}^{s}, x_{0}^{s-1}, \sigma^{s-1}, y_{0}, \mathbf{y}_{b})$
\State $\sigma^{s} \sim p(\sigma | \alpha_{0}^{s}, x_{0}^{s}, \sigma^{s-1}, y_{0}, \mathbf{y}_{b})$
\State $y^{s} \sim p(y | x_{0}^{s}, \sigma, M_{0})$
\EndFor\\
\Return $\mathbf{y}^{s}=\{y^{1},\cdots,y^{N_{s}}\}$
\end{algorithmic}
\vspace{5pt}

Initial parameter estimates $\alpha_0$, $x_0$, and $\sigma_0$ are highly model-dependent but can, for instance, be
obtained using maximum likelihood estimation to improve convergence speed.
It is standard practice to discard or ``burn'' an initial portion of the samples to minimize the influence of the
starting conditions on subsequent samples.
Additionally, samples are typically generated across two or more independent chains to ensure that the model sampling
has converged—meaning the chains appear to sample parameters from the same underlying distribution.
When multiple chains are used, the Gelman-Rubin statistic $\hat{R}$ can be employed to assess convergence.
A value of $\hat{R}$ close to 1 indicates that the chains have converged well~\cite{gelman1992inference}.

Novelty detection using \acrshort{mcmc} sampling is more involved than for the analytical solutions presented above.
It may involve analysis of biomarkers other than real variables, or involve probability distributions other than
the univariate/multivariate normal distributions.
Our approach is to approximate the posterior predictive distribution (equation\ref{eq:conditional}) as some suitable
density estimation derived from the posterior predictive samples $\mathbf{y}^{s}$
\[
p(y|y_{0}, \mathbf{y}_{b}, M_{0})\approx\widehat{p}_{s}(y|\mathbf{y}^{s})
\]
Density approximation $\widehat{p}_{s}$ might be based on kernel density estimation or, as we employ, a relevant
finite mixture model.
By estimating the maximum location of the posterior predictive distribution
\[
y^{\dagger} = \argmax_{y}\left[\widehat{p}_{s}(y|\mathbf{y}^{s})\right]
\]
we define the ``kernel distance'', $D(y)\in\mathbb{R}^{+}$:
\[
D(y) = \left|\widehat{p}_{s}(y^{\dagger}|\mathbf{y}^{s}) - \widehat{p}_{s}(y|\mathbf{y}^{s})\right|
\]
Computing this distance for all samples generated via the \acrshort{mcmc} algorithm, $D(\mathbf{y}^{s})$, the credible
region is the set of all $y$ for which $D(y) > (1-\alpha)^{\text{th}}$ percentile of $D(\mathbf{y}^{s})$, where
$\alpha$ represents some desired significance level (e.g. $\alpha$ = 0.05).
Any post-treatment measurement, $y_{1}$, that lies outside this region is then considered to be an anomaly.
An example of an estimated credible region defined for synthetic data simulated from a multivariate normal distribution
is illustrated in Figure~\ref{fig:kernel_dist}.
\begin{figure}[!htpb]
   \centering
   \includegraphics[width = 0.75\linewidth]{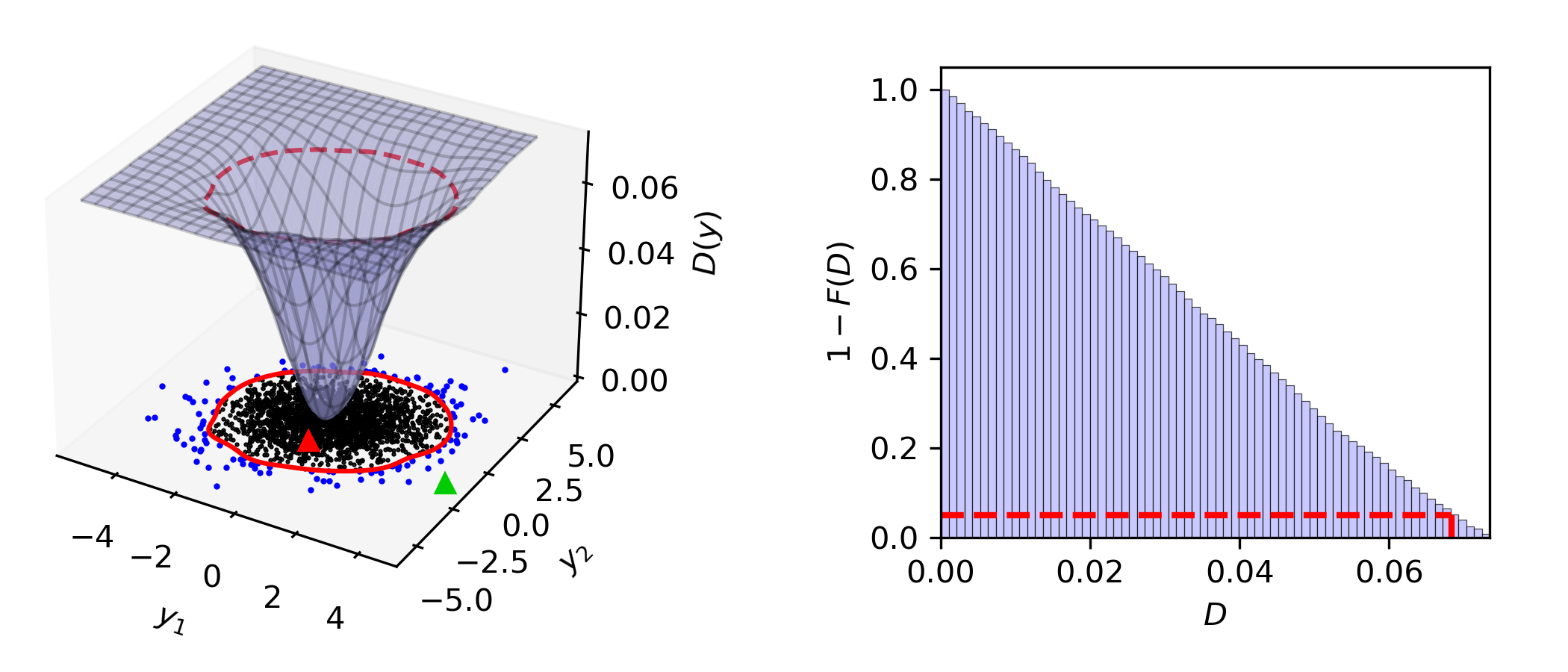}
   \caption{Example 95\% credible region (red contour) for 5000 samples from a 2-dimensional normal distribution,
      computed using the kernel distance, $D$ (blue surface). \acrshort{mcmc} samples that are found to be outside the
      credible regions are color-coded blue, whilst those inside are color-coded black. Triangles represent example
      data from a post-treatment measurement; the green triangle shows a case that would be considered a significant
      change compared to baseline, whilst the red triangle illustrates a case for which there is insufficient evidence
      for a significant change. In this example, we use an isotropic Gaussian kernel with bandwidth $h=0.5$ to generate
      kernel density estimate.}
   \label{fig:kernel_dist}
\end{figure}

\subsection{Bayes factors}\label{subsec:bayes_factors}
As an alternative to novelty detection, the Bayes factor compares the evidence in favour of both models $M_{0}$ and
$M_{1}$ by calculating the ratio of the marginal likelihoods of the observed data, $\mathbf{y}$:
\begin{equation}\label{eq:bf}
BF_{10} = \frac{p(\mathbf{y}|M_{1})}{p(\mathbf{y}|M_{0})} = \frac{\int_{\alpha_{1}}p(\mathbf{y}|\alpha_{1}, M_{1})
   p(\alpha_{1})\text{d}\alpha_{1}}{\int_{\alpha_{0}}p(\mathbf{y}|\alpha_{0}, M_{0})p(\alpha_{0})\text{d}\alpha_{0}}
\end{equation}
for $BF_{10} \in \mathbb{R}^{+}$~\cite{gronau_bridge}.
A theoretical advantage of the Bayes factor is that it allows one to both reject the null model, $M_{0}$, in favour of
$M_{1}$ when $BF_{10} > 1$ or confirm its existence when $BF_{10} < 1$.
Moreover, if $BF_{10} \approx 1$ to within some appropriate tolerance, this indicates that inadequate data is present
to choose between either model.
As when using p-values during statistical significance testing,  there are no fixed rules when using Bayes factors to
interpret results of some experiment; the final decision should depend on the context in which the experiment was
performed and might be influenced by supplementary information.
Our aim in this article is to outline methods for estimating the Bayes factor within repeatability studies as a tool to
aid measurement of the magnitude of change in biomarkers following treatment.
Whether this change reaches some threshold to warrant any clinical decision should always include other factors
such as patient history, bloods markers, treatment selection, and expected patient benefit.
That said, several experts have suggested approximate schemes for decision-making based on Bayes factors, including a
framework by Lee and Wagenmakers, as summarized in Table~\ref{tab:bf10}~\cite{lee2014bayesian}.

\begin{table}[!htpb]
\caption{Interpretation thresholds for Bayes factors suggested by Jeffreys}
\begin{center}
{
\renewcommand{\arraystretch}{1.25}
\begin{tabular}{|c|l|}
\hline
$\mathbf{BF_{10}}$ & \multicolumn{1}{|c|}{\textbf{Interpretation}} \\
\hline
$>$ 30 & Very strong evidence for $M_{1}$ compared to $M_{0}$\\
\hline
10--30 & Strong evidence for $M_{1}$ compared to $M_{0}$\\
\hline
3--10 & Moderate evidence for $M_{1}$ compared to $M_{0}$\\
\hline
1--3 & Anecdotal evidence for $M_{1}$ compared to $M_{0}$ \\
\hline
$\frac{1}{3}$ - 1 & Anecdotal evidence for $M_{0}$ compared to $M_{1}$ \\
\hline
$\frac{1}{10}$ - $\frac{1}{3}$ & Moderate evidence for $M_{0}$ compared to $M_{1}$ \\
\hline
$<\frac{1}{10}$ & Strong evidence for $M_{0}$ compared to $M_{1}$ \\
\hline
\end{tabular}
}
\end{center}
\label{tab:bf10}
\end{table}

\subsubsection{Analytical solutions}\label{subsubsec:bayes_factors_analytical}
For our working example of real-value normally-distributed measurements (table~\ref{tab:real_assumptions}), it is
possible to derive a formula for the Bayes factor for the measured change in a single patient, $d_{y} = y_{1}-y_{0}$,
provided the parameters for both models are known.
In the Supplementary Material, theorem~\ref{theorem:bf_gaussian_proof}, we demonstrate that the natural logarithm of
the Bayes factor in this instance is given by
\begin{equation}\label{eq:logbf_normal}
\ln (BF_{10}) = \frac{1}{2}\left(\ln (1-\eta) - (1-\eta)\left(\frac{d_{y} - \mu_{\Delta}}{\sqrt{2}\sigma} \right)^{2} +
   \left( \frac{d_{y}}{\sqrt{2}\sigma}\right)^{2}\right)
\end{equation}
where 
\[
\eta = \frac{\sigma_{\Delta}^{2}}{\sigma_{\Delta}^{2} + 2\sigma^{2}}
\]
encapsulates the magnitude of the heterogeneity of biomarker changes following treatment, $\sigma_{\Delta}$, in terms
of measurement uncertainty (described as the \acrfull{ims} in~\cite{Thrussell:2022vj}).
Some intuition of the Bayes factor can be made by investigating the functional form of its expectation when model
$M_{0}$ is true or when model $M_{1}$ is true (see Supplementary Material, corollary~\ref{theorem:expected_bf} for the
derivation):
\begin{align}
\mathbb{E}_{d_{y}|M_{0}}\{\ln (BF_{10})\} &= \frac{1}{2} \left(\ln(1-\eta) + \eta - (1-\eta)\xi^{2} \right) \notag \\
\mathbb{E}_{d_{y}|M_{1}}\{\ln (BF_{10})\} &= \frac{1}{2} \left(\ln(1-\eta) +
   \frac{\eta}{1-\eta} + \xi^{2} \right) \notag
\end{align}
where $\xi = |\mu_{\Delta}|/\sqrt{2}\sigma$ represents the \emph{effect size} of the treatment.
Plots of these functions are illustrated in Figure ~\ref{fig:expected_bf} for different values of $\eta$ and $\xi$.
There is a clear trend that as $\eta$ and $\xi$ increase, the expectation of the Bayes factor when model $M_{1}$ is
true increases, indicating better sensitivity for detecting post-treatment change.
Similarly, increasing $\xi$ reduces the expectation value when model $M_{0}$ is true indicating improved specificity
for confirming no-change.
However, it is interesting that for moderate effect sizes this specificity is lost as the heterogeneity of treatment
change increases ($\eta \rightarrow 1$).
\begin{figure}[htbp]
   \centering
   \includegraphics[width=0.75\textwidth]{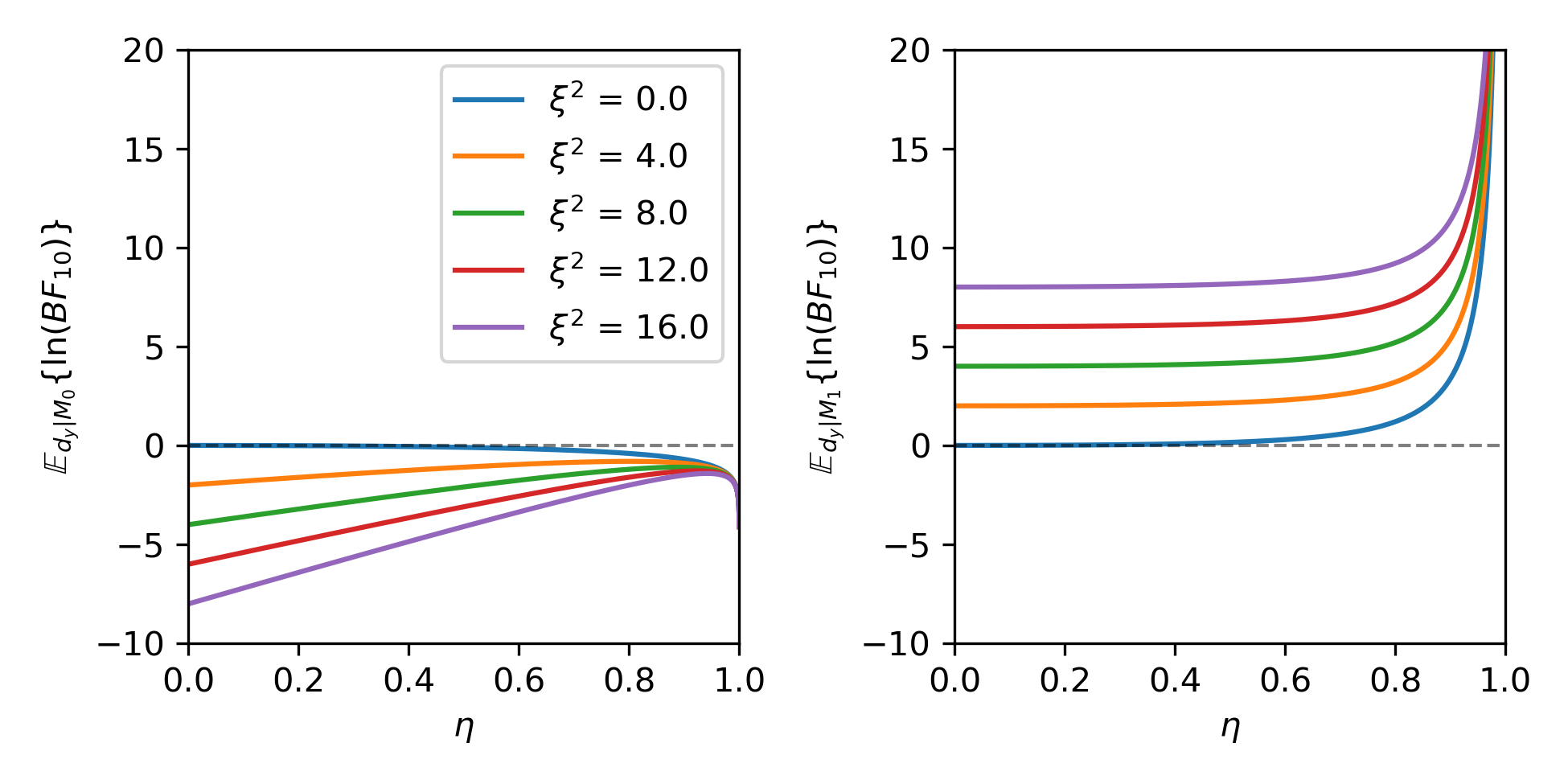}
   \caption{Variation in the expected Bayes factor for real, normally-distributed biomarker measurements.}
   \label{fig:expected_bf}
\end{figure}

\subsection{Numerical solutions with Bayesian mixture modelling}\label{subsec:mixture_modelling}
In the results presented in the preceding section we made the assumption that model $M_{1}$ parameters,
$\alpha_{1} = \left(\sigma_{\Delta}, \mu_{\Delta}\right)$ are known.
Unfortunately, unlike for $M_{0}$ and parameters $\alpha_{0}$ where double baseline experiments can be specifically
constructed such that the relevant model assumptions are somewhat guaranteed, it is difficult to acquire data that
provide direct estimation of $\alpha_{1}$.
In a clinical trial setting, for example, this would require knowledge of the response label for each patient,
$z_{i} \in \{M_{0}, M_{1}\}$.
This could, for example, be derived from other criteria such as patient overall survival, progression-free survival, or
even subjective radiologist assessment of all available patient data.
However, in many cancer trials medical imaging is used as a key criteria for early response assessment and so joint
estimation of model labels $z_{i}$ and model parameters would be preferable.
Here we propose that Bayesian mixture modelling may be used to achieve this and subsequently estimate Bayes factors for
each subject within a trial.

Consider a trial consisting of $N$ subjects, where biomarker measurements are made for each subject before and after
treatment, $\mathbf{y} = \{y_{i0}, y_{i1} \}$ for $i \in \{1, \dots, N\}$.
Furthermore, assume that double baseline measurements, $\mathbf{y}_{j0} = \{y_{j01}, y_{j02}\}$ are available
for $j \in \{1, \dots, N_{b}\}$.
These data - whether belonging to the same group, a subset, or an entirely different cohort than the main trial -
should be representative of the biomarker measurements made within the main study, and acquired in such a way that
ensures the assumption of model $M_{0}$ holds.
The prior probability of each model is now incorporated as an unknown parameter $\lambda \in (0, 1)$ such that
$p(z_{i} = M_{1}) = \lambda$ and $p(z_{i} = M_{0}) = 1-\lambda$.
The posterior distribution of all model parameters, marginalised over the unknown data labels is:
\begin{align}
p(&\alpha_{0}, \alpha_{1}, \sigma, \lambda, \mathbf{x}, \mathbf{x}_{b} | \mathbf{y}, \mathbf{y}_{b}) \propto \notag \\
& p(\alpha_{0}) p(\alpha_{1}) p(\sigma) \prod\limits_{j=1}^{N_{b}} p(y_{bj1}|x_{bj}, \sigma)p(y_{bj2}|x_{bj}, \sigma)
   p(x_{bj}|\alpha_{0}) \notag \\
&\quad \times \prod\limits_{i=1}^{N}p(y_{i0}|x_{i0}, \sigma)p(x_{i0} | \alpha_{0}) \{\lambda\cdot p(y_{i1} | x_{i1},
   \sigma)p(x_{i1} | x_{i0}, \alpha_{1})\notag\\
&\qquad\qquad + (1-\lambda)\cdot p(y_{i1} | x_{i0}, \sigma)\}\label{eq:posterior_trial}
\end{align}
where $p(\alpha_{0})$,  $p(\alpha_{1})$, and $p(\sigma)$ represent suitably chosen prior distribution on model
parameters, and $\mathbf{x}$ and $\mathbf{x}_{b}$ represent the true (unknown) biomarker values from which measurements
$\mathbf{y}$ and $\mathbf{y}_{b}$ have been acquired.
Samples from this posterior distribution may be directly sampled using Bayesian inference engines such as Stan, which
provides a sequence of samples $s \in \{1, 2, \dots, N_{s} \}$ of model parameters
($\alpha^{s}_{0}, \alpha^{s}_{1}, \sigma^{s}, \lambda^{s}, \mathbf{x}^{s}, \mathbf{x}^{bs}$).

Directly obtaining Bayes factors using this approach is complicated by not knowing the prior probability of either
model, $p(M_{0})$ or $p(M_{1}$).  We therefore use a related quantity, the \acrfull{po}, defined as
\[
PO_{10} = \frac{p(z_{i} = M_{1}|y_{i0}, y_{i1})}{p(z_{i} = M_{0}|y_{i0}, y_{i1})}
\]
which is equivalent to the Bayes factor when $P(M_{0}) = P(M_{1})$.
Although the posterior distribution for each patient label $p(z_{i} = M_{k} | y_{i0}, y_{i1})$ is not known, samples
from this distribution can be obtained during the \acrshort{mcmc} process by sampling from a categorical
distribution, conditional on all model parameters:
\[
z^{s} \sim 
\begin{cases}
    M_{0},& \text{with probability } p^{s}_{0}\\
    M_{1},& \text{with probability } 1 - p^{s}_{0}\\
\end{cases}
\]
where
\[
p^{s}_{0} = \frac{\widetilde{p^{s}_{0}}}{\widetilde{p^{s}_{0}} + \widetilde{p^{s}_{1}}}
\]
and
\begin{align}
\widetilde{p^{s}_{0}} &= p(y_{i0} | x^{s}_{i0}, \sigma^{s})p(y_{i1} | x^{s}_{i0}, \sigma^{s})(1 - \lambda^{s}) \notag \\
\widetilde{p^{s}_{1}} &= p(y_{i0} | x^{s}_{i0}, \sigma^{s})p(y_{i1} | x^{s}_{i1}, \sigma^{s}) \lambda^{s} \notag
\end{align}
Finally, we estimate the posterior odds by:
\[
PO_{10} = \frac{\hat{p}(z_{i} = M_{1} | y_{i0}, y_{i1})}{\hat{p}(z_{i} = M_{0} | y_{i0}, y_{i1})} =
\frac{\frac{1}{N_{s}}\sum_{s}(z_{i}^{s} == M_{1})}{\frac{1}{N_{s}}\sum_{s}(z_{i}^{s} == M_{0})}
\]

\subsection{Application to habitat mapping}\label{subsec:habitat_modelling}
In this article we demonstrate a use-case of the methods presented above in monitoring heterogeneous change in tumours.
There are many approaches for quantifying intra-tumoral heterogeneity in cancer.
This may include techniques like radiomics, which attempt to derive hand-crafted or deep features from images that are
designed to represent apparent texture within the images.
This approach quantifies heterogeneity within images to a set of (hopefully) independent univariate biomarkers, for
which  conventional repeatability assessments are useful.
However, radiomics features generally ignore the spatial context, since the extracted features represent the entire
image or tumor.
An alternative is ``habitat imaging'', which attempts to automatically classify subregions within the tumour that
ideally represent distinct biological phenotypes.
The methods for defining such maps are numerous, but in general they lead to a biomarker that may be considered to be
the proportion of tumour in each of $M$ possible tissue categories.
That is $x = \{x_{1}, \dots, x_{M}\}$ is a simplex with the property that $\sum_{m}x_{m} = 1$, and $x_{m}$
is the proportion of tissue belong to class $m$.
As far as we are aware, no method of repeatability evaluation for such biomarkers is available.

Our approach to obtaining response habitats is illustrated in Figure~\ref{fig:habitat_example}, inspired from the work
of previous publications~\cite{Koh:2009wy, blackledge2017assessing, dalili2017quantitative}.
Considering each tumour individually, we derive the 10\textsuperscript{th} and 90\textsuperscript{th} percentiles of
the \acrshort{adc} distribution from baseline measurements, leading to $M=3$ distinct \acrshort{adc} regions.
The number of \acrshort{adc} voxels within these regions in the tumour after treatment constitute our biomarker
measurement of interest, $y_{1} \in \mathbb{Z}^{3}$.
It is important to note that we do not directly normalize the measurement, $y_{1m}\rightarrow y_{1m}/\sum_{m}y_{1m}$, as
this leads to numerical instability wherever there are no voxels within a particular ADC range (can occur for small
tumours).
Visualization of these biomarker measurements, however, is performed following normalization to improve interpretation.
\begin{figure*}[h]
   \centering
   \includegraphics[width = \linewidth]{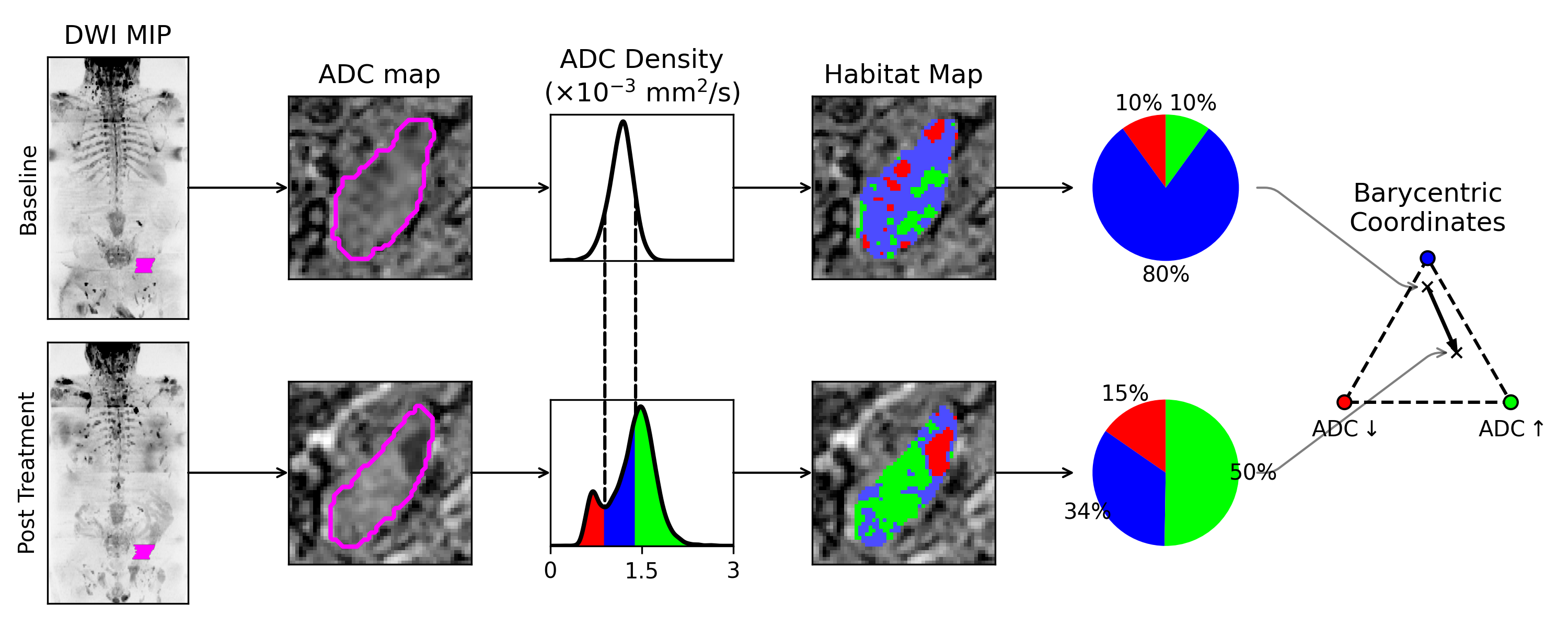}
   \caption{Our approach to habitat mapping for post-treatment assessment consists (from left to right) of: (i)
       Delineating \acrfull{rois} around individual lections on clinical imaging (represented here as a pink overlay on
       the high b-value \acrfull{mip}).
       (ii) Transferring these contours to the derived \acrshort{adc} maps.
       (iii) Deriving \acrshort{adc} percentiles from the baseline distribution and applying them to the follow-up scan
       to derive regions within the percentile ranges.
       (iv) Visualizing the location of responding regions as green where \acrshort{adc} has increased, and progressing
       regions as red where \acrshort{adc} has decreased.
       (v) Identifying the proportions of like habitat voxels as a pie-chart to represent the simplex, or (vi) using
       barycentric coordinates to monitor how the derived simplexes change following treatment.}
   \label{fig:habitat_example}
\end{figure*}

Whether a particular value of $y_{1}$ constitutes a significant change following treatment can be determined using
our theoretical framework and the availability of some double-baseline data $\mathbf{y}_{b}$.
A key assumption in our model is that $x_{0}$ is a known quantity, as it has been defined by setting the percentile
values for the baseline distribution, $x_{0} = \mu_{0} = (0.1, 0.8, 0.1)$.
Additional parameters, along with the data generating model, are presented in Table~\ref{tab:simplex_assumptions}

\begin{table}[!htbp]
\centering
\caption{Paramters and assumed models for unnormalised simplex (habitat) biomarkers.
         Note: Half-Cauchy distribution, HC, assumed to have location 0. }
\begin{tabular}{|p{0.18\linewidth} | p{0.7\linewidth}|}
\hline
    {\bf{Parameter}} & {\bf{Description}} \\
\hline
    $\mu_{0} \in \mathbb{R}^{M}$ & The assumed baseline habitat values.
    Here equal to $(0.1, 0.8, 0.1)$. \\
\hline
    $\mu_{1} \in \mathbb{R}^{M}$ & The center (average) of post-treatment values. \\
\hline
    $\tau \in \mathbb{R}^{+}$ & Precision of baseline measurements. \\
\hline
    $\tau_{1} \in \mathbb{R}^{+}$ & Concentration of population post-treatment values. \\
\hline
    $N_{v} \in \mathbb{Z}^{+}$ & Number of tumour voxels ($\sum_{m}y_{m} = N_{v}$). \\
\hline
    $\gamma_{\tau} \in \mathbb{R}^{+}$ & Scale of the half-Cauchy prior used when estimating $\tau$ and $\tau_{1}$. \\
\hline
\hline
    \multicolumn{2}{|c|}{\bf{Model Distributions}} \\
\hline
    \multicolumn{2}{|c|}{$\tau \sim \text{HC}(\gamma_{\tau})$} \\
    \multicolumn{2}{|c|}{$x_{i1} \sim \text{Dir}(\mu_{1}, \tau_{1})$ if $M_{1}$ is true \textbf{else} $x_{i1} \sim \text{Dir}(\mu_{0}, \tau)$} \\
    \multicolumn{2}{|c|}{$y_{i1} \sim \text{Mult}(x_{i1}, N_{v}) $}\\
    \multicolumn{2}{|l|}{where}\\
    \multicolumn{2}{|c|}{$\text{Dir}(x; \mu, \tau) = \frac{\Gamma(\tau)}{\prod\limits_{m}\Gamma(\tau\mu_{m})}\prod\limits_{m=1}^{M}
   x_{m}^{\tau\mu_{m} - 1}$}\\
    \multicolumn{2}{|c|}{$\text{Mult}(y; x, N_{v}) = \frac{N_{v}!}{y_{1}!\cdots y_{M}!}x_{1}^{y_{1}}\cdots x_{M}^{y_{M}}$}\\
    \multicolumn{2}{|c|}{$\text{HC}(x; \gamma_{\tau}) = 2\left[ 1+\left(x/\gamma_{\tau}\right)^{2} \right]/(\pi\gamma_{\tau})$}\\
\hline
\end{tabular}
\label{tab:simplex_assumptions}
\end{table}

%% file: sections/models.tikz
\begin{tikzpicture}

\node[obs] (y0) {$y_{i0}$};
\node[obs, right=of y0] (y1) {$y_{i1}$};
\node[latent] (x0) at ($(y0)!0.5!(y1) + (0,1.3cm)$) {$x_{i0}$};
\node[latent, left=of x0, xshift=-0.3cm] (a0) {$\alpha_{0}$};
\node[latent, below=of y0, yshift=0.2cm] (beta) {$\beta$};
\node[latent, below=of y1, yshift=0.2cm] (sigma) {$\sigma$};
\plate {plate3} {(x0)(y0)(y1)} {$N$};
\edge {a0}{x0}
\edge {beta, sigma, x0}{y0}
\edge {beta, sigma, x0}{y1}
\node[above=7pt of plate3.north, xshift=-1cm] (model_0_label) {Model M$_{0}$};

\begin{scope}[xshift=4.5cm]
\node[obs] (y0) {$y_{i0}$};
\node[obs, right=of y0] (y1) {$y_{i1}$};
\node[latent] (x0) at ($(y0)!0.0!(y1) + (0,1.3cm)$) {$x_{i0}$};
\node[latent] (x1) at ($(y0)!1.0!(y1) + (0,1.3cm)$) {$x_{i1}$};
\node[latent, left=of x0, xshift=+0.4cm] (a0) {$\alpha_{0}$};
\node[latent, above=of x1, yshift=-0.4cm] (ad) {$\alpha_{\Delta}$};
\node[latent, below=of y0, yshift=0.2cm] (beta) {$\beta$};
\node[latent, below=of y1, yshift=0.2cm] (sigma) {$\sigma$};
\plate {plate3} {(x0)(y0)(y1)} {$N$};
\edge {a0}{x0}
\edge {ad, x0}{x1}
\edge {beta, sigma, x0}{y0}
\edge {beta, sigma, x1}{y1}
\node[above=7pt of plate3.north, xshift=-1cm] (model_1_label) {Model M$_{1}$};
\end{scope}

\end{tikzpicture}

%% file: sections/methods.tex
\subsection{Simulation Studies}

Data were simulated for both median ADC and habitat models (tables~\ref{tab:real_assumptions}
and~\ref{tab:simplex_assumptions} respectively) based on the parameter values and ranges outlined in
tables~\ref{tab:median_sims} and~\ref{tab:habitat_sims}.
For parameters specified as ranges, Latin hypercube sampling was employed to generate five hundred unique combinations
of these parameters, ensuring a thorough exploration of the parameter space.
For each parameter combination, data were sampled 20 times, and the models were fit using Stan with the code
provided in Supplementary Information Sections~\ref{subsec:median_adc_stan_code} and~\ref{subsec:habitat_stan_code}.
For habitat data generation, the number of voxels associated with each datum was sampled from a log-normal distribution
characterized by a mean ($\mu_{v}$) and standard deviation ($\sigma_{v}$), derived from voxel size estimates observed in
the clinical studies.

For each data sample and parameter combination, we evaluated the following diagnostics:
\begin{itemize}
 \item {\bf{Bias}}: The bias for each parameter was calculated as
       $\left(\left(y_{med}-y_{true}\right)/y_{true}\right)\times 100\%$, where $y_{true}$ is the true parameter, and
       $y_{med}$ is the median of the posterior samples for that parameter.
 \item {\bf{Coverage}}: The coverage was defined as the percentage of cases in which the true parameter value fell
       within the 95\% credible interval of the posterior samples.
 \item {\bf{$\hat{R}$ (Gelman-Rubin Diagnostic)}}: The Gelman-Rubin diagnostic~\cite{gelman1992inference} was estimated
       using the ArviZ package~\cite{kumar2019arviz} to assess the convergence of the \acrshort{hmc} chains.
 \item {{\bf{Diagnostic Accuracy}}}: For each parameter combination, the diagnostic accuracy of the posterior odds
       estimate was quantified through sensitivity, specificity, and the area under the receiver operating
       characteristic (ROC) curve (AUC). The true label of change/no change was sampled as
       $l = \text{Uniform}(0, 1) < \lambda$.
\end{itemize}

\begin{table}[t]
\centering
\caption{Parameter values and ranges for median ADC simulations.}
\begin{tabular}{|c|c|}
 \hline
 {\bf{Parameter}} & {\bf{Value/Range}} \\
 \hline
 Mixing proportion, $\lambda$ & $0.4 \rightarrow 0.95$ \\
 \hline
 IMS, $\eta = \sigma_{\Delta}^{2} / \left(\sigma_{\Delta}^{2} + 2\sigma^{2}\right)$ & $0.6 \rightarrow 0.99$  \\
 \hline
 ICC $=\sigma_{0}^{2}/\left(\sigma_{0}^{2} + \sigma^{2}\right)$ & $0.6 \rightarrow 0.99$ \\
 \hline
 Mean change, $\mu_{\Delta}$ & $-0.5 \rightarrow 2.0 \times 10^{-3}$ mm$^{2}$/s \\
 \hline
 Baseline mean, $\mu_{0}$ & $1.0 \times 10^{-3}$ mm$^{2}$/s\\
\hline
 Measurement variation, $\sigma$ & $0.05 \times 10^{-3}$ mm$^{2}$/s \\
\hline
\end{tabular}
\label{tab:median_sims}
\end{table}

\begin{table}[t]
\centering
\caption{Parameter values and ranges for habitat simulations.}
\begin{tabular}{|c|c|}
 \hline
 {\bf{Parameter}} & {\bf{Value/Range}} \\
 \hline
 Mixing proportion, $\lambda$ & $0.4 \rightarrow 0.95$ \\
 \hline
 IMS, $\eta = \tau/\left(\tau + \tau_{1}\right)$ & $0.6 \rightarrow 0.99$ \\
 \hline
 Baseline value, $\mu_{0}$ & (0.1, 0.8, 0.1) \\
 \hline
 Post-treatment center, $\mu_{1}$ & (0.01, 0.54, 0.36) \\
 \hline
 Measurement precision, $\tau$ & 11.54 \\
 \hline
 Voxel size mean, $\mu_{v}$ & 6.37 \\
 \hline
 Voxel size std-dev, $\sigma_{v}$ & 1.38 \\
 \hline
\end{tabular}
\label{tab:habitat_sims}
\end{table}

\subsection{Patient cohort and image acquisition}\label{subsec:patient_cohort}

All WBDWI data were acquired at a single institution as part of three independent imaging studies.
These studies were approved by the local research and ethics committee, and all patients were recruited and consented
at a single institution.
Inclusion criteria were: histopathological diagnosis of prostate cancer, history of bone metastases and no
contraindication for MRI acquisition.

\textbf{Study 1} investigated the use of ADC derived from WBDWI to assess the response of Radium-223 as a therapeutic
agent in 35 patients diagnosed with \acrfull{mcrpc}.
Images were acquired before and after treatment with a median time between scans of 28 days (range: 21--35 days) using
a protocol that consisted of either two b-values (b = 50 and 800s/mm$^{2}$ in two patients or 50 and 900s/mm$^{2}$ in
one other), or three b-values (b = 50, 600, and 900s/mm$^{2}$ in all other patients).
All images were acquired on a 1.5T system, using a 2D twice-refocused spin-echo \acrfull{epi} sequence, with
trace-weighted diffusion encoding and inversion-recovery fat suppression.
Images were acquired over 3--5 imaging stations, comprising 40 slices/station in the axial plane over 128 columns
(left--right/readout direction) and 96--116 rows (anterior--superior/phase--encode direction).
Scanner-based interpolation was used to double the isotropic in--plane image resolution to 1.7mm$^{2}$ and a slice
thickness of 5mm.

\textbf{Study 2} investigated the use of ADC derived from WBDWI to assess the response of a chemotherapeutic agent in
14 patients diagnosed with \acrshort{mcrpc}.
Images were acquired before and after treatment with a median time between scans of 84 days (range: 42--173 days) using
a protocol that consisted of three b-values (b = 50, 600, and 900s/mm$^{2}$).
All images were acquired on a 1.5T system, using a 2D twice-refocused spin-echo \acrfull{epi} sequence, with
trace-weighted diffusion encoding and inversion-recovery fat suppression.
Images were acquired over 3--5 imaging stations, comprising 40 or 50 slices/station in the axial plane over 128 or 150
columns (left--right/readout direction) and 124 or 144 rows (anterior--superior/phase--encode direction).
Scanner-based interpolation was used to double the isotropic in--plane image resolution to 1.3, 1.4, or 1.5mm$^{2}$ and
a slice thickness of 5mm.

\textbf{Study 3} was a double baseline study comprising 10 patients diagnosed with \acrshort{mcrpc}, scanned prior to
treatment.
All patients were scanned twice on the same day, with repositioning between the examinations.
Diffusion-weighted imaging consisted of three b-values (b = 50, 600, and 900s/mm$^{2}$), acquired on a 1.5T system
using a 2D twice-refocused spin-echo \acrfull{epi} sequence, using trace-weighted diffusion encoding and
inversion--recovery fat suppression.
Images were acquired over 3 or 4 imaging stations, comprising 42 slices/station in the axial plane over 128
columns (left--right/readout direction) and 104 or 128 rows (anterior--superior/phase--encode direction).
Scanner-based interpolation was used to double the isotropic in--plane image resolution to 1.7mm$^{2}$ and a slice
thickness of 6mm.

\subsection{Image and data processing}\label{subsec:image_processing}
Using the OsiriX Dicom viewing platform~\cite{Rosset:2004wy}, a radiologist with over five years of experience in body
MRI for cancer, manually delineated up to six target lesions in each patient on the low b-value images, using available
anatomical imaging (e.g.\ Dixon or T2-weighted) to assist in target definition.
They ensured that the same lesion was visible at both time-points for each study, so that they may be monitored over
time.
This resulted in a total of 120, 68, and 73 lesions delineated for Study 1, Study 2, and Study 3 respectively.
The median lesion volumes were (at baseline) 12.5ml (\acrfull{iqr}: 4.2--31.7ml), 22.0ml (9.0--35.9ml), and 5.1ml
(2.0--9.2ml) for studies 1, 2, and 3 respectively.
A dedicated script was developed using the pyOsiriX plugin~\cite{blackledge2016rapid} that calculated maps of
\acrshort{adc} using maximum--likelihood estimation and incorporating all available b-values for each study.

For each lesion, median ADC values within each lesion were recorded.
Furthermore, habitat maps were calculated according the methods described in section~\ref{subsec:habitat_modelling},
using baseline ADC percentile thresholds of 10\% and 90\%.
We subsequently used the Stan programming language to fit our proposed mixture models for each marker as defined in
section~\ref{subsec:habitat_modelling} for \textbf{Study 1} and \textbf{Study 2} independently (code provided in
Supplementary Material~\ref{subsec:median_adc_stan_code} and~\ref{subsec:habitat_stan_code}).
In both cases data from \textbf{Study 3} were used as the double baseline images.
Posterior predictive samples were obtained so that we could investigate our proposed numerical approach to anomaly
detection.
The posterior density of these samples was approximated using a 5-component Dirichlet mixture model using an in-house
expectation maximization routine.
The position with maximum posterior density was subsequently found using a Nelder-Mead optimization routine, with a
coarse grid-search used for initialization.
This allowed us to determine the credible region, using the method outlined in
section~\ref{subsubsec:anomaly_detection_numerical}.
Furthermore, model label samples were generated in order to test our approach to posterior odds calculation
(section~\ref{subsec:mixture_modelling}]).
For each model fit, we used 3 independent \acrshort{mcmc} chains, each consisting of 5000 samples (following a burn-in
period of 500 samples).
We did not perform any sample thinning, leading to a total of 15000 samples per parameter.
Sample convergence was determined using the Gelman-Rubin statistic, $\hat{R}$,~\cite{gelman1992inference}
using the ArviZ package~\cite{kumar2019arviz}.

%% file: sections/results.tex
\subsection{Simulation Studies}

Figure~\ref{fig:adc_median_simualtions} presents the simulation results for our median ADC model.
Overall, the parameter estimates appear to be unbiased, with credible intervals derived from posterior samples
demonstrating good coverage for all sampled parameter values.
However, a notable exception occurs in scenarios where $|\mu_{\Delta}|\rightarrow 0$ and $\eta\rightarrow0.5$.
In these cases, the model struggles to distinguish between post-treatment change and no change, as the models for change
and no-change become increasingly similar.
This challenge is also reflected in the accuracy diagnostics, which are particularly poor in this parameter region.
Notably, in certain parameter ranges, the posterior odds-based model comparison achieves 100\% sensitivity and 100\%
specificity.
This performance starkly contrasts with the use of p-values, which are theoretically limited to a maximum specificity of
95\% by design.

Figure~\ref{fig:adc_habitat_simualtions} presents the simulation results for our habitat model.
Overall, there appears to be minimal bias across the range of parameters investigated, with the posterior credible
intervals effectively capturing the true values.
In terms of accuracy diagnostics, the habitat model demonstrates improved area AUC and sensitivity compared to the
median ADC model.
However, this improvement comes at the cost of reduced specificity in certain parameter ranges.
A general decline in metric performance is observed as the proportion of regions exhibiting change decreases
($\lambda \rightarrow 0.4$), reflecting fewer samples representing change.
Additionally, a drop in performance is evident as the IMS decreases ($\eta \rightarrow 0.6$), further challenging the
model's ability to distinguish between change and no-change scenarios.

\begin{figure*}[!htpb]
   \centering
   \includegraphics[width = \linewidth]{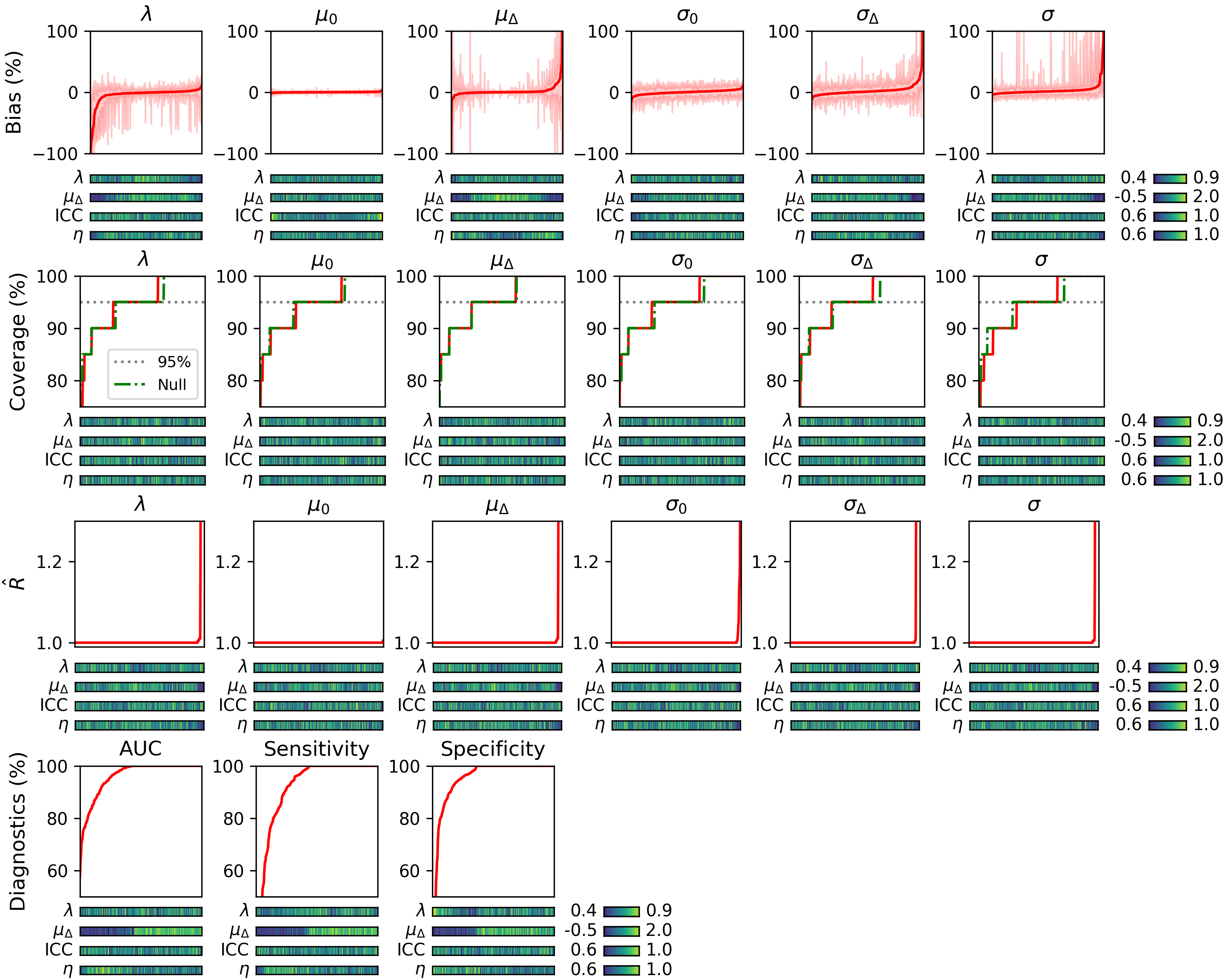}
   \caption{
   Simulation results for real-valued biomarkers (median ADC in this instance). The \textbf{top row} displays the
   percentage bias in all estimated model parameters. The solid red line represents the median estimates across 20 data
   repetitions for each combination of true parameters, while the pink region indicates the 10th to 90th percentile
   range. Values are ordered from smallest to largest bias, with potential explanatory parameters shown in the color bar
   below (also ordered by median bias). No significant bias is observed for any combination of explanatory parameters;
   however, low \(\mu_\Delta\) has the most notable influence. The \textbf{second row} demonstrates the coverage of the
   derived posterior distributions, defined as the proportion of cases where the true parameter falls within the 95\%
   high-density interval (HDI) of the posterior. For context, a null coverage curve is included, derived by randomly
   sampling a standard normal variate and calculating the frequency with which it falls within the empirical 95\%
   confidence interval, then ordering the results. The close agreement between the empirical and null coverage curves
   confirms that the estimator achieves the expected nominal coverage. The \textbf{third row} illustrates the median
   Rubin-Gelman statistic (\(\hat{R}\)) across all simulations for each parameter combination. Excellent convergence
   (\(\hat{R} \to 1\)) is observed for most parameter combinations, except in a small region where \(\mu_\Delta \to 0\)
   and \(\eta \to 0.5\). In this region, model assumptions are violated due to minimal differences between baseline and
   post-treatment data. The \textbf{bottom row} depicts diagnostic metrics of log-\acrshort{po} for the simulation
   studies. Performance improves with increasing \(\mu_\Delta\) and/or \(\eta\), as larger expected differences between
   changing and non-changing biomarkers enhance model accuracy.}
   \label{fig:adc_median_simualtions}
\end{figure*}

\begin{figure}[!htpb]
   \centering
   \includegraphics[width = 0.6\linewidth]{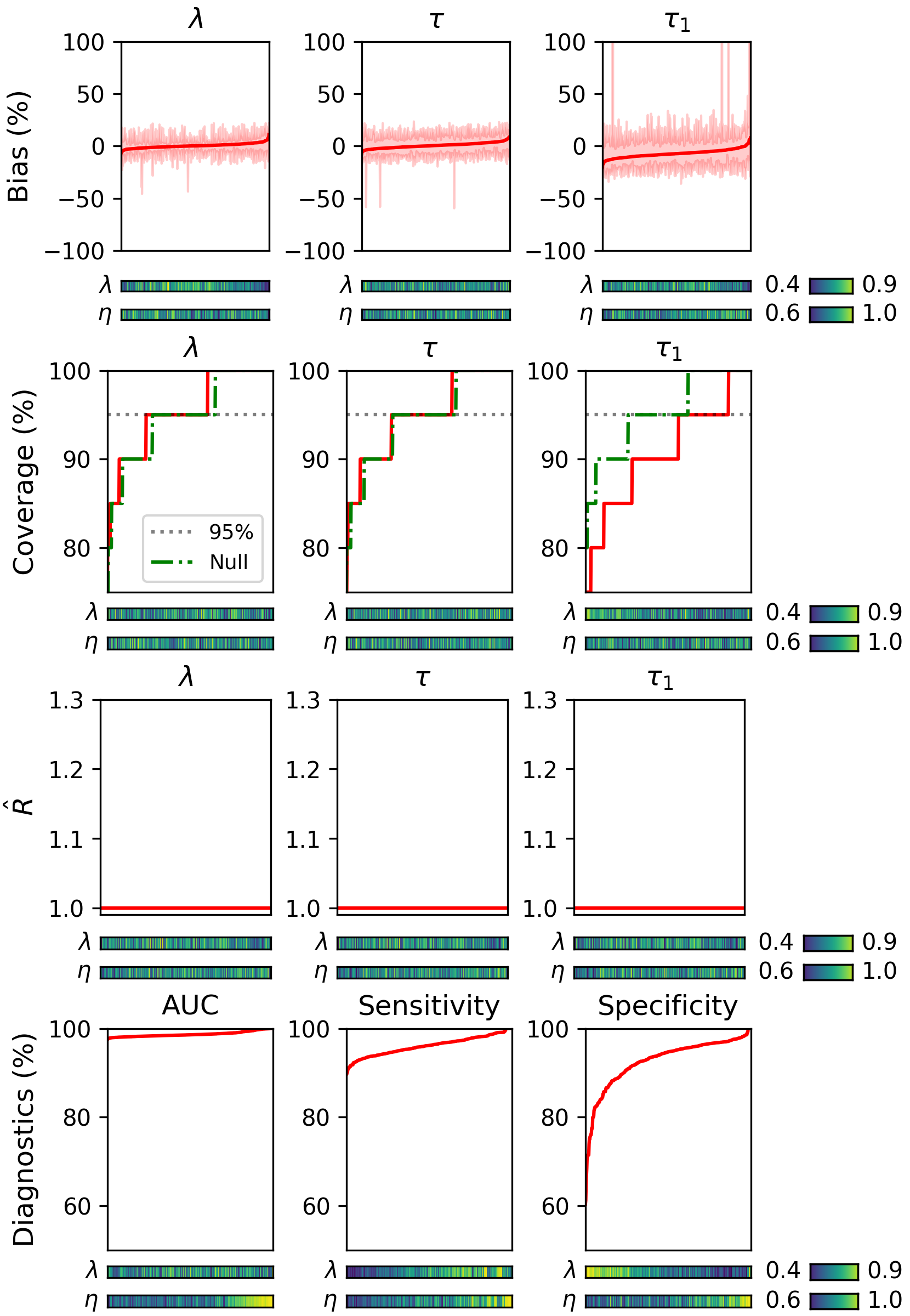}
   \caption{
   Simulation results for dirichlet-multinomial (habitat) biomarkers (median ADC in this instance). All results are
   displayed in the same format as for real varaibles (Figure~\ref{fig:adc_median_simualtions}]). Again, we observe no
   significant evidence of substantial bias in estimators (top row), and observe good convergence for all parameters of
   interest (third row). Good coverage is obserbed for parameters $\tau$ and $\lambda$, but mildly poorer coverage than
   the null curve for $\tau_{1}$ (second row). Finally, there is an improvement in diagnostic performance (bottow
   row) as $\eta$ increases, whilst an increase in $\lambda$ leads to improved sensitivity but poorer specificity.}
   \label{fig:adc_habitat_simualtions}
\end{figure}

\subsection{Patient Studies}\label{sec:results_patients}
In figure~\ref{fig:pat_examples} we present 5 patient example cases, two from study 1 and three from study 2.
In each case we present summaries of each of the target lesions, including their derived habitat map (central slice
only), and their location on the high b-value \acrshort{mip}.
We color-code the lesion on the \acrshort{mip} according to the value of the estimated posterior odds, which is binned
into one of five possible categories: $PO_{10} \leq \frac{1}{3}$ (maroon) indicating moderate evidence of no change,
$\frac{1}{3} < PO_{10} \leq 3$ (orange) indicating insufficient evidence to say whether a lesion demonstrates
significant change,$3 < PO_{10} \leq 10$ (yellow) indicating moderate evidence of significant change,
$10 < PO_{10} \leq 30$ (teal) and $BF_{10} > 10$ (dark green) indicating strong and very strong evidence of change
respectively.
These colors are all used to plot the changes in habitats in barycentric coordinates, along with the credible region
shown as a grey area (for novelty detection).
One observation is that posterior odds appear to confer with the idea of novelty detection; those lesions that
demonstrate significant change after treatment tend to move far beyond the credible region, whilst those with low
posterior odds ($<1/3$) remain within the credible region.

In these examples it is clear that inter-patient and intra-patient/inter-lesion response heterogeneity is evident within
these cohorts.
For those lesions that do change significantly we highlight that there is also divergence in the direction of change;
whilst the vast majority of changing lesions tend to demonstrate an increase in the proportion of voxels with higher
ADC than baseline (figure~\ref{fig:pat_example_1013} for example), there are some that show a general decrease in
ADC (figure~\ref{fig:pat_example_1010}), and also a patient that demonstrates clear intra-tumour heterogeneity in
ADC changes (figure~\ref{fig:pat_example_002}).
\begin{figure}[!htpb]
    \centering
    \begin{subfigure}[t]{0.48\linewidth}
        \includegraphics[width=\linewidth]{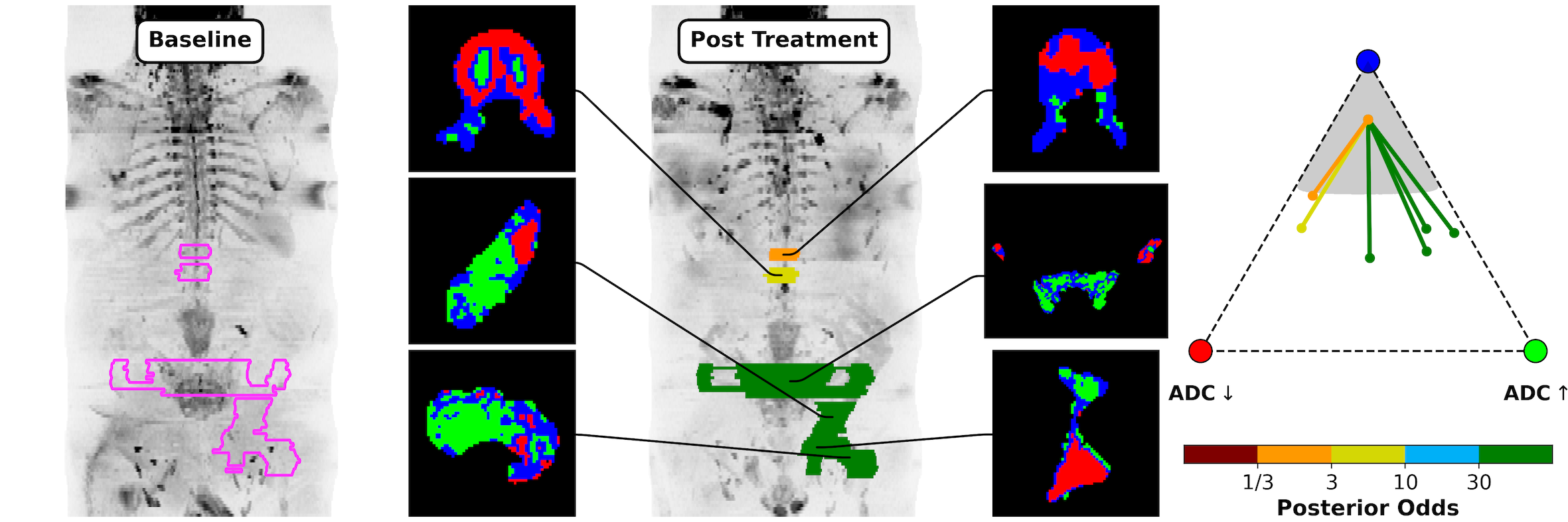}
        \caption{In this patient (study 1) we see wide varierty of responses across different lesions.
            Regions within the vertebral bodies demonstrate no obvious change following treatment
            ($\frac{1}{3} < BF_{10} \leq 3$), whilst those in the pevis do demonstrate significant change with a genera
            increase in ADC.
            One particularly interesting lesion is in the acetabulum which demonstrates intra-tumoral heterogenous
            change, clearly depcted in the bottom--right habitat map.
            demonstrates }
        \label{fig:pat_example_002}
    \end{subfigure}
    \hfill
    \begin{subfigure}[t]{0.48\linewidth}
        \centering
        \includegraphics[width=\linewidth]{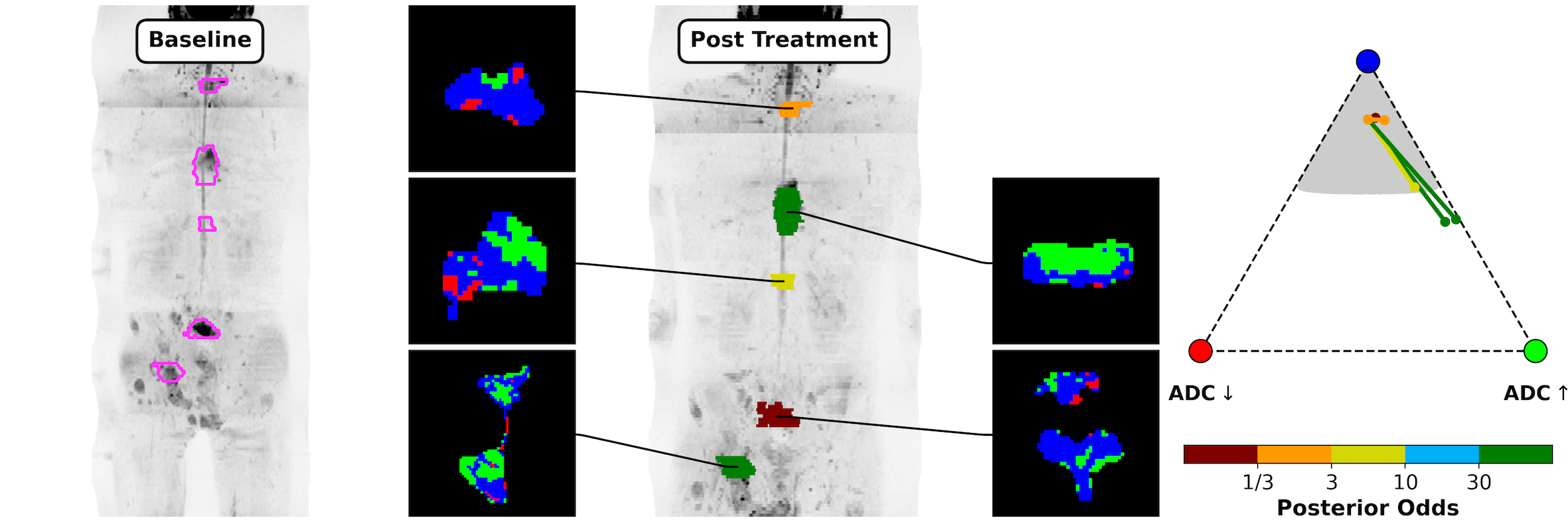}
        \caption{This patient (Study 1) presented with inter-tumoral heterogneity of response, with lesions in the
            sacrum (bottom--right map) and T1 vertebra (top--right) presenting with very little change,
            whilst lesions in the sternum (right--middle) and acetabulum (bottom--left) demonstrate a significant shift
            towards increasing ADC.}
        \label{fig:pat_example_0105}
    \end{subfigure}
    \vspace{1em}
    \begin{subfigure}[t]{0.48\linewidth}
        \centering
        \includegraphics[width=\linewidth]{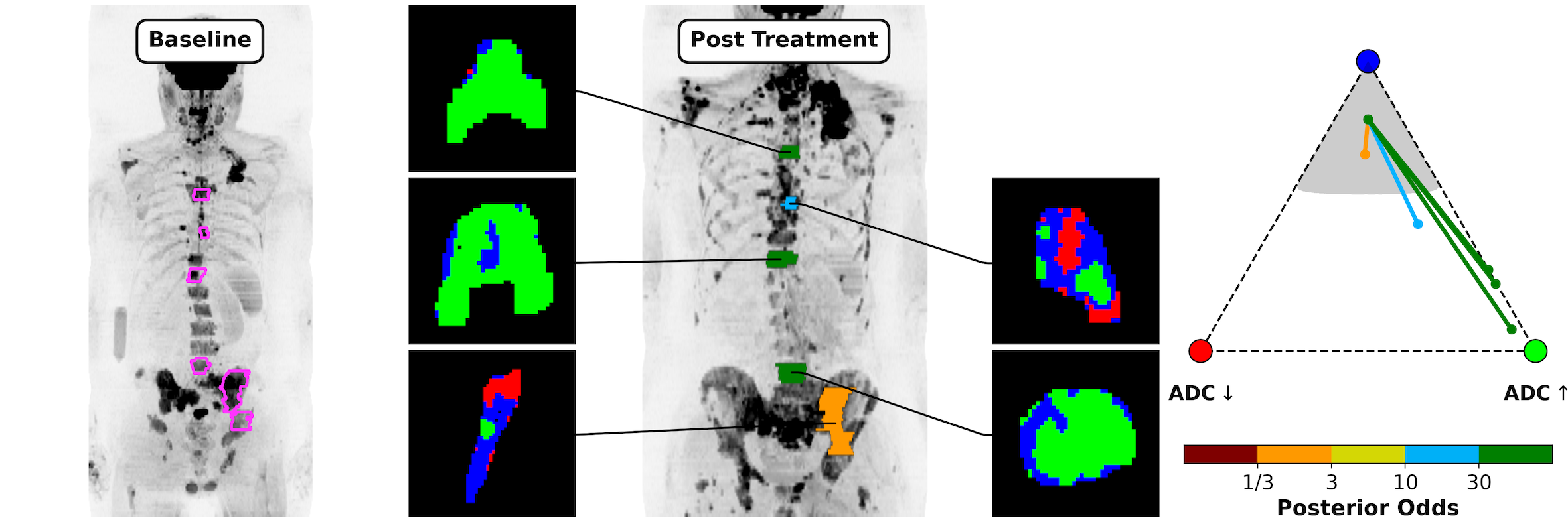}
        \caption{This patient (Study 2) demonstrated a general trend of increasing ADC across all observed lesions,
            indicating a positive respoonse to treatment. However, two lesions demonstrated less evidence of response,
            potenitally indicating a future resistance to treatment.}
        \label{fig:pat_example_1007}
    \end{subfigure}
    \hfill
    \begin{subfigure}[t]{0.48\linewidth}
        \centering
        \includegraphics[width=\linewidth]{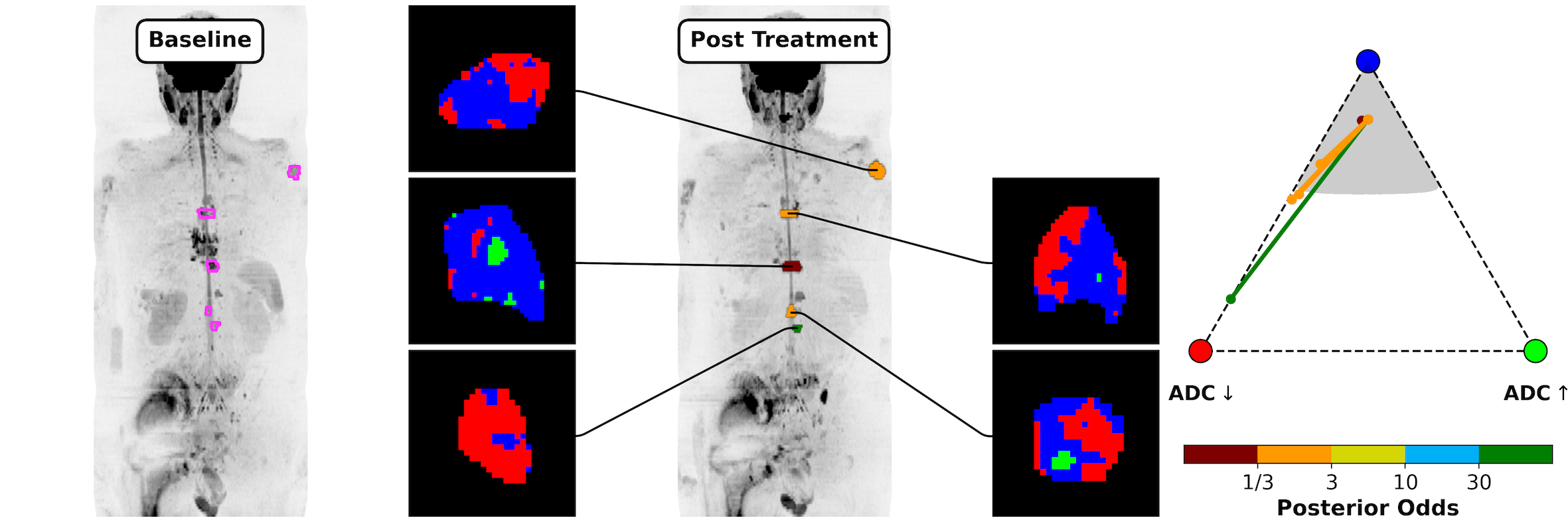}
        \caption{In this patient (Study 2), most tumours appear to reduce in ADC.
            However, only one of these is considered to be significant, perhaps due to the fact that lesions are
            small and thus measurement precision might be poor (as captured by the posterior odds).
            Retrospective evaluation of this patient by an expert radiologist confirmed that this drop in ADC
            was genuine, and not due to a return of yellow marrow following treatment.}
        \label{fig:pat_example_1010}
    \end{subfigure}
    \vspace{1em}
    \begin{subfigure}[t]{0.48\linewidth}
        \centering
        \includegraphics[width=\linewidth]{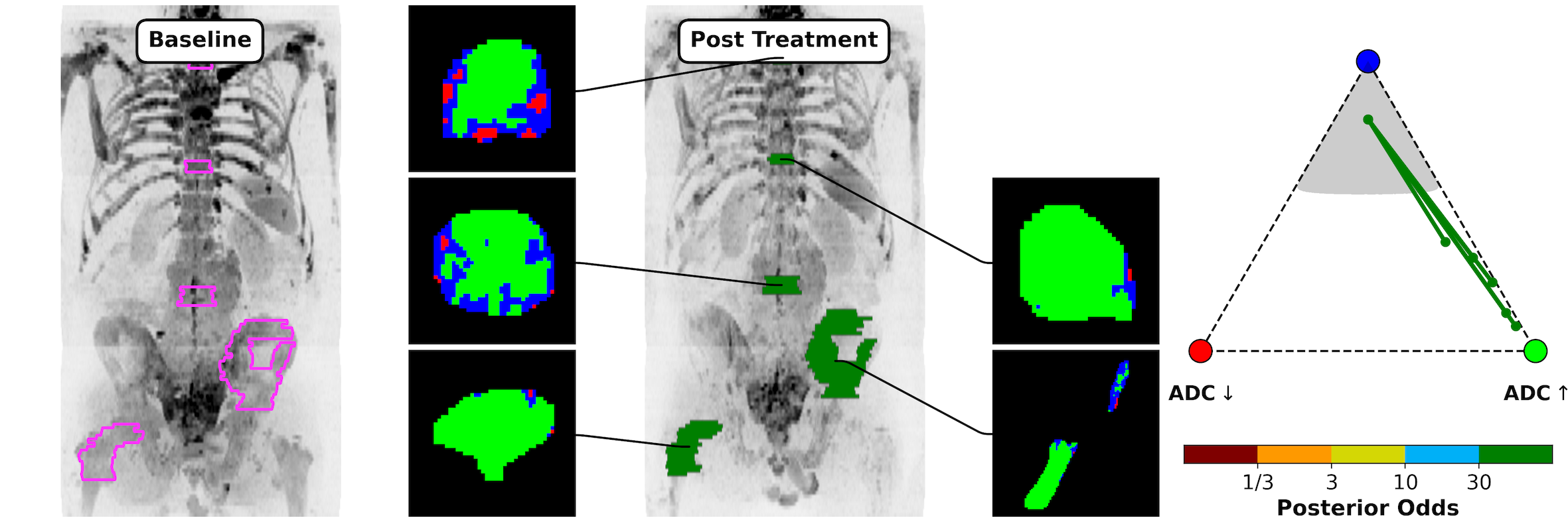}
        \caption{In this case we see a patient (Study 2) for whom all lesions changed significantly, with a general
            trend of increasing ADC.}
        \label{fig:pat_example_1013}
    \end{subfigure}

    \caption{Example results for five patient evaluated in this study. We have presented a \acrfull{mip} with lesions
        colour--coded according to the derived posterior odds.
        Habitat maps (central slice) are displayed as an overlay on the ADC map for each lesion, and a summary of
        inter-lesion changes can be observed in the triangular plots depciting the noramlised habitat simplexes in
        barycentric coordinates.
        All imaging data are presented with identical windowing settings.}
    \label{fig:pat_examples}
\end{figure}

These observations are further verified in figure~\ref{fig:study_summary} which demonstrates the changes occurring in
all lesions across both studies.
Whilst the majority of lesions demonstrate a significant increase in the proportion of voxels with high ADC value,
there are many that do not demonstrate any change (maroon colored lines) and several that significantly change in the
other direction (signalling an increase in the number of voxels with lower ADC).
Moreover, in Figure~\ref{fig:study_mixtures}, we present the posterior distributions for mixing weights ($\lambda$) for
each study.
In both cases, there is evidence that $\lambda$ is significantly below 1, illustrating presence of response
heterogeneity amongst individual lesions.
This is true for both median ADC and habitat biomarkers.

To compare the posterior odds derived from both median ADC and habitat biomarkers, a scatter plot of log-\acrshort{po}
for both markers is presented in figure~\ref{fig:lpo_comparison}.
There is a significant positive correlation between both, indicating that they may convey similar information.
However, this association is not perfect, indicating that they may indeed demonstrate different information in certain
lesions.

\subsection{Prior Sensitivity Analysis}
To evaluate the sensitivity of our model fit to the choice of prior, we conducted a prior sensitivity analysis, with
results presented in Figure\ref{fig:sensitivity_analysis} for Study 1 and Supplementary
Figure~\ref{fig:sensitivity_analysis_2} for Study 2.
In this analysis, we varied the prior widths for the ADC median model (\(\sigma_\mu\) and \(\gamma_\sigma\),
Table~\ref{tab:real_assumptions}) and for the habitat model (\(\gamma_\tau\), Table~\ref{tab:simplex_assumptions}).
The posterior distributions of the following parameters were estimated both with and without the inclusion of
post-treatment data: \(\sigma\), \(\sigma_0\), and \(\mu_0\) for median ADC biomarkers, and \(\tau\) for habitat
biomarkers.

As expected, we observed that narrow priors (\(\sigma_\mu \leq 0.1 \times 10^{-3} \, \text{mm}^2/\text{s}\)) influence
the inference of the estimated distributions.
However, less informative priors had a smaller impact on the posterior estimates, indicating greater robustness of the
model to these choices.
In general, the inclusion of post-treatment data had minimal effect on the posterior distributions.
An exception to this trend was observed in the estimation of measurement precision (\(\tau\)) for habitat biomarkers in
Study 1.
This discrepancy may suggest some misalignment between the assumed model and the acquired post-treatment data.
In such cases, we recommend relying on estimates derived solely from repeat baseline data, as we have done when
calculating credible intervals for novelty detection.

\begin{figure}[!htpb]
    \centering
    \includegraphics[width=0.7\linewidth]{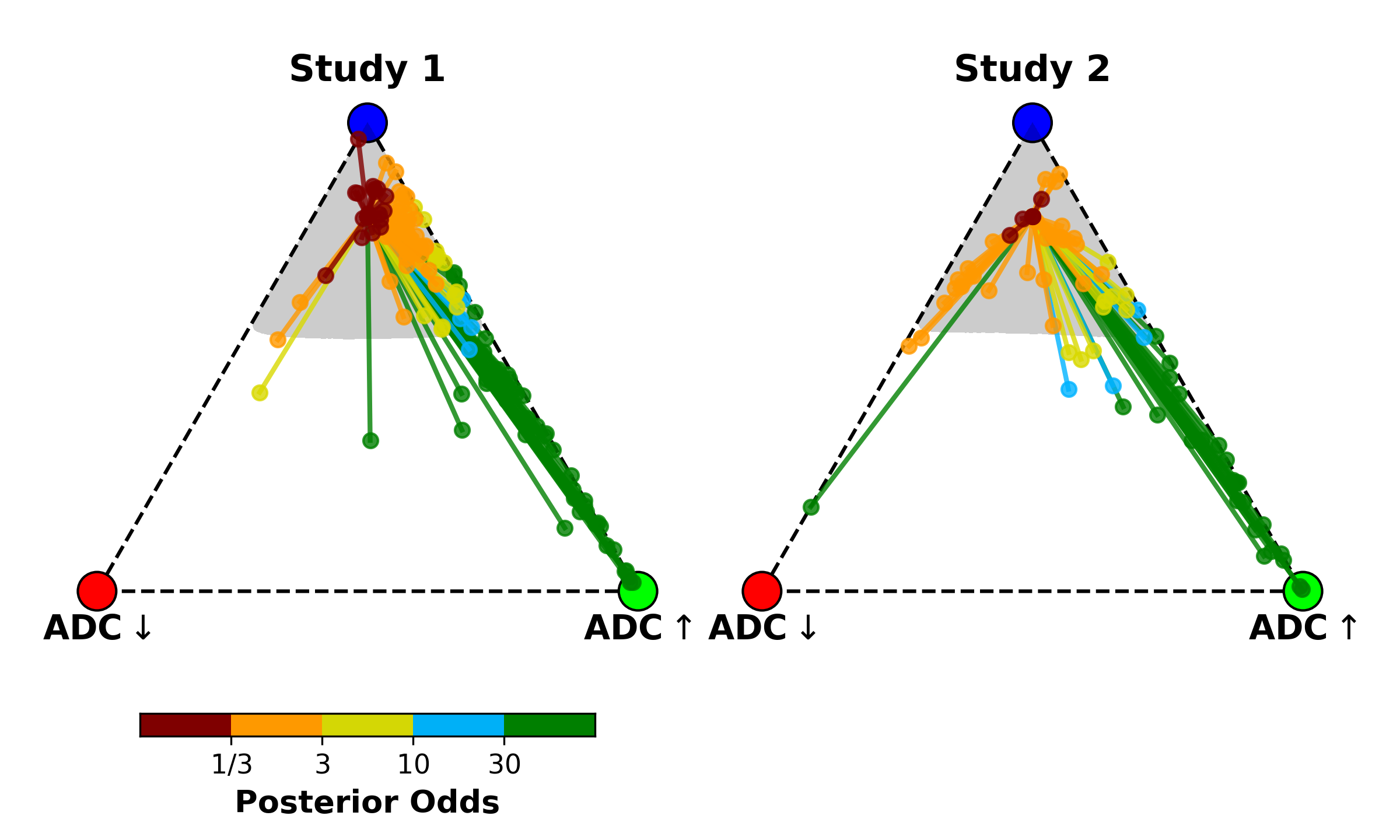}
    \caption{Barycenric coordinate plots of all lesions within both studies.
        For each lesion, its derived posterior odds is depcted as a color-coding, and the credible regions for
        identifying `anomalies' is shown as a grey shaded area.}
    \label{fig:study_summary}
\end{figure}

\begin{figure}[!htpb]
    \centering
    \includegraphics[width=0.7\linewidth]{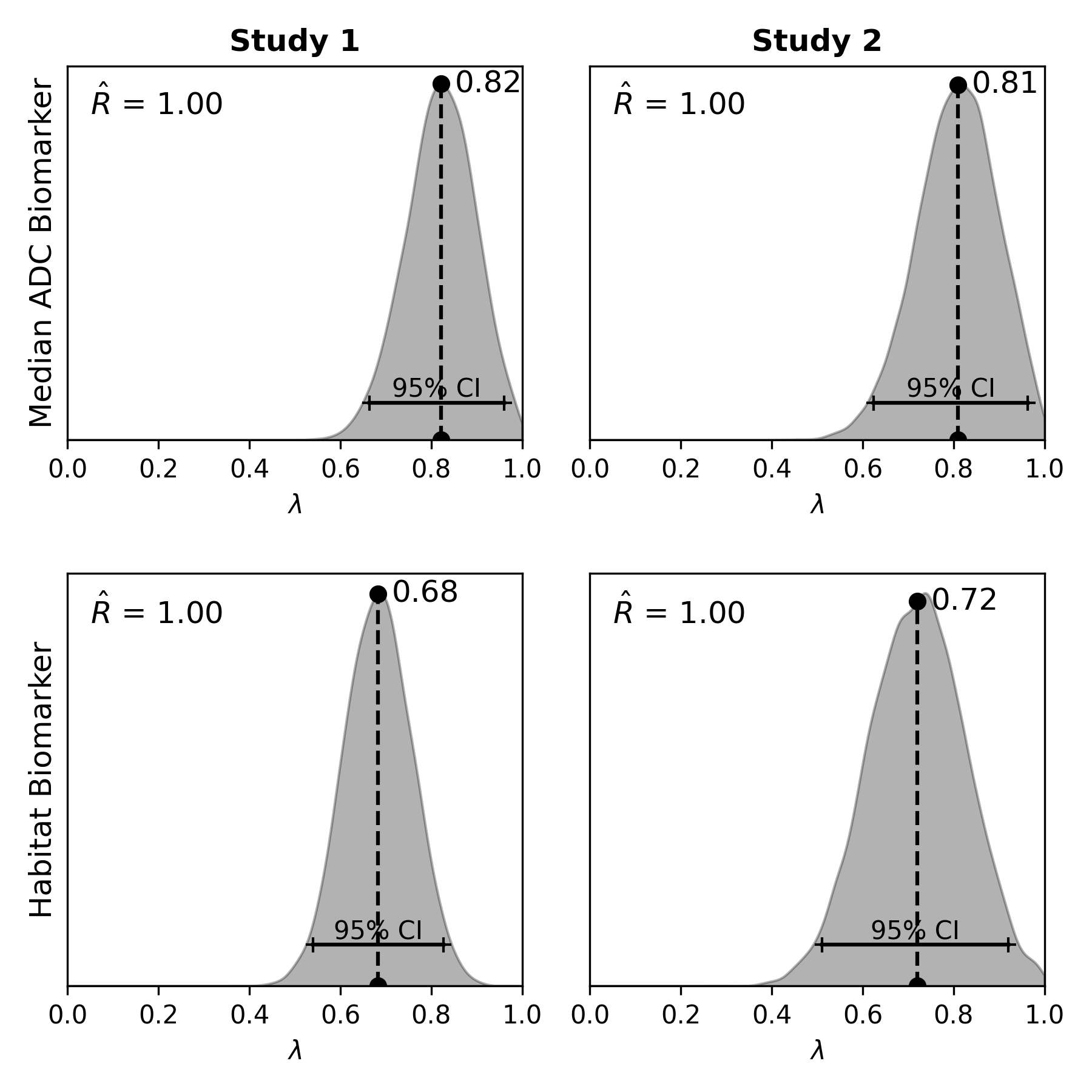}
    \caption{Kernel density plots of the posterior distributions for the mixing weight $\lambda$, shown alongside the
        median esimtate (vertical dashed line), and the Rubin-Gelman statsitic ($\hat{R}$).}
    \label{fig:study_mixtures}
\end{figure}

\begin{figure}[!htpb]
   \centering
   \includegraphics[width = 0.5\linewidth]{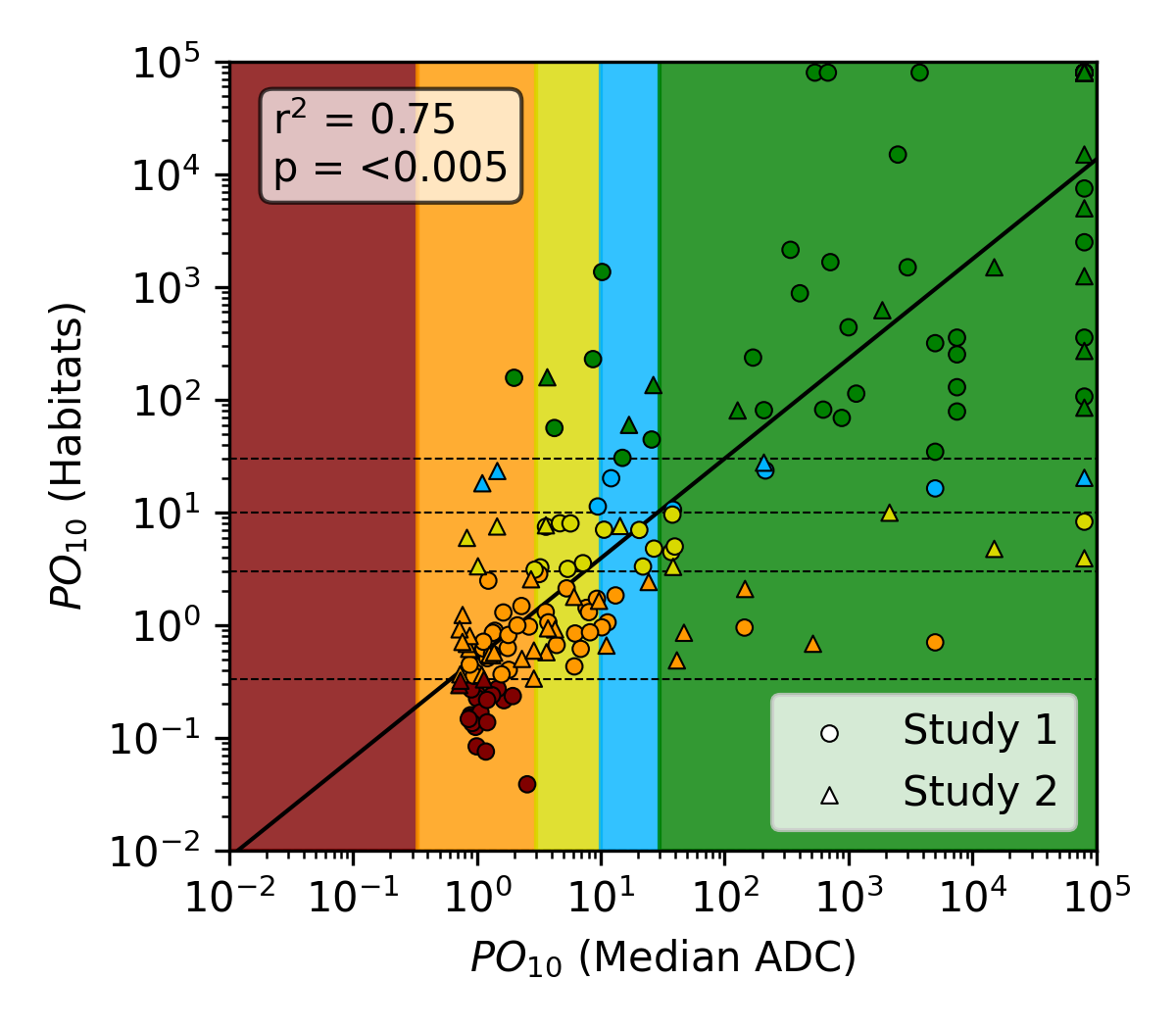}
   \caption{A comparison of log-\acrshort{po} derived for median ADC within each lesion and for habitats defined for
   each region. Note that our estimation method for \acrshort{po} can lead to inifite values. These were capped to
   80,000 and not included in estimation of $r^{2}$ or the p-value. Vertical colors represent the \acrshort{po}
   classifications defined in section~\ref{sec:results_patients} for the median ADC data, whilst the symbol colors
   represnt the \acrshort{po} for the habitat data. Data from studies 1 and 2 are seperated by symbol shape.}
   \label{fig:lpo_comparison}
\end{figure}

\begin{figure}[!htpb]
    \centering
    \includegraphics[width=0.75\linewidth]{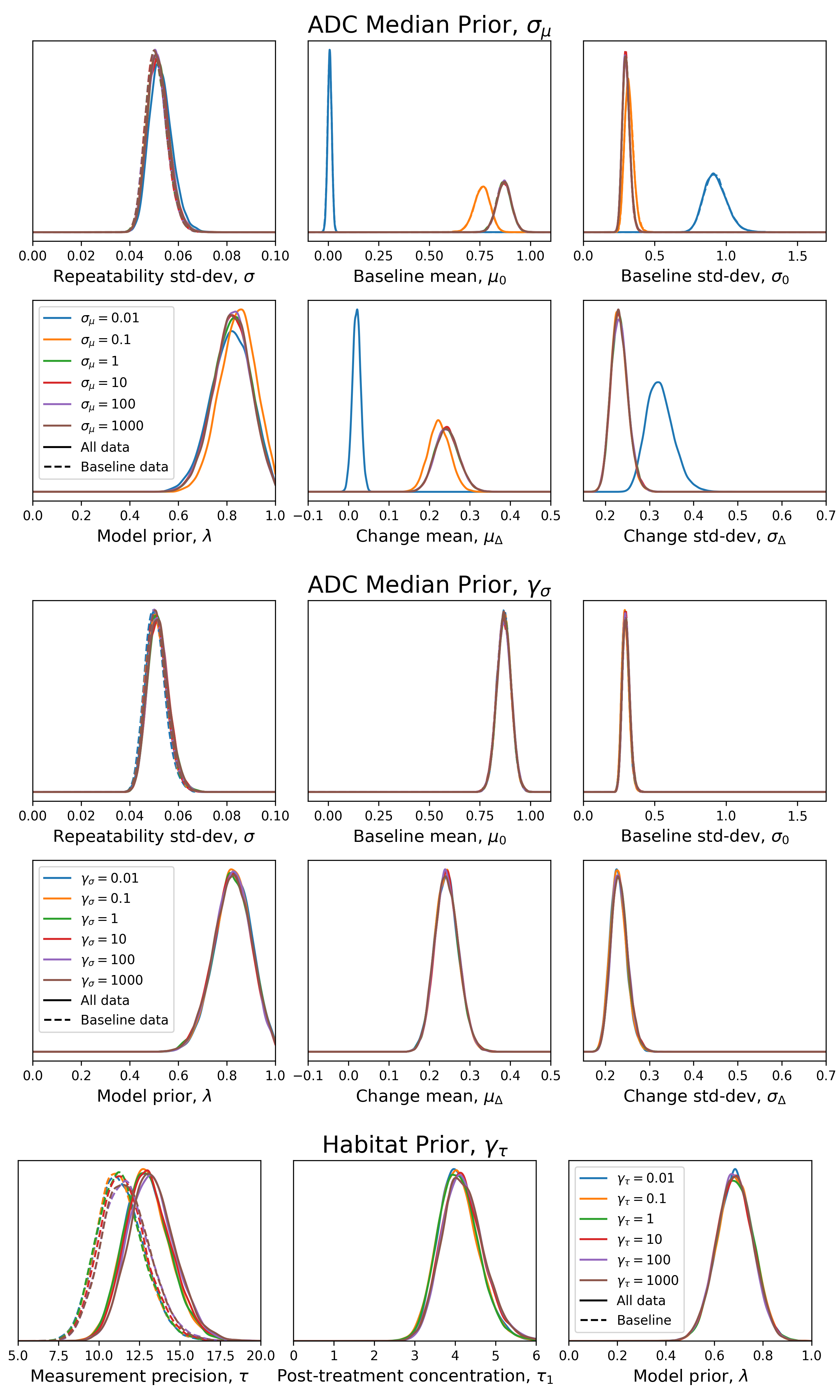}
    \caption{Prior sensitivity analysis for Study 1. The primary prior parameter influencing the posterior
    distributions was \(\sigma_\mu\), which defines the width of the zero-mean normal prior for \(\mu_0\) and
    \(\mu_\Delta\) in ADC median biomarkers. For habitat biomarkers (bottom row), the inclusion of post-treatment data
    in the model fit (solid curves) notably affected the estimation of measurement precision (\(\tau\)). For ADC median
    biomarkers, all parameters except \(\lambda\) are expressed in units of \(\times 10^{-3} \, \text{mm}^2/\text{s}\).}
    \label{fig:sensitivity_analysis}
\end{figure}

%% file: sections/discussion.tex
In this article, we aim to develop a novel framework for detecting heterogeneous tumour responses in clinical studies.
The application of Bayesian theory to assessing biomarker precision in cancer imaging is somewhat underutilized,
particularly given recent advancements in \acrfull{mcmc} algorithms, including \acrfull{hmc}, and their efficient
and generalized implementations for posterior parameter sampling.
We propose that  the detection of cancer treatment response can be framed as a decision between two competing models in
a biomarker study: one where the underlying true biomarker has changed (the alternative model) and one where it has not
the null model).
The goal of Bayesian analysis then is to quantify this choice by considering the precision of the measurement
device and any other parameters that describe the population distribution of the biomarkers.
We explore two Bayesian approaches to quantify the choice between these models: anomaly detection and Bayes factors
(or their close equivalent the \acrfull{po}).

Through theoretical analysis, we have demonstrated that anomaly detection closely parallels conventional practices in
response biomarker assessment.
In both approaches, the primary objective is to test and potentially reject the null hypothesis of no change in a
biomarker for a patient, based on measurements taken before and after an intervention (e.g., treatment).
This assessment is conducted at a predefined significance level, leveraging dedicated test-retest baseline biomarker
data (ideally in a similar population) to derive the parameters of the null model (including measurement precision) and
the observed post-treatment change for the patient.
Integrating this method into a Bayesian framework, however, exposes the underlying assumptions of the model that can be
overlooked and may not be valid in practice.
For instance, we illustrate that conventional Bland-Altman analysis often assumes an infinite possible range for the
biomarker, which may not be realistic—particularly in the case of quantitative imaging biomarkers, where measurements
are typically bounded by physical or biological constraints.
Bayesian analysis, facilitated by numerical \acrshort{mcmc} sampling, allows us to relax these assumptions.
It provides a comprehensive characterization of the uncertainties inherent in all estimated model parameters and,
perhaps most importantly, it enables the assessment of biomarkers that are not univariate variables, a topic we explore
extensively in this article.

We have also explored theoretical interpretations of Bayes factors in response biomarker studies.
Notably, we demonstrated that Bayes factors can incorporate expected post-treatment changes into the decision-making
process by integrating parameters from the alternative model.
A key advantage of using Bayes factors is their ability to help researchers evaluate whether the data provide sufficient
evidence to support the null model, allowing for the conclusion that there has been no change following the
intervention.
However, caution is warranted.
Unlike the null model, estimating parameters for the alternative model can be challenging due to the difficulty of
obtaining controlled experimental data.
Identifying lesions or subjects that have responded to an intervention is often complex and may, in fact, be the primary
role of the imaging biomarker itself.
To address this, we propose that Bayesian mixture modeling offers a robust solution.
This approach simultaneously estimates the parameters of the alternative model and quantifies posterior odds
(equivalent to Bayes factors when the prior probability for each data model is the same) for each subject in the study.

To illustrate the efficacy of these methods, we applied them in studies evaluating the response to two systemic
treatments for \acrfull{mcrpc} using \acrfull{wbdwi}.
These studies demonstrated that posterior odds estimates detect significant changes in individual lesions.
Furthermore, color-coding lesions according to their posterior odds offers a practical visualization method,
highlighting the heterogeneity of response both among patients and between lesions within a single patient.
Use of Bayesian mixture modeling enabled the quantification of this response heterogeneity, revealing a
consistent response rate of approximately 70\% across both studies.
Our framework incorporated habitat mapping from \acrshort{adc} maps directly into the decision-making process, providing
detailed visualizations of changing regions within the lesions.

It is crucial to emphasize that Bayesian techniques rely on the careful definition of the underlying models.
These models should be developed collaboratively with domain experts and those with expertise in probabilistic modeling
to ensure their validity and applicability.
Dedicated simulation studies, such as those conducted in this work, are invaluable for demonstrating that effective
parameter inference is achievable for well-defined models.
In addition, prior sensitivity analysis is a critical component of any Bayesian study.
This step ensures that the chosen prior distribution—unless specifically and carefully selected—does not unduly
influence or bias the final results, as highlighted by Depaoli et al.~\cite{depaoli2020importance}.
Moreover, it is crucial that the test-retest data used to determine parameters in the response model are representative
of the data to which the model will be applied.
The size of the test-retest cohorts plays a significant role in influencing the accuracy of biomarker precision
estimates, which, in turn, directly impacts the reliability and validity of model comparisons.

A slight limitation of this study is that we have treated spatially isolated lesions as independent.
This approach was necessary due to the relatively small number of lesions delineated per patient in our dataset.
One of the significant advantages of Bayesian methods in biomarker studies is their natural ability to incorporate
inter- and intra-patient correlations in biomarker changes through hierarchical modelling~\cite{donners2024inter}.
Looking ahead, for cases where a large number of lesions are delineated within a single patient, it would be feasible to
model a per-patient response heterogeneity index, $\lambda_{i}$ for $i \in {1, \dots, N}$ patients.
This could include a complete characterization of the posterior distribution for these parameters, offering a more
comprehensive understanding of intra-patient response variability.
Applying these techniques in the clinic will eventually require automatic delineation of multiple \acrfull{rois}
across various time points in whole-body imaging.
This task, necessitating an experienced radiologist, can be labor-intensive.
There are promising areas of research, such as the use of \acrfull{ai} to automate this process, potentially reducing
costs and facilitating clinical adoption~\cite{candito2024deep, ceranka2020computer}.

From a technical perspective, we acknowledge potential improvements in our approach to calculating Bayes factors.
Techniques like bridge sampling~\cite{gronau_bridge} have shown greater stability in computing Bayes factors compared to
other methods.
However, current implementations primarily support complete model comparisons and do not extend to estimating individual
subjects within a study via mixture models, as we have explored.
Future work could investigate adapting bridge sampling for this purpose.
Moreover, our study does not include biological validation of the techniques, as our primary aim was to establish a
theoretical framework for assessing the repeatability of tumor response heterogeneity measurements.
Future research might validate these methods through targeted biopsy or surgical specimen
analysis~\cite{Donners:2022vn}, linking our results with promising directions in digital
histopathology~\cite{Zhu:2023ws} and spatial gene expression profiling~\cite{Brady:2021tb}.
Ultimately, assessing the clinical utility of our approach will require larger patient studies with access to relevant
clinical end-points such as progression-free survival or overall survival~\cite{Perez-Lopez:2017wk,perez_volume}.

It is important to recognize the breadth of machine learning tools available for novelty detection that could be
explored for tumor response analysis.
These tools range from isolation forests~\cite{liu2008isolation} and one-class support vector
machines~\cite{Scholkopf:2001tr} to deep learning approaches like autoencoders~\cite{pidhorskyi2018generative}
and Siamese networks~\cite{Liu:2022wh}.
A significant advantage of these methods is their ability to operate without predefined data models.
However, they generally require access to large datasets, which may not always be available, particularly for
less common cancer types or specific imaging modalities.
Additionally, unlike Bayesian methods, these approaches generally do not inherently account for uncertainty in their
decision-making processes, an essential feature for any clinical decision-support tool.
An additional research area of interest here is the application of more informative prior distributions to model
parameters.
While we have opted for uninformative priors to avoid biasing our results, Bayesian techniques inherently
facilitate adaptive learning.
More informative priors could effectively integrate data from previous similar studies into the decision-making process,
enhancing the accuracy and relevance of the models.
Finally, we note that although have focussed on the application of techniques to a single post-treatment scans, these
methods can be adapted to biomarker studies with multiple post-treatment imaging studies for longitudinal assessment.
However, this will depend on the ability to integrate these methods into a medical device that helps
clinicians understand and visualize the complex changes occurring over time.

Bayesian methods in response biomarker studies present significant advantages over traditional statistical
approaches.
One of their most crucial benefits is the ability to analyze imaging features beyond normally-distributed biomarkers, a
topic we have extensively explored in this article.
This is particularly evident in habitat imaging which is emerging as a valuable technique for performing
``virtual biopsies'', and links the intra-tumoral heterogeneity captured from radiological imaging to the underlying
pathological and biological mechanisms.
Accurate characterization of habitats can be used to guide radiotherapy treatment to target the most aggressive tumour
areas, guide surgical biopsies and monitor treatment response~\cite{paverd2024radiology}.
We anticipate that frameworks such as ours are essential to assess the validity of such approaches and support their
integration in the clinical setting by providing clinicians with a clear and interpretable summary of which tumors in
each patient are showing significant changes post-treatment.
This could offer an early opportunity to modify or intensify treatments in areas showing resistance, potentially
improving patient outcomes and further advancing personalized medicine.

%% file: sections/supplementary-material.tex
\subsection{Mathematical derivations}\label{subsec:theory_derivation}
\begin{theorem}\label{theorem:conditional_posterior}
If
\[
y \sim \mathcal{N}(x_{0}, \sigma) \qquad y_{0} \sim \mathcal{N}(x_{0}, \sigma) \qquad x_{0} \sim \mathcal{N}(\mu_{0}, \sigma_{0})
\]
then 
\begin{equation}\label{eq:y_cond}
p(y | y_{0}, \sigma, \mu_{0}, \sigma_{0}, M_{0}) = \mathcal{N}\left(\mu^{*}, \sigma^{*} \right)
\end{equation}
where 
\begin{gather*}
\mu^{*} = ICC\cdot y_{0} + (1-ICC)\cdot \mu_{0} \quad \sigma^{*} = \sqrt{1+ICC}\cdot \sigma \\
ICC = \frac{\sigma_{0}^{2}}{\sigma^{2} + \sigma_{0}^{2}}
\end{gather*}
\end{theorem}
\begin{proof}
We begin by substituting the measurement error term, $\sigma$ for \emph{precision}, $\tau = \frac{1}{\sigma^{2}}$, such that
\begin{align}
p(y_{0}|x_{0}, \tau) &\propto \exp\left\{-\frac{\tau}{2}(y_{0}-x_{0})^{2}\right\} \notag \\
p(y|x_{0}, \tau) &\propto \exp\left\{-\frac{\tau}{2}(y-x_{0})^{2}\right\} \notag
\end{align}
Using Bayes' theorem we derive the conditional distribution of $y$ given $y_{0}$, $\tau$, and population parameters
$\mu_{0}$ and $\tau_{0} = \frac{1}{\sigma_{0}^{2}}$ by marginalising over the common true biomarker value $x_{0}$.
\begin{align}
p(y|y_{0},\tau, \mu_{0}, \tau_{0}) &\propto p(y_{0}, y|\tau, \mu_{0}, \tau_{0})\notag\\
&= \int\limits_{-\infty}^{\infty}p(y_{0}, y | x_{0}, \tau)p(x_{0}|\mu_{0}, \tau_{0})\text{d}x_{0} \notag \\
&= \int\limits_{-\infty}^{\infty}p(y_{0} | x_{0}, \tau)p(y | x_{0}, \tau)p(x_{0}|\mu_{0}, \tau_{0})\text{d}x_{0} \notag
\end{align}
Therefore
\begin{align}
p(y|&y_{0}, \tau, \mu_{0}, \tau_{0}) \notag \\
&\propto \int\limits_{-\infty}^{\infty} \exp\left\{-\frac{\tau}{2}(y_{0} - x_{0})^{2}\right\}\exp\left\{-\frac{\tau}{2}(y - x_{0})^{2}\right\}\notag\\
&\qquad\qquad \times\exp\left\{-\frac{\tau_{0}}{2}(x_{0} - \mu_{0})^{2}\right\} \text{d}{x_{0}} \notag \\
&\propto \int\limits_{-\infty}^{\infty}\exp\left\{ -\frac{2\tau + \tau_{0}}{2}\left(x_{0}^{2} - 2x_{0}\frac{\tau(y_{0}+y) + \tau_{0}\mu_{0}}{2\tau + \tau_{0}}\right)\right\}\text{d}{x_{0}} \notag \\
& \qquad\qquad \times \exp\left\{-\frac{\tau}{2}y^{2} \right\}\notag\\
&= \int\limits_{-\infty}^{\infty}\exp\left\{-\frac{2\tau + \tau_{0}}{2}\left(x_{0} - \frac{\tau(y_{0} + y) + \tau_{0}\mu_{0}}{2\tau + \tau_{0}}\right)^{2}\right\}\text{d}x \notag \\
&\qquad\qquad \times \exp\left\{-\frac{\tau}{2}y^{2} \right\}\exp\left\{\frac{(\tau(y_{0}+y)+\tau_{0}\mu_{0})^{2}}{2(2\tau+\tau_{0})}\right\}\notag \\
&\propto \exp\left\{-\frac{\tau}{2} y^{2} + \frac{\tau^{2}(y_{0}+y)^{2} + 2\tau_{0}\mu_{0}\tau y}{2(2\tau + \tau_{0})}\right\} \notag \\
&\propto \exp\left\{-\frac{\tau}{2(2\tau+\tau_{0})}\left(y^{2}(\tau + \tau_{0}) - 2y(y_{0}\tau+\tau_{0}\mu_{0}) \right)\right\}\notag\\
&\propto\exp\left\{-\frac{1}{2}\frac{\tau(\tau + \tau_{0})}{2\tau + \tau_{0}}\left(y - \frac{y_{0}\tau + \tau_{0}\mu_{0}}{\tau+\tau_{0}}\right)^{2}\right\} \notag
\end{align}
Now making the substitutions
\begin{align}
\mu^{*} &= \frac{y_{0}\tau + \tau_{0}\mu_{0}}{\tau+\tau_{0}}, \qquad \tau^{*} = \frac{\tau(\tau + \tau_{0})}{2\tau + \tau_{0}} \notag
\end{align}
we deduce that
\[
p(y|y_{0}, \tau, \mu_{0}, \tau_{0}) = \mathcal{N}\left(\mu^{*}, \frac{1}{\sqrt{\tau_{*}}}\right)
\]
Lastly, we have that 
\begin{align}
\mu^{*} &= \frac{\tau}{\tau + \tau_{0}}\cdot y_{0} + \left(1 - \frac{\tau}{\tau + \tau_{0}}\right)\cdot \mu_{0}\notag \\
&= \frac{\sigma_{0}^{2}}{\sigma^{2} + \sigma^{2}_{0}}\cdot y_{0} + \left(1 - \frac{\sigma_{0}^{2}}{\sigma^{2} + \sigma^{2}_{0}}\right)\cdot  \mu_{0} \notag \\
&=ICC\cdot y_{0} + \left(1 - ICC\right)\cdot \mu_{0} \notag
\end{align}
and
\begin{align}
\sigma^{*} &= \sqrt{\frac{1}{\tau^{*}}}  = \sqrt{\frac{2\tau + \tau_{0}}{\tau(\tau+\tau_{0})}}\notag \\
&= \sqrt{\frac{1}{\tau}\left( 2\cdot \frac{\tau}{\tau + \tau_{0}} + \left(1 - \frac{\tau}{\tau + \tau_{0}}\right)\right)} \notag \\
&= \sqrt{\sigma^{2}\left(\frac{\sigma^{2}_{0}}{\sigma_{0}^{2} + \sigma^{2}} + 1 \right)} \notag \\
&= \sqrt{1 + ICC}\cdot \sigma \notag
\end{align}
\end{proof}

\begin{corollary}\label{theorem:studentt}
Under the assumption that $ICC \rightarrow 1$, the marginal distribution of the change in biomarker measurement,
$d = y - y_{0}$, given a dataset of double-baseline measurements,
$\mathbf{d}_{b} = \{d_{b1}, d_{b2}, \dots, d_{bN_{b}}\}$ where $d_{bj} = y_{bj2} - y_{bj1}$, is a
Student-t distribution:
\[
p(d | \mathbf{d}_{b}, M_{0}) = \textnormal{St} \left(\sqrt{d^{2}/\overline{d^{2}_{b}}}, N_{b} \right)
\]
where $\overline{d^{2}_{b}} = \frac{1}{N_{b}}\sum_{j}d_{bj}^{2}$
\end{corollary}
\begin{proof}
Again to make notation simpler we use the precision of the measurement system $\tau = \frac{1}{\sigma^{2}}$ such that
according to equation~\ref{eq:y_cond} the distribution of repeat baseline difference is given by
\[
p(d_{bj}|\tau) = \mathcal{N}\left(0, \sqrt{\frac{2}{\tau}}\right)
\]
when $ICC = 1$.
To determine the posterior probability of $\tau$ we use Bayes' theorem such that
\begin{align}
p(\tau|\mathbf{d}_{b}) &\propto p(\tau)p(\mathbf{d}_{b}|\tau)\notag \\
&= p(\tau)\prod_{j=1}^{N_{b}}p(d_{bj}|\tau) \notag
\end{align}
Using a non-informative prior $p(\tau)\propto \tau^{-1}$ we obtain:
\begin{align}
p(\tau|\mathbf{d}_{b}) \propto \tau^{N_{b}/2-1}\exp\left\{-\frac{\tau}{2}\sum\limits_{j=1}^{N_{b}}d_{bj}^{2}\right\} \notag
\end{align}
Marginalizing over $\tau$ thus provides the distribution of the change in measurement after treatment conditioned on the repeatability data:
\begin{align}
p(d|\mathbf{d}_{b}) &= \int\limits_{0}^{\infty}p(d|\tau)p(\tau|\mathbf{d}_{b})\text{d}\tau\notag \\
&\propto \int\limits_{0}^{\infty}\tau^{1/2}\exp\left\{-\frac{\tau}{2}d^{2}\right\}\tau^{N_{b}/2-1}\exp\left\{-\frac{\tau}{2}\sum\limits_{j=1}^{N_{b}}d_{bj}^{2}\right\} \text{d}\tau\notag \\
&\propto  \int\limits_{0}^{\infty}\tau^{(N_{b}-1)/2}\exp\left\{-\frac{\tau}{2}\left(d^{2} + \sum\limits_{j=1}^{N_{b}}d_{bj}^{2}\right)\right\} \text{d}\tau\notag
\end{align}
Using the result
\[
\int\limits_{0}^{\infty}\tau^{A}e^{-B\tau} \text{d}\tau = \frac{\Gamma(A+1)}{B^{(A+1)}}
\]
where the necessary conditions $A > -1$ and $\text{Re}(B) > 0$ are guaranteed by the data, we arrive at 
\begin{align}
p(d|\mathbf{d}_{b}) &\propto \left(d^{2} + \sum\limits_{j=1}^{N_{b}}d_{bj}^{2}\right)^{-\frac{N_{b}+1}{2}} \notag \\
&=  \left(d^{2} + N_{b}\overline{d^{2}_{b}}\right)^{-\frac{N_{b}+1}{2}} \notag \\
&\propto \left(\frac{d^{2}/\overline{d^{2}_{b}}}{N_{b}} +1 \right)^{-\frac{N_{b}+1}{2}} \notag
\end{align}
where
\[
\overline{d^{2}_{b}} = \frac{1}{N_{b}}\sum\limits_{i=1}^{N_{b}}d_{bj}^{2}
\]
We notice from its functional form that this is a Student t distribution
\[
\text{St}(t;\nu) = \frac{\Gamma\left(\frac{\nu+1}{2}\right)}{\sqrt{\nu\pi}\Gamma\left(\frac{\nu}{2}\right)}\left(\frac{t^{2}}{\nu}+1\right)^{-\frac{\nu+1}{2}}
\]
with $t^{2} = d^{2}/\overline{d^{2}_{b}}$ and $\nu = N_{b}$.
\end{proof}

\begin{theorem}\label{theorem:bf_gaussian_proof}
The natural logarithm of the Bayes factor for a measured post-treatment change $d_{y} = y_{1} - y_{0}$, in the case of
normally distributed variables with \textbf{known} model $M_{1}$ parameters $\mu_{\Delta}$, $\sigma_{\Delta}$, and $\sigma$ is
\[
\ln (BF_{10}) = \frac{1}{2}\left(\ln (1-\eta) - (1-\eta)\left(\frac{d_{y} - \mu_{\Delta}}{\sqrt{2}\sigma} \right)^{2} + \left( \frac{d_{y}}{\sqrt{2}\sigma}\right)^{2}\right)
\]
where
\[
\eta = \frac{\sigma_{\Delta}^{2}}{\sigma_{\Delta}^{2} + 2\sigma^{2}}
\]
\end{theorem}
\begin{proof}
We use the approach of Savage and Dickey~\cite{Wagenmakers:2010tp}.
Firstly, we make the substitutions $\tau = \frac{1}{\sigma^{2}}$ and $\tau_{\Delta} = \frac{1}{\sigma_{\Delta}^{2}}$.
Next we note that model $M_{0}$ is a special case of $M_{1}$ where $d_{x} = x_{1} - x_{0} = 0$, such that
\begin{align}
BF_{10} &= \frac{p(d_{y}|M_{1})}{p(d_{y}|d_{x}=0, M_{1})} \notag \\
&= \frac{\int p(d_{y}| \tau, \mu_{\Delta}, \tau_{\Delta}, M_{1})p(\tau, \mu_{\Delta}, \tau_{\Delta}|M_{1}) \text{d}\tau\text{d}\mu_{\Delta}\text{d}\tau_{\Delta}}{\int p(d_{y}|d_{x}=0, \tau, \mu_{\Delta}, \tau_{\Delta}, M_{1})p(\tau, \mu_{\Delta}, \tau_{\Delta}|M_{1}) \text{d}\tau\text{d}\mu_{\Delta}\text{d}\tau_{\Delta}} \notag
\end{align}
However, if we can assume that parameters $\tau$, $\mu_{\Delta}$, and $\tau_{\Delta}$ are known, then we can instead
write
\begin{align}
BF_{10} &=  \frac{p(d_{y}|\tau,\mu_{\Delta}, \tau_{\Delta}, M_{1})}{p(d_{y}|d_{x}=0, \tau, \mu_{\Delta}, \tau_{\Delta}, M_{1})} \notag
\end{align}
which by Bayes' theorem becomes
\begin{align}
BF_{10} &= \frac{p(d_{x} = 0|\mu_{\Delta}, \tau_{\Delta}, M_{1})p(d_{y}|\tau, \mu_{\Delta}, \tau_{\Delta}, M_{1})}{p(d_{x}=0|d_{y}, \tau,  \mu_{\Delta}, \tau_{\Delta}, M_{1})p(d_{y}|\tau, \mu_{\Delta}, \tau_{\Delta}, M_{1})} \notag\\
&=\frac{p(d_{x} = 0|\mu_{\Delta}, \tau_{\Delta}, M_{1})}{p(d_{x}=0|d_{y}, \tau, \mu_{\Delta}, \tau_{\Delta}, M_{1})}\label{eq:conitional_bayes_factor}
\end{align}
Under the assumption that variables are real and normally distributed we have that
$p(d_{x}|\mu_{\Delta}, \tau_{\Delta}, M_{1}) = \mathcal{N}\left(\mu_{\Delta}, 1/\sqrt{\tau_{\Delta}}\right)$ and so to
find the numerator of equation~\ref{eq:conitional_bayes_factor} we write
\[
p(d_{x} = 0|\mu_{\Delta}, \tau_{\Delta}, M_{1}) = \sqrt{\frac{\tau_{\Delta}}{2\pi}}\exp\left\{-\frac{1}{2}\tau_{\Delta}\mu_{\Delta}^{2} \right\}
\]
Next we identify
\[
p(d_{x} | d_{y}, \tau, \mu_{\Delta}, \tau_{\Delta}, M_{1}) \propto p(d_{y} | d_{x}, \tau, M_{1})p(d_{x} |\mu_{\Delta}, \tau_{\Delta}, M_{1})
\]
Using the result $p(d_{y} | d_{x}, \tau, M_{1}) = \mathcal{N}\left(d_{x}, \sqrt{2 / \tau}\right)$ we obtain
\begin{align}
p(d_{x} | d_{y}, &\tau, \mu_{\Delta}, \tau_{\Delta}, M_{1})\notag\\
&\propto \exp\left\{-\frac{\tau}{4}\left(d_{y} - d_{x}\right)^{2} \right\}\exp\left\{-\frac{\tau_{\Delta}}{2}\left(d_{x} - \mu_{\Delta}\right)^{2} \right\} \notag \\
&\propto \exp\left\{-\frac{1}{4}(\tau + 2\tau_{\Delta})d_{x}^{2} + \frac{1}{2}\left(\tau d_{y} + 2\tau_{\Delta}\mu_{\Delta}\right)d_{x}\right\} \notag \\
&\propto \exp\left\{-\frac{1}{4}(\tau + 2\tau_{\Delta})\left( d_{x} - \frac{\tau d_{y} + 2\tau_{\Delta}\mu_{\Delta}}{\tau + 2\tau_{\Delta}} \right)^{2}\right\} \notag
\end{align}
which we recognise as another normal distribution:
\[
p(d_{x} | d_{y}, \tau, \mu_{\Delta}, \tau_{\Delta}, M_{1}) = \mathcal{N}\left(\frac{\tau d_{y} + 2\tau_{\Delta}\mu_{\Delta}}{\tau + 2\tau_{\Delta}}, \sqrt{\frac{2}{\tau + 2 \tau_{\Delta}}} \right)
\]
such that
\begin{align}
p(d_{x} = 0 | &d_{y}, \tau, \mu_{\Delta}, \tau_{\Delta}, M_{1}) =\notag \\
& \sqrt{\frac{\tau + 2\tau_{\Delta}}{4\pi}}\exp\left\{-\frac{1}{4}\frac{(\tau d_{y} + 2\tau_{\Delta}\mu_{\Delta})^{2}}{\tau + 2\tau_{\Delta}} \right\} \notag
\end{align}
Finally the Bayes Factor is derived as 
\begin{align}
BF_{10} &= \frac{p(d_{x} = 0|\mu_{\Delta}, \tau_{\Delta}, M_{1})}{p(d_{x}=0|d_{y}, \tau, \mu_{\Delta}, \tau_{\Delta}, M_{1})}\notag \\
& = \sqrt{\frac{2\tau_{\Delta}}{\tau + 2\tau_{\Delta}}}\exp\left\{\frac{1}{4}\left(\frac{(\tau d_{y} + 2\tau_{\Delta}\mu_{\Delta})^{2}}{\tau + 2\tau_{\Delta}} - 2\tau_{\Delta}\mu_{\Delta}^{2} \right) \right\} \notag \\
& = \sqrt{\frac{2\tau_{\Delta}}{\tau + 2\tau_{\Delta}}}\exp\left\{\frac{1}{4}\frac{\tau^{2}d_{y}^{2} + 4\tau d_{y}\tau_{\Delta}\mu_{\Delta} - 2\tau\tau_{\Delta}\mu_{\Delta}^{2}}{\tau + 2\tau_{\Delta}} \right\} \notag \\
& = \sqrt{\frac{2\tau_{\Delta}}{\tau + 2\tau_{\Delta}}}\exp\left\{-\frac{\tau}{4}\frac{2\tau_{\Delta}}{\tau + 2\tau_{\Delta}}(d_{y} - \mu_{\Delta})^{2} \right\}\exp\left\{\frac{1}{4}\tau d_{y}^{2} \right\} \notag
\end{align}
Replacing precision terms with their respective error terms we have
\begin{align}
BF_{10} = \sqrt{1 - \eta}\exp\left\{-\frac{1}{2}(1 - \eta)\left(\frac{d_{y} - \mu_{\Delta}}{\sqrt{2}\sigma}\right)^{2} \right\}\exp\left\{\frac{1}{2}\left(\frac{d_{y}}{\sqrt{2}\sigma}\right)^{2}\right\}  \notag
\end{align}
where 
\[
\eta = \frac{\sigma_{\Delta}^{2}}{\sigma_{\Delta}^{2} + 2\sigma^{2}}
\]
which in log form becomes
\[
\ln (BF_{10}) = \frac{1}{2}\left(\ln (1-\eta) - (1-\eta)\left(\frac{d_{y} - \mu_{\Delta}}{\sqrt{2}\sigma} \right)^{2} + \left( \frac{d_{y}}{\sqrt{2}\sigma}\right)^{2}\right)
\]
\end{proof}

\begin{corollary}\label{theorem:expected_bf}
The expected values of equation~\ref{eq:logbf_normal} under the assumption of models $M_{0}$ or $M_{1}$ are
\begin{align}
\mathbb{E}_{d_{y}|M_{0}}\{\ln (BF_{10})\} &= \frac{1}{2} \left(\ln(1-\eta) + \eta - (1-\eta)\xi^{2} \right) \notag \\
\mathbb{E}_{d_{y}|M_{1}}\{\ln (BF_{10})\} &= \frac{1}{2} \left(\ln(1-\eta) + \frac{\eta}{1-\eta} + \xi^{2} \right) \notag
\end{align}
where $\xi = |\mu_{\Delta}|/\sqrt{2}\sigma$.
\end{corollary}
\begin{proof}
If we then investigate the expectation of this value over some distribution of $d_{y}$:
\begin{align}
\mathbb{E}&\left\{ \ln (BF_{10})\right\}\notag\\
&= \frac{1}{2}\left(\ln (1-\eta) - \frac{1-\eta}{2\sigma^{2}}\mathbb{E}\left\{\left(d_{y} - \mu_{\Delta}\right)^{2}\right\} + \frac{1}{2\sigma^{2}}\mathbb{E}\left\{d_{y}^{2}\right\}\right) \notag \\
&= \frac{1}{2}\left(\ln (1-\eta) + \frac{\eta}{2\sigma^{2}}\mathbb{E}\left\{d_{y}^{2}\right\} + \frac{1 - \eta}{2\sigma_{2}}\mu_{\Delta}\left(2\mathbb{E}\{d_{y}\} - \mu_{\Delta}\right)\right) \notag
\end{align}
Considering that 
\begin{align}
p(d_{y}|&\sigma, \mu_{\Delta}, \sigma_{\Delta}, M_{0}) = \mathcal{N}\left(0, \sqrt{2}\sigma\right)\notag\\
&\Rightarrow\ \mathbb{E}_{d_{y}|M_{0}}(d_{y}) = 0 \quad \text{and} \quad \mathbb{E}_{d_{y}|M_{0}}(d^{2}_{y}) = 2\sigma^{2}  \notag
\end{align}
and
\begin{align}
p(d_{y}|&\sigma, \mu_{\Delta}, \sigma_{\Delta}, M_{1}) = \mathcal{N}\left(\mu_{\Delta}, \sqrt{2\sigma^{2} + \sigma_{\Delta}^{2}}\right)\notag\\
&\Rightarrow\ \mathbb{E}_{d_{y}|M_{1}}(d_{y}) = \mu_{\Delta} \quad \text{and} \quad \mathbb{E}_{d_{y}|M_{1}}(d^{2}_{y}) = \mu_{\Delta}^{2} + 2\sigma^{2} + \sigma_{\Delta}^{2}  \notag
\end{align}
we have that 
\begin{align}
\mathbb{E}_{d_{y}|M_{0}}\{\ln (BF_{10})\} &= \frac{1}{2} \left(\ln(1-\eta) + \eta - (1-\eta)\xi^{2} \right) \notag \\
\mathbb{E}_{d_{y}|M_{1}}\{\ln (BF_{10})\} &= \frac{1}{2} \left(\ln(1-\eta) + \frac{\eta}{1-\eta} + \xi^{2} \right) \notag
\end{align}
\end{proof}

\newpage
\subsection{Sensitivity Analysis Results for Study 2}
\begin{figure}[htp]
    \centering
    \includegraphics[width=0.5\linewidth]{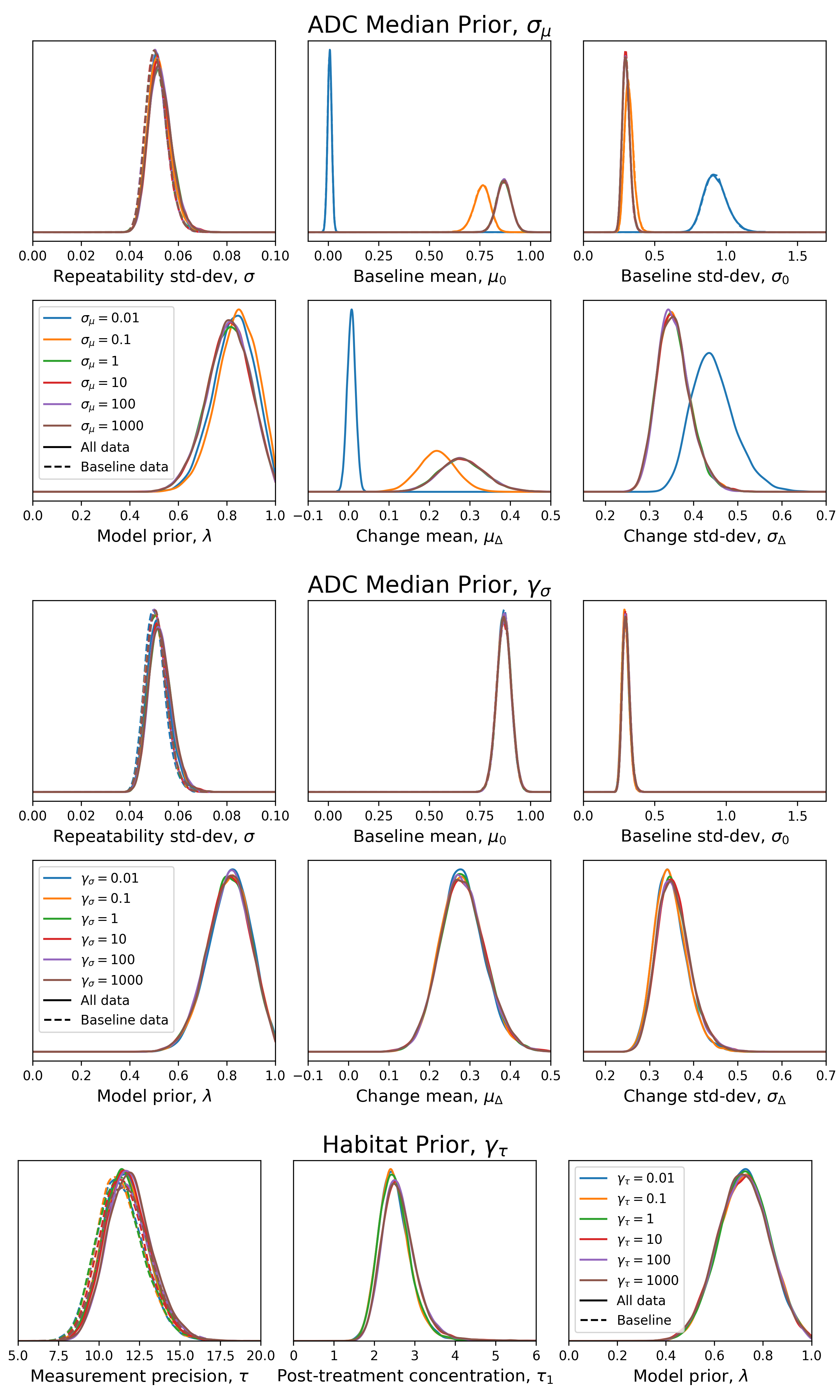}
    \caption{Prior sensitivity analysis for Study 2. The only prior parameter that had significant impact on posterior
    distributions were low values of $\sigma_{\mu}$, which reprsents the width of the zero-mean normal prior for
    $\mu_{0}$ and $\mu_{\Delta}$ in ADC median biomarkers.
    For ADC median biomarkers, other than $\lambda$, all values have units $\times 10^{-3}$ mm$^{2}$/s.}
    \label{fig:sensitivity_analysis_2}
\end{figure}

\newpage
\subsection{Median ADC Stan Code}\label{subsec:median_adc_stan_code}
Below we present the Stan code that was used for inference of the model for median ADC estimates (real values).
\begin{framed}
\scriptsize
\input{adc_model.tex}
\end{framed}

\newpage
\subsection{Dirichlet-Multinomial Stan Code}\label{subsec:habitat_stan_code}
Below we present the Stan code that was used for inference of the model for habitat estimates (Dirichlet-Multinomial
variates).
\begin{framed}
\scriptsize
\input{habitat_model.tex}
\end{framed}

%% file: adc_model.tex
\begin{Verbatim}[commandchars=\\\{\}]
\PY{k+kn}{data} \PY{p}{\PYZob{}}
    \PY{k+kt}{int} \PY{o}{\PYZlt{}}\PY{k}{lower}\PY{p}{=}\PY{l+m+mf}{1}\PY{o}{\PYZgt{}} \PY{n}{Nb}\PY{p}{;}          \PY{c+c1}{// Number of repeat baseline data}
    \PY{k+kt}{int} \PY{o}{\PYZlt{}}\PY{k}{lower}\PY{p}{=}\PY{l+m+mf}{1}\PY{o}{\PYZgt{}} \PY{n}{Np}\PY{p}{;}          \PY{c+c1}{// Number of post\PYZhy{}treatment data}
    \PY{k+kt}{vector}\PY{p}{[}\PY{n}{Nb}\PY{p}{]} \PY{n}{yb1}\PY{p}{;}            \PY{c+c1}{// Baseline measurement 1}
    \PY{k+kt}{vector}\PY{p}{[}\PY{n}{Nb}\PY{p}{]} \PY{n}{yb2}\PY{p}{;}            \PY{c+c1}{// Baseline measurement 2}
    \PY{k+kt}{vector}\PY{p}{[}\PY{n}{Np}\PY{p}{]} \PY{n}{yp1}\PY{p}{;}            \PY{c+c1}{// Post\PYZhy{}treatment measurement 1}
    \PY{k+kt}{vector}\PY{p}{[}\PY{n}{Np}\PY{p}{]} \PY{n}{yp2}\PY{p}{;}            \PY{c+c1}{// Post\PYZhy{}treatment measurement 2}
    \PY{k+kt}{real} \PY{o}{\PYZlt{}}\PY{k}{lower}\PY{p}{=}\PY{l+m+mf}{0}\PY{o}{\PYZgt{}} \PY{n}{sd\PYZus{}prior}\PY{p}{;}   \PY{c+c1}{// Prior width for std\PYZhy{}devs (gamma)}
    \PY{k+kt}{real} \PY{o}{\PYZlt{}}\PY{k}{lower}\PY{p}{=}\PY{l+m+mf}{0}\PY{o}{\PYZgt{}} \PY{n}{mu\PYZus{}prior}\PY{p}{;}   \PY{c+c1}{// Prior width for means}
\PY{p}{\PYZcb{}}

\PY{k+kn}{transformed data} \PY{p}{\PYZob{}}
    \PY{k+kt}{vector}\PY{p}{[}\PY{n}{Nb}\PY{p}{]} \PY{n}{db} \PY{o}{=} \PY{n}{yb2} \PY{o}{\PYZhy{}} \PY{n}{yb1}\PY{p}{;}
    \PY{k+kt}{vector}\PY{p}{[}\PY{n}{Nb}\PY{p}{]} \PY{n}{mb} \PY{o}{=} \PY{p}{(}\PY{n}{yb1} \PY{o}{+} \PY{n}{yb2}\PY{p}{)} \PY{o}{/} \PY{l+m+mf}{2}\PY{p}{;}
    \PY{k+kt}{vector}\PY{p}{[}\PY{n}{Np}\PY{p}{]} \PY{n}{dp} \PY{o}{=} \PY{n}{yp2} \PY{o}{\PYZhy{}} \PY{n}{yp1}\PY{p}{;}
\PY{p}{\PYZcb{}}

\PY{k+kn}{parameters} \PY{p}{\PYZob{}}
    \PY{k+kt}{real} \PY{o}{\PYZlt{}}\PY{k}{lower}\PY{p}{=}\PY{l+m+mf}{0}\PY{o}{\PYZgt{}} \PY{n}{sdr}\PY{p}{;}              \PY{c+c1}{// Measurement error}
    \PY{k+kt}{real} \PY{o}{\PYZlt{}}\PY{k}{lower}\PY{p}{=}\PY{l+m+mf}{0}\PY{o}{\PYZgt{}} \PY{n}{sd0}\PY{p}{;}              \PY{c+c1}{// Std\PYZhy{}dev of baseline}
    \PY{k+kt}{real} \PY{o}{\PYZlt{}}\PY{k}{lower}\PY{p}{=}\PY{l+m+mf}{0}\PY{o}{\PYZgt{}} \PY{n}{sdd}\PY{p}{;}              \PY{c+c1}{// Std\PYZhy{}dev of differences}
    \PY{k+kt}{real} \PY{n}{mu0}\PY{p}{;}                        \PY{c+c1}{// Population baseline value}
    \PY{k+kt}{real} \PY{n}{mud}\PY{p}{;}                        \PY{c+c1}{// Population mean difference}
    \PY{k+kt}{real} \PY{o}{\PYZlt{}}\PY{k}{lower}\PY{p}{=}\PY{l+m+mf}{0}\PY{p}{,}\PY{+w}{ }\PY{k}{upper}\PY{p}{=}\PY{l+m+mf}{1}\PY{o}{\PYZgt{}} \PY{n}{lambda}\PY{p}{;}  \PY{c+c1}{// Mixing proportion}
\PY{p}{\PYZcb{}}

\PY{k+kn}{model} \PY{p}{\PYZob{}}
    \PY{c+c1}{// Priors}
    \PY{n}{sdr} \PY{o}{\PYZti{}}\PY{+w}{ }\PY{n+nb}{cauchy}\PY{p}{(}\PY{l+m+mf}{0}\PY{p}{,} \PY{n}{sd\PYZus{}prior}\PY{p}{)}\PY{p}{;}
    \PY{n}{sd0} \PY{o}{\PYZti{}}\PY{+w}{ }\PY{n+nb}{cauchy}\PY{p}{(}\PY{l+m+mf}{0}\PY{p}{,} \PY{n}{sd\PYZus{}prior}\PY{p}{)}\PY{p}{;}
    \PY{n}{sdd} \PY{o}{\PYZti{}}\PY{+w}{ }\PY{n+nb}{cauchy}\PY{p}{(}\PY{l+m+mf}{0}\PY{p}{,} \PY{n}{sd\PYZus{}prior}\PY{p}{)}\PY{p}{;}
    \PY{n}{mu0} \PY{o}{\PYZti{}}\PY{+w}{ }\PY{n+nb}{normal}\PY{p}{(}\PY{l+m+mf}{0}\PY{p}{,} \PY{n}{mu\PYZus{}prior}\PY{p}{)}\PY{p}{;}
    \PY{n}{mud} \PY{o}{\PYZti{}}\PY{+w}{ }\PY{n+nb}{normal}\PY{p}{(}\PY{l+m+mf}{0}\PY{p}{,} \PY{n}{mu\PYZus{}prior}\PY{p}{)}\PY{p}{;}
    \PY{n}{lambda} \PY{o}{\PYZti{}}\PY{+w}{ }\PY{n+nb}{uniform}\PY{p}{(}\PY{l+m+mf}{0}\PY{p}{,} \PY{l+m+mf}{1}\PY{p}{)}\PY{p}{;}

    \PY{c+c1}{// repeat baseline likelihood}
    \PY{n}{db} \PY{o}{\PYZti{}}\PY{+w}{ }\PY{n+nb}{normal}\PY{p}{(}\PY{l+m+mf}{0}\PY{p}{,} \PY{n+nb}{sqrt}\PY{p}{(}\PY{l+m+mf}{2}\PY{p}{)} \PY{o}{*} \PY{n}{sdr}\PY{p}{)}\PY{p}{;}
    \PY{n}{mb} \PY{o}{\PYZti{}}\PY{+w}{ }\PY{n+nb}{normal}\PY{p}{(}\PY{n}{mu0}\PY{p}{,} \PY{n+nb}{sqrt}\PY{p}{(}\PY{n+nb}{square}\PY{p}{(}\PY{n}{sd0}\PY{p}{)} \PY{o}{+} \PY{n+nb}{square}\PY{p}{(}\PY{n}{sdr}\PY{p}{)}\PY{o}{/}\PY{l+m+mf}{2}\PY{p}{)}\PY{p}{)}\PY{p}{;}

    \PY{c+c1}{// post\PYZhy{}treatment likelihood}
    \PY{k+kt}{vector}\PY{p}{[}\PY{n}{Np}\PY{p}{]} \PY{n}{lp0}\PY{p}{;}
    \PY{k+kt}{vector}\PY{p}{[}\PY{n}{Np}\PY{p}{]} \PY{n}{lp1}\PY{p}{;}
    \PY{k}{for} \PY{p}{(}\PY{n}{n} \PY{k}{in} \PY{l+m+mf}{1}\PY{o}{:}\PY{n}{Np}\PY{p}{)} \PY{p}{\PYZob{}}
        \PY{n}{lp0}\PY{p}{[}\PY{n}{n}\PY{p}{]} \PY{o}{=} \PY{n+nb}{log1m}\PY{p}{(}\PY{n}{lambda}\PY{p}{)} \PY{o}{+}
                 \PY{n+nb}{normal\PYZus{}lpdf}\PY{p}{(}\PY{n}{dp}\PY{p}{[}\PY{n}{n}\PY{p}{]} \PY{p}{|} \PY{l+m+mf}{0}\PY{p}{,} \PY{n+nb}{sqrt}\PY{p}{(}\PY{l+m+mf}{2}\PY{p}{)} \PY{o}{*} \PY{n}{sdr}\PY{p}{)}\PY{p}{;}
        \PY{n}{lp1}\PY{p}{[}\PY{n}{n}\PY{p}{]} \PY{o}{=} \PY{n+nb}{log}\PY{p}{(}\PY{n}{lambda}\PY{p}{)} \PY{o}{+}
                 \PY{n+nb}{normal\PYZus{}lpdf}\PY{p}{(}\PY{n}{dp}\PY{p}{[}\PY{n}{n}\PY{p}{]} \PY{p}{|} \PY{n}{mud}\PY{p}{,}
                             \PY{n+nb}{sqrt}\PY{p}{(}\PY{n+nb}{square}\PY{p}{(}\PY{n}{sdd}\PY{p}{)} \PY{o}{+} \PY{n+nb}{square}\PY{p}{(}\PY{n}{sdr}\PY{p}{)} \PY{o}{*} \PY{l+m+mf}{2}\PY{p}{)}\PY{p}{)}\PY{p}{;}
    \PY{p}{\PYZcb{}}
    \PY{k}{target +=} \PY{n+nb}{sum}\PY{p}{(}\PY{n+nb}{log\PYZus{}sum\PYZus{}exp}\PY{p}{(}\PY{n}{lp0}\PY{p}{,} \PY{n}{lp1}\PY{p}{)}\PY{p}{)}\PY{p}{;}
\PY{p}{\PYZcb{}}

\PY{k+kn}{generated quantities} \PY{p}{\PYZob{}}
    \PY{k+kt}{vector}\PY{p}{[}\PY{n}{Np}\PY{p}{]} \PY{n}{z}\PY{p}{;}           \PY{c+c1}{// Model label samples}
    \PY{k}{for} \PY{p}{(}\PY{n}{n} \PY{k}{in} \PY{l+m+mf}{1}\PY{o}{:}\PY{n}{Np}\PY{p}{)}\PY{p}{\PYZob{}}
        \PY{k+kt}{real} \PY{n}{lp0} \PY{o}{=} \PY{n+nb}{log1m}\PY{p}{(}\PY{n}{lambda}\PY{p}{)} \PY{o}{+}
                   \PY{n+nb}{normal\PYZus{}lpdf}\PY{p}{(}\PY{n}{dp}\PY{p}{[}\PY{n}{n}\PY{p}{]} \PY{p}{|} \PY{l+m+mf}{0}\PY{p}{,} \PY{n+nb}{sqrt}\PY{p}{(}\PY{l+m+mf}{2}\PY{p}{)} \PY{o}{*} \PY{n}{sdr}\PY{p}{)}\PY{p}{;}
        \PY{k+kt}{real} \PY{n}{lp1} \PY{o}{=} \PY{n+nb}{log}\PY{p}{(}\PY{n}{lambda}\PY{p}{)} \PY{o}{+}
                   \PY{n+nb}{normal\PYZus{}lpdf}\PY{p}{(}\PY{n}{dp}\PY{p}{[}\PY{n}{n}\PY{p}{]} \PY{p}{|} \PY{n}{mud}\PY{p}{,}
                               \PY{n+nb}{sqrt}\PY{p}{(}\PY{n+nb}{square}\PY{p}{(}\PY{n}{sdd}\PY{p}{)} \PY{o}{+} \PY{n+nb}{square}\PY{p}{(}\PY{n}{sdr}\PY{p}{)} \PY{o}{*} \PY{l+m+mf}{2}\PY{p}{)}\PY{p}{)}\PY{p}{;}
        \PY{n}{z}\PY{p}{[}\PY{n}{n}\PY{p}{]} \PY{o}{=} \PY{n+nb}{categorical\PYZus{}rng}\PY{p}{(}\PY{p}{[}\PY{n+nb}{exp}\PY{p}{(}\PY{n}{lp0} \PY{o}{\PYZhy{}} \PY{n+nb}{log\PYZus{}sum\PYZus{}exp}\PY{p}{(}\PY{n}{lp0}\PY{p}{,} \PY{n}{lp1}\PY{p}{)}\PY{p}{)}\PY{p}{,}
                                \PY{n+nb}{exp}\PY{p}{(}\PY{n}{lp1} \PY{o}{\PYZhy{}} \PY{n+nb}{log\PYZus{}sum\PYZus{}exp}\PY{p}{(}\PY{n}{lp0}\PY{p}{,} \PY{n}{lp1}\PY{p}{)}\PY{p}{)}\PY{p}{]}\PY{o}{\PYZsq{}}\PY{p}{)}\PY{p}{;}
    \PY{p}{\PYZcb{}}
\PY{p}{\PYZcb{}}
\end{Verbatim}

%% file: habitat_model.tex
\begin{Verbatim}[commandchars=\\\{\}]
\PY{k+kn}{functions} \PY{p}{\PYZob{}}
    \PY{k+kt}{vector} \PY{n}{dir\PYZus{}mult}\PY{p}{(}\PY{k+kt}{array}\PY{p}{[}\PY{p}{,}\PY{p}{]} \PY{k+kt}{int} \PY{n}{y}\PY{p}{,} \PY{k+kt}{vector} \PY{n}{mu}\PY{p}{,} \PY{k+kt}{real} \PY{n}{prec}\PY{p}{)} \PY{p}{\PYZob{}}
        \PY{k+kt}{int} \PY{n}{N} \PY{o}{=} \PY{n+nb}{dims}\PY{p}{(}\PY{n}{y}\PY{p}{)}\PY{p}{[}\PY{l+m+mf}{1}\PY{p}{]}\PY{p}{;}
        \PY{k+kt}{int} \PY{n}{K} \PY{o}{=} \PY{n+nb}{dims}\PY{p}{(}\PY{n}{y}\PY{p}{)}\PY{p}{[}\PY{l+m+mf}{2}\PY{p}{]}\PY{p}{;}
        \PY{k+kt}{vector}\PY{p}{[}\PY{n}{N}\PY{p}{]} \PY{n}{y\PYZus{}sum} \PY{o}{=} \PY{n+nb}{to\PYZus{}vector}\PY{p}{(}\PY{n+nb}{rep\PYZus{}array}\PY{p}{(}\PY{l+m+mf}{0}\PY{p}{,} \PY{n}{N}\PY{p}{)}\PY{p}{)}\PY{p}{;}
        \PY{k+kt}{vector}\PY{p}{[}\PY{n}{N}\PY{p}{]} \PY{n}{p\PYZus{}sum} \PY{o}{=} \PY{n+nb}{to\PYZus{}vector}\PY{p}{(}\PY{n+nb}{rep\PYZus{}array}\PY{p}{(}\PY{l+m+mf}{0}\PY{p}{,} \PY{n}{N}\PY{p}{)}\PY{p}{)}\PY{p}{;}
        \PY{k}{for} \PY{p}{(}\PY{n}{k} \PY{k}{in} \PY{l+m+mf}{1}\PY{o}{:}\PY{n}{K}\PY{p}{)} \PY{p}{\PYZob{}}
            \PY{k+kt}{vector}\PY{p}{[}\PY{n}{N}\PY{p}{]} \PY{n}{y\PYZus{}col} \PY{o}{=} \PY{n+nb}{to\PYZus{}vector}\PY{p}{(}\PY{n}{y}\PY{p}{[}\PY{o}{:}\PY{p}{,} \PY{n}{k}\PY{p}{]}\PY{p}{)}\PY{p}{;}
            \PY{n}{y\PYZus{}sum} \PY{o}{+=} \PY{n}{y\PYZus{}col}\PY{p}{;}
            \PY{n}{p\PYZus{}sum} \PY{o}{+=} \PY{n+nb}{lgamma}\PY{p}{(}\PY{n}{y\PYZus{}col} \PY{o}{+} \PY{n}{mu}\PY{p}{[}\PY{n}{k}\PY{p}{]} \PY{o}{*} \PY{n}{prec}\PY{p}{)}\PY{p}{;}
        \PY{p}{\PYZcb{}}
        \PY{k}{return} \PY{n+nb}{lgamma}\PY{p}{(}\PY{n}{prec}\PY{p}{)} \PY{o}{\PYZhy{}}
               \PY{n+nb}{lgamma}\PY{p}{(}\PY{n}{y\PYZus{}sum} \PY{o}{+} \PY{n}{prec}\PY{p}{)} \PY{o}{+}
               \PY{n}{p\PYZus{}sum} \PY{o}{\PYZhy{}} \PY{n+nb}{sum}\PY{p}{(}\PY{n+nb}{lgamma}\PY{p}{(}\PY{n}{mu} \PY{o}{*} \PY{n}{prec}\PY{p}{)}\PY{p}{)}\PY{p}{;}
    \PY{p}{\PYZcb{}}

    \PY{k+kt}{real} \PY{n}{dir\PYZus{}mult\PYZus{}lpmf}\PY{p}{(}\PY{k+kt}{array}\PY{p}{[}\PY{p}{,}\PY{p}{]} \PY{k+kt}{int} \PY{n}{y}\PY{p}{,} \PY{k+kt}{vector} \PY{n}{mu}\PY{p}{,} \PY{k+kt}{real} \PY{n}{prec}\PY{p}{)} \PY{p}{\PYZob{}}
        \PY{k}{return} \PY{n+nb}{sum}\PY{p}{(}\PY{n}{dir\PYZus{}mult}\PY{p}{(}\PY{n}{y}\PY{p}{,} \PY{n}{mu}\PY{p}{,} \PY{n}{prec}\PY{p}{)}\PY{p}{)}\PY{p}{;}
    \PY{p}{\PYZcb{}}
\PY{p}{\PYZcb{}}

\PY{k+kn}{data} \PY{p}{\PYZob{}}
    \PY{k+kt}{int} \PY{o}{\PYZlt{}}\PY{k}{lower}\PY{p}{=}\PY{l+m+mf}{2}\PY{o}{\PYZgt{}} \PY{n}{K}\PY{p}{;}                 \PY{c+c1}{// Number of components}
    \PY{k+kt}{int} \PY{o}{\PYZlt{}}\PY{k}{lower}\PY{p}{=}\PY{l+m+mf}{1}\PY{o}{\PYZgt{}} \PY{n}{Nb}\PY{p}{;}                \PY{c+c1}{// Number of repeat baseline data}
    \PY{k+kt}{int} \PY{o}{\PYZlt{}}\PY{k}{lower}\PY{p}{=}\PY{l+m+mf}{1}\PY{o}{\PYZgt{}} \PY{n}{Np}\PY{p}{;}                \PY{c+c1}{// Number of post\PYZhy{}treatment data}
    \PY{k+kt}{simplex}\PY{p}{[}\PY{n}{K}\PY{p}{]} \PY{n}{mu0}\PY{p}{;}                  \PY{c+c1}{// Known mean pre\PYZhy{}treatment value}
    \PY{k+kt}{array}\PY{p}{[}\PY{n}{Nb}\PY{p}{,} \PY{n}{K}\PY{p}{]} \PY{k+kt}{int} \PY{o}{\PYZlt{}}\PY{k}{lower}\PY{p}{=}\PY{l+m+mf}{0}\PY{o}{\PYZgt{}} \PY{n}{yb}\PY{p}{;}   \PY{c+c1}{// Repeat baseline measurements}
    \PY{k+kt}{array}\PY{p}{[}\PY{n}{Np}\PY{p}{,} \PY{n}{K}\PY{p}{]} \PY{k+kt}{int} \PY{o}{\PYZlt{}}\PY{k}{lower}\PY{p}{=}\PY{l+m+mf}{0}\PY{o}{\PYZgt{}} \PY{n}{yp}\PY{p}{;}   \PY{c+c1}{// Study measurements}
    \PY{k+kt}{real} \PY{o}{\PYZlt{}}\PY{k}{lower}\PY{p}{=}\PY{l+m+mf}{0}\PY{o}{\PYZgt{}} \PY{n}{prec\PYZus{}prior}\PY{p}{;}       \PY{c+c1}{// Prior width for precision (gamma)}
\PY{p}{\PYZcb{}}

\PY{k+kn}{parameters} \PY{p}{\PYZob{}}
    \PY{k+kt}{real}\PY{o}{\PYZlt{}}\PY{k}{lower}\PY{p}{=}\PY{l+m+mf}{0}\PY{o}{\PYZgt{}} \PY{n}{prec}\PY{p}{;}              \PY{c+c1}{// Precision of the measurement}
    \PY{k+kt}{real}\PY{o}{\PYZlt{}}\PY{k}{lower}\PY{p}{=}\PY{l+m+mf}{0}\PY{o}{\PYZgt{}} \PY{n}{conc}\PY{p}{;}              \PY{c+c1}{// Concentration of the changes}
    \PY{k+kt}{simplex}\PY{p}{[}\PY{n}{K}\PY{p}{]} \PY{n}{mu1}\PY{p}{;}                  \PY{c+c1}{// Population average of changes}
    \PY{k+kt}{real}\PY{o}{\PYZlt{}}\PY{k}{lower}\PY{p}{=}\PY{l+m+mf}{0}\PY{p}{,}\PY{+w}{ }\PY{k}{upper}\PY{p}{=}\PY{l+m+mf}{1}\PY{o}{\PYZgt{}} \PY{n}{lambda}\PY{p}{;}   \PY{c+c1}{// Mixing proportion}
\PY{p}{\PYZcb{}}

\PY{k+kn}{model} \PY{p}{\PYZob{}}
    \PY{c+c1}{// Priors}
    \PY{n}{prec} \PY{o}{\PYZti{}}\PY{+w}{ }\PY{n+nb}{cauchy}\PY{p}{(}\PY{l+m+mf}{0}\PY{p}{,} \PY{n}{prec\PYZus{}prior}\PY{p}{)}\PY{p}{;}
    \PY{n}{conc} \PY{o}{\PYZti{}}\PY{+w}{ }\PY{n+nb}{cauchy}\PY{p}{(}\PY{l+m+mf}{0}\PY{p}{,} \PY{n}{prec\PYZus{}prior}\PY{p}{)}\PY{p}{;}
    \PY{n}{lambda} \PY{o}{\PYZti{}}\PY{+w}{ }\PY{n+nb}{uniform}\PY{p}{(}\PY{l+m+mf}{0}\PY{p}{,} \PY{l+m+mf}{1}\PY{p}{)}\PY{p}{;}
    \PY{n}{mu1} \PY{o}{\PYZti{}}\PY{+w}{ }\PY{n+nb}{dirichlet}\PY{p}{(}\PY{n+nb}{rep\PYZus{}vector}\PY{p}{(}\PY{l+m+mf}{1}\PY{p}{,} \PY{n}{K}\PY{p}{)}\PY{p}{)}\PY{p}{;}

    \PY{c+c1}{// Repeat Baseline Likelihood}
    \PY{n}{yb} \PY{o}{\PYZti{}} \PY{n}{dirichlet\PYZus{}multinomial}\PY{p}{(}\PY{n}{mu0}\PY{p}{,} \PY{n}{prec}\PY{p}{)}\PY{p}{;}

    \PY{c+c1}{// Post\PYZhy{}treatment Likelihood}
    \PY{k+kt}{vector}\PY{p}{[}\PY{n}{Np}\PY{p}{]} \PY{n}{lp0}\PY{p}{;}
    \PY{k+kt}{vector}\PY{p}{[}\PY{n}{Np}\PY{p}{]} \PY{n}{lp1}\PY{p}{;}
    \PY{k}{for} \PY{p}{(}\PY{n}{n} \PY{k}{in} \PY{l+m+mf}{1}\PY{o}{:}\PY{n}{Np}\PY{p}{)} \PY{p}{\PYZob{}}
        \PY{n}{lp0}\PY{p}{[}\PY{n}{n}\PY{p}{]} \PY{o}{=} \PY{n+nb}{log1m}\PY{p}{(}\PY{n}{lambda}\PY{p}{)} \PY{o}{+} \PY{n}{dir\PYZus{}mult\PYZus{}lpmf}\PY{p}{(}\PY{n}{yp}\PY{p}{[}\PY{n}{n}\PY{p}{]} \PY{p}{|} \PY{n}{mu0} \PY{o}{*} \PY{n}{prec}\PY{p}{)}\PY{p}{;}
        \PY{n}{lp1}\PY{p}{[}\PY{n}{n}\PY{p}{]} \PY{o}{=} \PY{n+nb}{log}\PY{p}{(}\PY{n}{lambda}\PY{p}{)} \PY{o}{+} \PY{n}{dir\PYZus{}mult\PYZus{}lpmf}\PY{p}{(}\PY{n}{yp}\PY{p}{[}\PY{n}{n}\PY{p}{]} \PY{p}{|} \PY{n}{mu1} \PY{o}{*} \PY{n}{conc}\PY{p}{)}\PY{p}{;}
    \PY{p}{\PYZcb{}}
    \PY{k}{target +=} \PY{n+nb}{sum}\PY{p}{(}\PY{n+nb}{log\PYZus{}sum\PYZus{}exp}\PY{p}{(}\PY{n}{lp0}\PY{p}{,} \PY{n}{lp1}\PY{p}{)}\PY{p}{)}\PY{p}{;}
\PY{p}{\PYZcb{}}

\PY{k+kn}{generated quantities} \PY{p}{\PYZob{}}
    \PY{k+kt}{vector}\PY{p}{[}\PY{n}{Np}\PY{p}{]} \PY{n}{z}\PY{p}{;}           \PY{c+c1}{// Model label samples}
    \PY{k}{for} \PY{p}{(}\PY{n}{n} \PY{k}{in} \PY{l+m+mf}{1}\PY{o}{:}\PY{n}{Np}\PY{p}{)} \PY{p}{\PYZob{}}
        \PY{k+kt}{real} \PY{n}{lp0} \PY{o}{=} \PY{n+nb}{log1m}\PY{p}{(}\PY{n}{lambda}\PY{p}{)} \PY{o}{+} \PY{n}{dir\PYZus{}mult\PYZus{}lpmf}\PY{p}{(}\PY{n}{yp}\PY{p}{[}\PY{n}{n}\PY{p}{]} \PY{p}{|} \PY{n}{mu0} \PY{o}{*} \PY{n}{prec}\PY{p}{)}\PY{p}{;}
        \PY{k+kt}{real} \PY{n}{lp1} \PY{o}{=} \PY{n+nb}{log}\PY{p}{(}\PY{n}{lambda}\PY{p}{)} \PY{o}{+} \PY{n}{dir\PYZus{}mult\PYZus{}lpmf}\PY{p}{(}\PY{n}{yp}\PY{p}{[}\PY{n}{n}\PY{p}{]} \PY{p}{|} \PY{n}{mu1} \PY{o}{*} \PY{n}{conc}\PY{p}{)}\PY{p}{;}
        \PY{n}{z}\PY{p}{[}\PY{n}{n}\PY{p}{]} \PY{o}{=} \PY{n+nb}{categorical\PYZus{}rng}\PY{p}{(}\PY{p}{[}\PY{n+nb}{exp}\PY{p}{(}\PY{n}{lp0} \PY{o}{\PYZhy{}} \PY{n+nb}{log\PYZus{}sum\PYZus{}exp}\PY{p}{(}\PY{n}{lp0}\PY{p}{,} \PY{n}{lp1}\PY{p}{)}\PY{p}{)}\PY{p}{,}
                                \PY{n+nb}{exp}\PY{p}{(}\PY{n}{lp1} \PY{o}{\PYZhy{}} \PY{n+nb}{log\PYZus{}sum\PYZus{}exp}\PY{p}{(}\PY{n}{lp0}\PY{p}{,} \PY{n}{lp1}\PY{p}{)}\PY{p}{)}\PY{p}{]}\PY{o}{\PYZsq{}}\PY{p}{)}\PY{p}{;}
    \PY{p}{\PYZcb{}}
\PY{p}{\PYZcb{}}
\end{Verbatim}

%% file: repeatability_paper.bbl
\begin{thebibliography}{54}
\providecommand{\natexlab}[1]{#1}
\providecommand{\url}[1]{\texttt{#1}}
\expandafter\ifx\csname urlstyle\endcsname\relax
  \providecommand{\doi}[1]{doi: #1}\else
  \providecommand{\doi}{doi: \begingroup \urlstyle{rm}\Url}\fi

\bibitem[Gillies et~al.(2016)Gillies, Kinahan, and Hricak]{Gillies:2016wk}
Robert~J Gillies, Paul~E Kinahan, and Hedvig Hricak.
\newblock Radiomics: Images are more than pictures, they are data.
\newblock \emph{Radiology}, 278\penalty0 (2):\penalty0 563--77, Feb 2016.
\newblock \doi{10.1148/radiol.2015151169}.

\bibitem[O'Connor et~al.(2017)O'Connor, Aboagye, Adams, Aerts, Barrington,
  Beer, Boellaard, Bohndiek, Brady, Brown, Buckley, Chenevert, Clarke,
  Collette, Cook, deSouza, Dickson, Dive, Evelhoch, Faivre-Finn, Gallagher,
  Gilbert, Gillies, Goh, Griffiths, Groves, Halligan, Harris, Hawkes, Hoekstra,
  Huang, Hutton, Jackson, Jayson, Jones, Koh, Lacombe, Lambin, Lassau, Leach,
  Lee, Leen, Lewis, Liu, Lythgoe, Manoharan, Maxwell, Miles, Morgan, Morris,
  Ng, Padhani, Parker, Partridge, Pathak, Peet, Punwani, Reynolds, Robinson,
  Shankar, Sharma, Soloviev, Stroobants, Sullivan, Taylor, Tofts, Tozer, van
  Herk, Walker-Samuel, Wason, Williams, Workman, Yankeelov, Brindle, McShane,
  Jackson, and Waterton]{OConner_IB}
James P~B O'Connor, Eric~O Aboagye, Judith~E Adams, Hugo J W~L Aerts, Sally~F
  Barrington, Ambros~J Beer, Ronald Boellaard, Sarah~E Bohndiek, Michael Brady,
  Gina Brown, David~L Buckley, Thomas~L Chenevert, Laurence~P Clarke, Sandra
  Collette, Gary~J Cook, Nandita~M deSouza, John~C Dickson, Caroline Dive,
  Jeffrey~L Evelhoch, Corinne Faivre-Finn, Ferdia~A Gallagher, Fiona~J Gilbert,
  Robert~J Gillies, Vicky Goh, John~R Griffiths, Ashley~M Groves, Steve
  Halligan, Adrian~L Harris, David~J Hawkes, Otto~S Hoekstra, Erich~P Huang,
  Brian~F Hutton, Edward~F Jackson, Gordon~C Jayson, Andrew Jones, Dow-Mu Koh,
  Denis Lacombe, Philippe Lambin, Nathalie Lassau, Martin~O Leach, Ting-Yim
  Lee, Edward~L Leen, Jason~S Lewis, Yan Liu, Mark~F Lythgoe, Prakash
  Manoharan, Ross~J Maxwell, Kenneth~A Miles, Bruno Morgan, Steve Morris, Tony
  Ng, Anwar~R Padhani, Geoff J~M Parker, Mike Partridge, Arvind~P Pathak,
  Andrew~C Peet, Shonit Punwani, Andrew~R Reynolds, Simon~P Robinson, Lalitha~K
  Shankar, Ricky~A Sharma, Dmitry Soloviev, Sigrid Stroobants, Daniel~C
  Sullivan, Stuart~A Taylor, Paul~S Tofts, Gillian~M Tozer, Marcel van Herk,
  Simon Walker-Samuel, James Wason, Kaye~J Williams, Paul Workman, Thomas~E
  Yankeelov, Kevin~M Brindle, Lisa~M McShane, Alan Jackson, and John~C
  Waterton.
\newblock Imaging biomarker roadmap for cancer studies.
\newblock \emph{Nat Rev Clin Oncol}, 14\penalty0 (3):\penalty0 169--186, Mar
  2017.
\newblock \doi{10.1038/nrclinonc.2016.162}.

\bibitem[Eisenhauer et~al.(2009)Eisenhauer, Therasse, Bogaerts, Schwartz,
  Sargent, Ford, Dancey, Arbuck, Gwyther, Mooney, Rubinstein, Shankar, Dodd,
  Kaplan, Lacombe, and Verweij]{Eisenhauer2009}
E.~A. Eisenhauer, P.~Therasse, J.~Bogaerts, L.~H. Schwartz, D.~Sargent,
  R.~Ford, J.~Dancey, S.~Arbuck, S.~Gwyther, M.~Mooney, L.~Rubinstein,
  L.~Shankar, L.~Dodd, R.~Kaplan, D.~Lacombe, and J.~Verweij.
\newblock New response evaluation criteria in solid tumours: revised recist
  guideline (version 1.1).
\newblock \emph{European journal of cancer (Oxford, England : 1990)},
  45:\penalty0 228--247, 1 2009.
\newblock ISSN 1879-0852.
\newblock \doi{10.1016/J.EJCA.2008.10.026}.
\newblock URL \url{https://pubmed.ncbi.nlm.nih.gov/19097774/}.

\bibitem[Charles-Edwards and deSouza(2006)]{Charles-Edwards:2006tb}
Elizabeth~M Charles-Edwards and Nandita~M deSouza.
\newblock Diffusion-weighted magnetic resonance imaging and its application to
  cancer.
\newblock \emph{Cancer Imaging}, 6\penalty0 (1):\penalty0 135--43, Sep 2006.
\newblock \doi{10.1102/1470-7330.2006.0021}.

\bibitem[Koh and Collins(2007)]{koh_cancer}
Dow-Mu Koh and David~J Collins.
\newblock Diffusion-weighted mri in the body: applications and challenges in
  oncology.
\newblock \emph{AJR Am J Roentgenol}, 188\penalty0 (6):\penalty0 1622--35, Jun
  2007.
\newblock \doi{10.2214/AJR.06.1403}.

\bibitem[Surov et~al.(2017)Surov, Meyer, and Wienke]{Surov:2017vk}
Alexey Surov, Hans~Jonas Meyer, and Andreas Wienke.
\newblock Correlation between apparent diffusion coefficient (adc) and
  cellularity is different in several tumors: a meta-analysis.
\newblock \emph{Oncotarget}, 8\penalty0 (35):\penalty0 59492--59499, Aug 2017.
\newblock \doi{10.18632/oncotarget.17752}.

\bibitem[Manduca et~al.(2021)Manduca, Bayly, Ehman, Kolipaka, Royston, Sack,
  Sinkus, and Van~Beers]{Manduca:2021wp}
Armando Manduca, Philip~J Bayly, Richard~L Ehman, Arunark Kolipaka, Thomas~J
  Royston, Ingolf Sack, Ralph Sinkus, and Bernard~E Van~Beers.
\newblock Mr elastography: Principles, guidelines, and terminology.
\newblock \emph{Magn Reson Med}, 85\penalty0 (5):\penalty0 2377--2390, May
  2021.
\newblock \doi{10.1002/mrm.28627}.

\bibitem[Sack(2023)]{sack_mre}
Ingolf Sack.
\newblock Magnetic resonance elastography from fundamental soft-tissue
  mechanics to diagnostic imaging.
\newblock \emph{Nature Reviews Physics}, 5\penalty0 (1):\penalty0 25--42, 2023.
\newblock \doi{10.1038/s42254-022-00543-2}.
\newblock URL \url{https://doi.org/10.1038/s42254-022-00543-2}.

\bibitem[Leach et~al.(2012)Leach, Morgan, Tofts, Buckley, Huang, Horsfield,
  Chenevert, Collins, Jackson, Lomas, Whitcher, Clarke, Plummer, Judson, Jones,
  Alonzi, Brunner, Koh, Murphy, Waterton, Parker, Graves, Scheenen, Redpath,
  Orton, Karczmar, Huisman, Barentsz, Padhani, and {Experimental Cancer
  Medicine Centres Imaging Network Steering Committee}]{Leach:2012tt}
M~O Leach, B~Morgan, P~S Tofts, D~L Buckley, W~Huang, M~A Horsfield, T~L
  Chenevert, D~J Collins, A~Jackson, D~Lomas, B~Whitcher, L~Clarke, R~Plummer,
  I~Judson, R~Jones, R~Alonzi, T~Brunner, D~M Koh, P~Murphy, J~C Waterton,
  G~Parker, M~J Graves, T~W~J Scheenen, T~W Redpath, M~Orton, G~Karczmar,
  H~Huisman, J~Barentsz, A~Padhani, and {Experimental Cancer Medicine Centres
  Imaging Network Steering Committee}.
\newblock Imaging vascular function for early stage clinical trials using
  dynamic contrast-enhanced magnetic resonance imaging.
\newblock \emph{Eur Radiol}, 22\penalty0 (7):\penalty0 1451--64, Jul 2012.
\newblock \doi{10.1007/s00330-012-2446-x}.

\bibitem[Sujlana et~al.(2018)Sujlana, Skrok, and Fayad]{DCE-review}
Parvinder Sujlana, Jan Skrok, and Laura~M Fayad.
\newblock Review of dynamic contrast-enhanced mri: Technical aspects and
  applications in the musculoskeletal system.
\newblock \emph{J Magn Reson Imaging}, 47\penalty0 (4):\penalty0 875--890, Apr
  2018.
\newblock ISSN 1522-2586 (Electronic); 1053-1807 (Linking).
\newblock \doi{10.1002/jmri.25810}.

\bibitem[Hoffmann et~al.(2024)Hoffmann, Masthoff, Kunz, Seidensticker, Bobe,
  Gerwing, Berdel, Schliemann, Faber, and Wildgruber]{Hoffmann:2024ws}
Emily Hoffmann, Max Masthoff, Wolfgang~G Kunz, Max Seidensticker, Stefanie
  Bobe, Mirjam Gerwing, Wolfgang~E Berdel, Christoph Schliemann, Cornelius
  Faber, and Moritz Wildgruber.
\newblock Multiparametric mri for characterization of the tumour
  microenvironment.
\newblock \emph{Nat Rev Clin Oncol}, Apr 2024.
\newblock \doi{10.1038/s41571-024-00891-1}.

\bibitem[McGranahan and Swanton(2015)]{McGranahan:2015vr}
Nicholas McGranahan and Charles Swanton.
\newblock Biological and therapeutic impact of intratumor heterogeneity in
  cancer evolution.
\newblock \emph{Cancer Cell}, 27\penalty0 (1):\penalty0 15--26, Jan 2015.
\newblock \doi{10.1016/j.ccell.2014.12.001}.

\bibitem[Maley et~al.(2017)Maley, Aktipis, Graham, Sottoriva, Boddy,
  Janiszewska, Silva, Gerlinger, Yuan, Pienta, Anderson, Gatenby, Swanton,
  Posada, Wu, Schiffman, Hwang, Polyak, Anderson, Brown, Greaves, and
  Shibata]{Maley:2017wa}
Carlo~C Maley, Athena Aktipis, Trevor~A Graham, Andrea Sottoriva, Amy~M Boddy,
  Michalina Janiszewska, Ariosto~S Silva, Marco Gerlinger, Yinyin Yuan,
  Kenneth~J Pienta, Karen~S Anderson, Robert Gatenby, Charles Swanton, David
  Posada, Chung-I Wu, Joshua~D Schiffman, E~Shelley Hwang, Kornelia Polyak,
  Alexander R~A Anderson, Joel~S Brown, Mel Greaves, and Darryl Shibata.
\newblock Classifying the evolutionary and ecological features of neoplasms.
\newblock \emph{Nat Rev Cancer}, 17\penalty0 (10):\penalty0 605--619, Oct 2017.
\newblock \doi{10.1038/nrc.2017.69}.

\bibitem[Amend and Pienta(2015)]{Amend:2015vn}
Sarah~R Amend and Kenneth~J Pienta.
\newblock Ecology meets cancer biology: the cancer swamp promotes the lethal
  cancer phenotype.
\newblock \emph{Oncotarget}, 6\penalty0 (12):\penalty0 9669--78, 2015.
\newblock \doi{10.18632/oncotarget.3430}.

\bibitem[Sala et~al.(2017)Sala, Mema, Himoto, Veeraraghavan, Brenton, Snyder,
  Weigelt, and Vargas]{Sala:2017vi}
E~Sala, E~Mema, Y~Himoto, H~Veeraraghavan, J~D Brenton, A~Snyder, B~Weigelt,
  and H~A Vargas.
\newblock Unravelling tumour heterogeneity using next-generation imaging:
  radiomics, radiogenomics, and habitat imaging.
\newblock \emph{Clin Radiol}, 72\penalty0 (1):\penalty0 3--10, Jan 2017.
\newblock \doi{10.1016/j.crad.2016.09.013}.

\bibitem[Crispin-Ortuzar et~al.(2020)Crispin-Ortuzar, Gehrung, Ursprung, Gill,
  Warren, Beer, Gallagher, Mitchell, Mendichovszky, Priest, Stewart, Sala, and
  Markowetz]{Crispin-Ortuzar:2020us}
Mireia Crispin-Ortuzar, Marcel Gehrung, Stephan Ursprung, Andrew~B Gill, Anne~Y
  Warren, Lucian Beer, Ferdia~A Gallagher, Thomas~J Mitchell, Iosif~A
  Mendichovszky, Andrew~N Priest, Grant~D Stewart, Evis Sala, and Florian
  Markowetz.
\newblock Three-dimensional printed molds for image-guided surgical biopsies:
  An open source computational platform.
\newblock \emph{JCO Clin Cancer Inform}, 4:\penalty0 736--748, Aug 2020.
\newblock \doi{10.1200/CCI.20.00026}.

\bibitem[Beer et~al.(2021)Beer, Martin-Gonzalez, Delgado-Ortet, Reinius, Rundo,
  Woitek, Ursprung, Escudero, Sahin, Funingana, Ang, Jimenez-Linan, Lawton,
  Phadke, Davey, Nguyen, Markowetz, Brenton, Crispin-Ortuzar, Addley, and
  Sala]{Beer:2021uo}
Lucian Beer, Paula Martin-Gonzalez, Maria Delgado-Ortet, Marika Reinius,
  Leonardo Rundo, Ramona Woitek, Stephan Ursprung, Lorena Escudero, Hilal
  Sahin, Ionut-Gabriel Funingana, Joo-Ern Ang, Mercedes Jimenez-Linan, Tristan
  Lawton, Gaurav Phadke, Sally Davey, Nghia~Q Nguyen, Florian Markowetz,
  James~D Brenton, Mireia Crispin-Ortuzar, Helen Addley, and Evis Sala.
\newblock Ultrasound-guided targeted biopsies of ct-based radiomic tumour
  habitats: technical development and initial experience in metastatic ovarian
  cancer.
\newblock \emph{Eur Radiol}, 31\penalty0 (6):\penalty0 3765--3772, Jun 2021.
\newblock \doi{10.1007/s00330-020-07560-8}.

\bibitem[Zormpas-Petridis et~al.(2020)Zormpas-Petridis, Poon, Clarke, Jerome,
  Boult, Blackledge, Carceller, Koers, Barone, Pearson, Moreno, Anderson,
  Sebire, McHugh, Koh, Chesler, Yuan, Robinson, and
  Jamin]{Zormpas-Petridis:2020ti}
Konstantinos Zormpas-Petridis, Evon Poon, Matthew Clarke, Neil~P Jerome,
  Jessica K~R Boult, Matthew~D Blackledge, Fernando Carceller, Alexander Koers,
  Giuseppe Barone, Andrew D~J Pearson, Lucas Moreno, John Anderson, Neil
  Sebire, Kieran McHugh, Dow-Mu Koh, Louis Chesler, Yinyin Yuan, Simon~P
  Robinson, and Yann Jamin.
\newblock Noninvasive mri native t1 mapping detects response to mycn-targeted
  therapies in the th-mycn model of neuroblastoma.
\newblock \emph{Cancer Res}, 80\penalty0 (16):\penalty0 3424--3435, Aug 2020.
\newblock \doi{10.1158/0008-5472.CAN-20-0133}.

\bibitem[Panico et~al.(2023)Panico, Avesani, Zormpas-Petridis, Rundo, Nero, and
  Sala]{Panico:2023wd}
Camilla Panico, Giacomo Avesani, Konstantinos Zormpas-Petridis, Leonardo Rundo,
  Camilla Nero, and Evis Sala.
\newblock Radiomics and radiogenomics of ovarian cancer: Implications for
  treatment monitoring and clinical management.
\newblock \emph{Radiol Clin North Am}, 61\penalty0 (4):\penalty0 749--760, Jul
  2023.
\newblock \doi{10.1016/j.rcl.2023.02.006}.

\bibitem[Obuchowski(2018)]{obuchowski_intero_change}
Nancy~A Obuchowski.
\newblock Interpreting change in quantitative imaging biomarkers.
\newblock \emph{Acad Radiol}, 25\penalty0 (3):\penalty0 372--379, Mar 2018.
\newblock ISSN 1878-4046 (Electronic); 1076-6332 (Print); 1076-6332 (Linking).
\newblock \doi{10.1016/j.acra.2017.09.023}.

\bibitem[Obuchowski and Buckler(2022)]{OBUCHOWSKI2022543}
Nancy~A. Obuchowski and Andrew~J. Buckler.
\newblock Estimating the precision of quantitative imaging biomarkers without
  test-retest studies.
\newblock \emph{Academic Radiology}, 29\penalty0 (4):\penalty0 543--549, 2022.
\newblock ISSN 1076-6332.
\newblock \doi{https://doi.org/10.1016/j.acra.2021.06.009}.
\newblock URL
  \url{https://www.sciencedirect.com/science/article/pii/S1076633221002798}.

\bibitem[Winfield et~al.(2017)Winfield, Tunariu, Rata, Miyazaki, Jerome,
  Germuska, Blackledge, Collins, de~Bono, Yap, deSouza, Doran, Koh, Leach,
  Messiou, and Orton]{Winfield:2017wq}
Jessica~M Winfield, Nina Tunariu, Mihaela Rata, Keiko Miyazaki, Neil~P Jerome,
  Michael Germuska, Matthew~D Blackledge, David~J Collins, Johann~S de~Bono,
  Timothy~A Yap, Nandita~M deSouza, Simon~J Doran, Dow-Mu Koh, Martin~O Leach,
  Christina Messiou, and Matthew~R Orton.
\newblock Extracranial soft-tissue tumors: Repeatability of apparent diffusion
  coefficient estimates from diffusion-weighted mr imaging.
\newblock \emph{Radiology}, 284\penalty0 (1):\penalty0 88--99, Jul 2017.
\newblock \doi{10.1148/radiol.2017161965}.

\bibitem[Winfield et~al.(2019)Winfield, Miah, Strauss, Thway, Collins, deSouza,
  Leach, Morgan, Giles, Moskovic, Hayes, Smith, Zaidi, Henderson, and
  Messiou]{Winfield:2019vy}
Jessica~M Winfield, Aisha~B Miah, Dirk Strauss, Khin Thway, David~J Collins,
  Nandita~M deSouza, Martin~O Leach, Veronica~A Morgan, Sharon~L Giles, Eleanor
  Moskovic, Andrew Hayes, Myles Smith, Shane~H Zaidi, Daniel Henderson, and
  Christina Messiou.
\newblock Utility of multi-parametric quantitative magnetic resonance imaging
  for characterization and radiotherapy response assessment in soft-tissue
  sarcomas and correlation with histopathology.
\newblock \emph{Front Oncol}, 9:\penalty0 280, 2019.
\newblock \doi{10.3389/fonc.2019.00280}.

\bibitem[Bland and Altman(1986)]{MartinBland1986}
J.~Martin Bland and Douglas~G. Altman.
\newblock Statistical methods for assessing agreement between two methods of
  clinical measurement.
\newblock \emph{The Lancet}, 327:\penalty0 307--310, 2 1986.
\newblock ISSN 0140-6736.
\newblock \doi{10.1016/S0140-6736(86)90837-8}.

\bibitem[Scher et~al.(2016)Scher, Morris, Stadler, Higano, Basch, Fizazi,
  Antonarakis, Beer, Carducci, Chi, Corn, Bono, Dreicer, George, Heath,
  Hussain, Kelly, Liu, Logothetis, Nanus, Stein, Rathkopf, Slovin, Ryan,
  Sartor, Small, Smith, Sternberg, Taplin, Wilding, Nelson, Schwartz, Halabi,
  Kantoff, and Armstrong]{Scher2016}
Howard~I. Scher, Michael~J. Morris, Walter~M. Stadler, Celestia Higano, Ethan
  Basch, Karim Fizazi, Emmanuel~S. Antonarakis, Tomasz~M. Beer, Michael~A.
  Carducci, Kim~N. Chi, Paul~G. Corn, Johann S.~De Bono, Robert Dreicer,
  Daniel~J. George, Elisabeth~I. Heath, Maha Hussain, Wm~Kevin Kelly, Glenn
  Liu, Christopher Logothetis, David Nanus, Mark~N. Stein, Dana~E. Rathkopf,
  Susan~F. Slovin, Charles~J. Ryan, Oliver Sartor, Eric~J. Small,
  Matthew~Raymond Smith, Cora~N. Sternberg, Mary~Ellen Taplin, George Wilding,
  Peter~S. Nelson, Lawrence~H. Schwartz, Susan Halabi, Philip~W. Kantoff, and
  Andrew~J. Armstrong.
\newblock Trial design and objectives for castration-resistant prostate cancer:
  Updated recommendations from the prostate cancer clinical trials working
  group 3.
\newblock \emph{Journal of Clinical Oncology}, 34, 2016.
\newblock ISSN 15277755.
\newblock \doi{10.1200/JCO.2015.64.2702}.

\bibitem[Blackledge et~al.(2014)Blackledge, Collins, Tunariu, Orton, Padhani,
  Leach, and Koh]{Blackledge2014}
Matthew~D. Blackledge, David~J. Collins, Nina Tunariu, Matthew~R. Orton,
  Anwar~R. Padhani, Martin~O. Leach, and Dow~Mu Koh.
\newblock Assessment of treatment response by total tumor volume and global
  apparent diffusion coefficient using diffusion-weighted mri in patients with
  metastatic bone disease: A feasibility study.
\newblock \emph{PLoS ONE}, 9, 2014.
\newblock ISSN 19326203.
\newblock \doi{10.1371/journal.pone.0091779}.

\bibitem[Padhani et~al.(2017)Padhani, Lecouvet, Tunariu, Koh, Keyzer, Collins,
  Sala, Schlemmer, Petralia, Vargas, Fanti, Tombal, and de~Bono]{Padhani2017}
Anwar~R. Padhani, Frederic~E. Lecouvet, Nina Tunariu, Dow~Mu Koh, Frederik~De
  Keyzer, David~J. Collins, Evis Sala, Heinz~Peter Schlemmer, Giuseppe
  Petralia, H.~Alberto Vargas, Stefano Fanti, H.~Bertrand Tombal, and Johann
  de~Bono.
\newblock Metastasis reporting and data system for prostate cancer: Practical
  guidelines for acquisition, interpretation, and reporting of whole-body
  magnetic resonance imaging-based evaluations of multiorgan involvement in
  advanced prostate cancer [figure presente.
\newblock \emph{European Urology}, 71:\penalty0 81--92, 2017.
\newblock ISSN 18737560.
\newblock \doi{10.1016/j.eururo.2016.05.033}.

\bibitem[Euser et~al.(2008)Euser, Dekker, and {le Cessie}]{EUSER2008978}
Anne~M. Euser, Friedo~W. Dekker, and Saskia {le Cessie}.
\newblock A practical approach to bland-altman plots and variation coefficients
  for log transformed variables.
\newblock \emph{Journal of Clinical Epidemiology}, 61\penalty0 (10):\penalty0
  978--982, 2008.
\newblock ISSN 0895-4356.
\newblock \doi{https://doi.org/10.1016/j.jclinepi.2007.11.003}.
\newblock URL
  \url{https://www.sciencedirect.com/science/article/pii/S0895435607004131}.

\bibitem[Obuchowski et~al.(2023)Obuchowski, Huang, deSouza, Raunig, Delfino,
  Buckler, Hatt, Wang, Moskowitz, Guimaraes, Giger, Hall, Kinahan, and
  Pennello]{Obuchowski:2023vj}
Nancy~A Obuchowski, Erich Huang, Nandita~M deSouza, David Raunig, Jana Delfino,
  Andrew Buckler, Charles Hatt, Xiaofeng Wang, Chaya Moskowitz, Alexander
  Guimaraes, Maryellen Giger, Timothy~J Hall, Paul Kinahan, and Gene Pennello.
\newblock A framework for evaluating the technical performance of
  multiparameter quantitative imaging biomarkers (mp-qibs).
\newblock \emph{Acad Radiol}, 30\penalty0 (2):\penalty0 147--158, Feb 2023.
\newblock \doi{10.1016/j.acra.2022.08.031}.

\bibitem[Gelman and Rubin(1992)]{gelman1992inference}
Andrew Gelman and Donald~B Rubin.
\newblock Inference from iterative simulation using multiple sequences.
\newblock \emph{Statistical science}, 7\penalty0 (4):\penalty0 457--472, 1992.

\bibitem[Gronau et~al.(2017)Gronau, Sarafoglou, Matzke, Ly, Boehm, Marsman,
  Leslie, Forster, Wagenmakers, and Steingroever]{gronau_bridge}
Quentin~F Gronau, Alexandra Sarafoglou, Dora Matzke, Alexander Ly, Udo Boehm,
  Maarten Marsman, David~S Leslie, Jonathan~J Forster, Eric-Jan Wagenmakers,
  and Helen Steingroever.
\newblock A tutorial on bridge sampling.
\newblock \emph{J Math Psychol}, 81:\penalty0 80--97, Dec 2017.
\newblock \doi{10.1016/j.jmp.2017.09.005}.

\bibitem[Lee and Wagenmakers(2014)]{lee2014bayesian}
Michael~D Lee and Eric-Jan Wagenmakers.
\newblock \emph{Bayesian cognitive modeling: A practical course}.
\newblock Cambridge university press, 2014.

\bibitem[Thrussell et~al.(2022)Thrussell, Winfield, Orton, Miah, Zaidi, Arthur,
  Thway, Strauss, Collins, Koh, Oelfke, Huang, O'Connor, Messiou, and
  Blackledge]{Thrussell:2022vj}
Imogen Thrussell, Jessica~M Winfield, Matthew~R Orton, Aisha~B Miah, Shane~H
  Zaidi, Amani Arthur, Khin Thway, Dirk~C Strauss, David~J Collins, Dow-Mu Koh,
  Uwe Oelfke, Paul~H Huang, James P~B O'Connor, Christina Messiou, and
  Matthew~D Blackledge.
\newblock Radiomic features from diffusion-weighted mri of retroperitoneal
  soft-tissue sarcomas are repeatable and exhibit change after radiotherapy.
\newblock \emph{Front Oncol}, 12:\penalty0 899180, 2022.
\newblock \doi{10.3389/fonc.2022.899180}.

\bibitem[Koh et~al.(2009)Koh, Blackledge, Collins, Padhani, Wallace, Wilton,
  Taylor, Stirling, Sinha, Walicke, Leach, Judson, and Nathan]{Koh:2009wy}
Dow-Mu Koh, Matthew Blackledge, David~J Collins, Anwar~R Padhani, Toni Wallace,
  Benjamin Wilton, N~Jane Taylor, J~James Stirling, Rajesh Sinha, Pat Walicke,
  Martin~O Leach, Ian Judson, and Paul Nathan.
\newblock Reproducibility and changes in the apparent diffusion coefficients of
  solid tumours treated with combretastatin a4 phosphate and bevacizumab in a
  two-centre phase i clinical trial.
\newblock \emph{Eur Radiol}, 19\penalty0 (11):\penalty0 2728--38, Nov 2009.
\newblock \doi{10.1007/s00330-009-1469-4}.

\bibitem[Blackledge et~al.(2017)Blackledge, Koh, Collins, Scurr, Hughes, Leach,
  et~al.]{blackledge2017assessing}
Matthew~David Blackledge, DM~Koh, David~J Collins, Erica Scurr, Julie Hughes,
  M~Leach, et~al.
\newblock Assessing response heterogeneity following radium 223 administration
  using whole body diffusion weighted mri.
\newblock In \emph{International Society for Magnetic Resonance in Medicine
  25th Annual Meeting \& Exhibition}, 2017.

\bibitem[Dalili et~al.(2017)Dalili, Padhani, Grimm, and
  Healthineers]{dalili2017quantitative}
Danoob Dalili, Anwar~R Padhani, Robert Grimm, and Siemens Healthineers.
\newblock Quantitative wb-mri with adc histogram analysis for response
  assessment in diffuse bone disease.
\newblock \emph{Magn. Flash}, 69:\penalty0 32--37, 2017.

\bibitem[Kumar et~al.(2019)Kumar, Carroll, Hartikainen, and
  Mart{\'\i}n]{kumar2019arviz}
Ravin Kumar, Colin Carroll, Ari Hartikainen, and Osvaldo~Antonio Mart{\'\i}n.
\newblock Arviz a unified library for exploratory analysis of bayesian models
  in python.
\newblock 2019.

\bibitem[Rosset et~al.(2004)Rosset, Spadola, and Ratib]{Rosset:2004wy}
Antoine Rosset, Luca Spadola, and Osman Ratib.
\newblock Osirix: an open-source software for navigating in multidimensional
  dicom images.
\newblock \emph{J Digit Imaging}, 17\penalty0 (3):\penalty0 205--16, Sep 2004.
\newblock \doi{10.1007/s10278-004-1014-6}.

\bibitem[Blackledge et~al.(2016)Blackledge, Collins, Koh, and
  Leach]{blackledge2016rapid}
Matthew~D Blackledge, David~J Collins, Dow-Mu Koh, and Martin~O Leach.
\newblock Rapid development of image analysis research tools: bridging the gap
  between researcher and clinician with pyosirix.
\newblock \emph{Computers in biology and medicine}, 69:\penalty0 203--212,
  2016.

\bibitem[Depaoli et~al.(2020)Depaoli, Winter, and
  Visser]{depaoli2020importance}
Sarah Depaoli, Sonja~D Winter, and Marieke Visser.
\newblock The importance of prior sensitivity analysis in bayesian statistics:
  demonstrations using an interactive shiny app.
\newblock \emph{Frontiers in psychology}, 11:\penalty0 608045, 2020.

\bibitem[Donners et~al.(2024)Donners, Candito, Rata, Sharp, Messiou, Koh,
  Tunariu, and Blackledge]{donners2024inter}
Ricardo Donners, Antonio Candito, Mihaela Rata, Adam Sharp, Christina Messiou,
  Dow-Mu Koh, Nina Tunariu, and Matthew~D Blackledge.
\newblock Inter-and intra-patient repeatability of radiomic features from
  multiparametric whole-body mri in patients with metastatic prostate cancer.
\newblock \emph{Cancers}, 16\penalty0 (9):\penalty0 1647, 2024.

\bibitem[Candito et~al.(2024)Candito, Holbrey, Ribeiro, Dragan, Messiou,
  Tunariu, Blackledge, and Koh]{candito2024deep}
Antonio Candito, Richard Holbrey, Ana Ribeiro, Alina Dragan, Christina Messiou,
  Nina Tunariu, Matthew~D Blackledge, and Dow-Mu Koh.
\newblock Deep learning assisted atlas-based delineation of the skeleton from
  whole-body diffusion weighted mri in patients with malignant bone disease.
\newblock \emph{Biomedical Signal Processing and Control}, 92:\penalty0 106099,
  2024.

\bibitem[Ceranka et~al.(2020)Ceranka, Lecouvet, De~Mey, and
  Vandemeulebroucke]{ceranka2020computer}
Jakub Ceranka, Fr{\'e}d{\'e}ric Lecouvet, Johan De~Mey, and Jef
  Vandemeulebroucke.
\newblock Computer-aided detection of focal bone metastases from whole-body
  multi-modal mri.
\newblock In \emph{Medical Imaging 2020: Computer-Aided Diagnosis}, volume
  11314, pages 174--180. SPIE, 2020.

\bibitem[Donners et~al.(2022)Donners, Figueiredo, Tunariu, Blackledge, Koh,
  de~la Maza, Chandran, de~Bono, and Fotiadis]{Donners:2022vn}
Ricardo Donners, Ines Figueiredo, Nina Tunariu, Matthew Blackledge, Dow-Mu Koh,
  Maria de Los Dolores~Fenor de~la Maza, Khobe Chandran, Johann~S de~Bono, and
  Nicos Fotiadis.
\newblock Multiparametric bone mri can improve ct-guided bone biopsy target
  selection in cancer patients and increase diagnostic yield and feasibility of
  next-generation tumour sequencing.
\newblock \emph{Eur Radiol}, 32\penalty0 (7):\penalty0 4647--4656, Jul 2022.
\newblock \doi{10.1007/s00330-022-08536-6}.

\bibitem[Zhu et~al.(2023)Zhu, Shi, Wei, Wang, Shi, Zhang, Yan, Liu, He, Wang,
  Cheng, Duan, Du, Meng, Zhao, Gu, Guo, Ni, He, Guan, and Han]{Zhu:2023ws}
Lianghui Zhu, Huijuan Shi, Huiting Wei, Chengjiang Wang, Shanshan Shi, Fenfen
  Zhang, Renao Yan, Yiqing Liu, Tingting He, Liyuan Wang, Junru Cheng, Hufei
  Duan, Hong Du, Fengjiao Meng, Wenli Zhao, Xia Gu, Linlang Guo, Yingpeng Ni,
  Yonghong He, Tian Guan, and Anjia Han.
\newblock An accurate prediction of the origin for bone metastatic cancer using
  deep learning on digital pathological images.
\newblock \emph{EBioMedicine}, 87:\penalty0 104426, Jan 2023.
\newblock \doi{10.1016/j.ebiom.2022.104426}.

\bibitem[Brady et~al.(2021)Brady, Kriner, Coleman, Morrissey, Roudier, True,
  Gulati, Plymate, Zhou, Birditt, Meredith, Geiss, Hoang, Beechem, and
  Nelson]{Brady:2021tb}
Lauren Brady, Michelle Kriner, Ilsa Coleman, Colm Morrissey, Martine Roudier,
  Lawrence~D True, Roman Gulati, Stephen~R Plymate, Zoey Zhou, Brian Birditt,
  Rhonda Meredith, Gary Geiss, Margaret Hoang, Joseph Beechem, and Peter~S
  Nelson.
\newblock Inter- and intra-tumor heterogeneity of metastatic prostate cancer
  determined by digital spatial gene expression profiling.
\newblock \emph{Nat Commun}, 12\penalty0 (1):\penalty0 1426, Mar 2021.
\newblock \doi{10.1038/s41467-021-21615-4}.

\bibitem[Perez-Lopez et~al.(2017)Perez-Lopez, Mateo, Mossop, Blackledge,
  Collins, Rata, Morgan, Macdonald, Sandhu, Lorente, Rescigno, Zafeiriou,
  Bianchini, Porta, Hall, Leach, de~Bono, Koh, and Tunariu]{Perez-Lopez:2017wk}
Raquel Perez-Lopez, Joaquin Mateo, Helen Mossop, Matthew~D Blackledge, David~J
  Collins, Mihaela Rata, Veronica~A Morgan, Alison Macdonald, Shahneen Sandhu,
  David Lorente, Pasquale Rescigno, Zafeiris Zafeiriou, Diletta Bianchini,
  Nuria Porta, Emma Hall, Martin~O Leach, Johann~S de~Bono, Dow-Mu Koh, and
  Nina Tunariu.
\newblock Diffusion-weighted imaging as a treatment response biomarker for
  evaluating bone metastases in prostate cancer: A pilot study.
\newblock \emph{Radiology}, 283\penalty0 (1):\penalty0 168--177, Apr 2017.
\newblock \doi{10.1148/radiol.2016160646}.

\bibitem[Perez-Lopez et~al.(2016)Perez-Lopez, Lorente, Blackledge, Collins,
  Mateo, Bianchini, Omlin, Zivi, Leach, de~Bono, Koh, and
  Tunariu]{perez_volume}
Raquel Perez-Lopez, David Lorente, Matthew~D Blackledge, David~J Collins,
  Joaquin Mateo, Diletta Bianchini, Aurelius Omlin, Andrea Zivi, Martin~O
  Leach, Johann~S de~Bono, Dow-Mu Koh, and Nina Tunariu.
\newblock Volume of bone metastasis assessed with whole-body diffusion-weighted
  imaging is associated with overall survival in metastatic
  castration-resistant prostate cancer.
\newblock \emph{Radiology}, 280\penalty0 (1):\penalty0 151--60, Jul 2016.
\newblock \doi{10.1148/radiol.2015150799}.

\bibitem[Liu et~al.(2008)Liu, Ting, and Zhou]{liu2008isolation}
Fei~Tony Liu, Kai~Ming Ting, and Zhi-Hua Zhou.
\newblock Isolation forest.
\newblock In \emph{2008 eighth ieee international conference on data mining},
  pages 413--422. IEEE, 2008.

\bibitem[Sch{\"o}lkopf et~al.(2001)Sch{\"o}lkopf, Platt, Shawe-Taylor, Smola,
  and Williamson]{Scholkopf:2001tr}
B~Sch{\"o}lkopf, J~C Platt, J~Shawe-Taylor, A~J Smola, and R~C Williamson.
\newblock Estimating the support of a high-dimensional distribution.
\newblock \emph{Neural Comput}, 13\penalty0 (7):\penalty0 1443--71, Jul 2001.
\newblock \doi{10.1162/089976601750264965}.

\bibitem[Pidhorskyi et~al.(2018)Pidhorskyi, Almohsen, and
  Doretto]{pidhorskyi2018generative}
Stanislav Pidhorskyi, Ranya Almohsen, and Gianfranco Doretto.
\newblock Generative probabilistic novelty detection with adversarial
  autoencoders.
\newblock \emph{Advances in neural information processing systems}, 31, 2018.

\bibitem[Liu et~al.(2022)Liu, Gao, Li, Xiong, Tang, and Chen]{Liu:2022wh}
Xuxing Liu, Weize Gao, Rankang Li, Yu~Xiong, Xiaoqin Tang, and Shanxiong Chen.
\newblock One shot ancient character recognition with siamese similarity
  network.
\newblock \emph{Sci Rep}, 12\penalty0 (1):\penalty0 14820, Sep 2022.
\newblock \doi{10.1038/s41598-022-18986-z}.

\bibitem[Paverd et~al.(2024)Paverd, Zormpas-Petridis, Clayton, Burge, and
  Crispin-Ortuzar]{paverd2024radiology}
Hania Paverd, Konstantinos Zormpas-Petridis, Hannah Clayton, Sarah Burge, and
  Mireia Crispin-Ortuzar.
\newblock Radiology and multi-scale data integration for precision oncology.
\newblock \emph{NPJ Precision Oncology}, 8\penalty0 (1):\penalty0 158, 2024.

\bibitem[Wagenmakers et~al.(2010)Wagenmakers, Lodewyckx, Kuriyal, and
  Grasman]{Wagenmakers:2010tp}
Eric-Jan Wagenmakers, Tom Lodewyckx, Himanshu Kuriyal, and Raoul Grasman.
\newblock Bayesian hypothesis testing for psychologists: a tutorial on the
  savage-dickey method.
\newblock \emph{Cogn Psychol}, 60\penalty0 (3):\penalty0 158--89, May 2010.
\newblock \doi{10.1016/j.cogpsych.2009.12.001}.

\end{thebibliography}
